\documentclass[11pt]{article}
\usepackage{psfrag}
\usepackage{graphicx}
\usepackage{mathrsfs}
\usepackage{amsfonts}
\usepackage{amssymb}
\usepackage{amsmath}

\usepackage{ifthen}
\newcommand{\withnotes}{1}

\newcommand{\notes}[2]{%
\ifthenelse{\equal{\withnotes}{0}}
{}
{\ifthenelse{\equal{#1}{0}}
 {}
 {#2}
}
}

\newcommand{\I}{\ensuremath{\mathtt{i}}}

\newcommand{\MM}{\ensuremath{\mathtt{m}}}
\newcommand{\PP}{\ensuremath{\mathtt{p}}}
\newcommand{\SSS}{\ensuremath{\mathtt{s}}}

 \newtheorem{theorem}{Theorem}
 \newtheorem{lemma}{Lemma}
 \newtheorem{corollary}{Corollary}
 
 \newtheorem{property}{Property}

\def\boxit#1{\vbox{\hrule\hbox{\vrule\kern3pt
  \vbox{\kern3pt#1\kern3pt}\kern3pt\vrule}\hrule}} 
\def\Box{\rule{2mm}{3mm}} 

 \newenvironment{proof}{\trivlist\item[]\emph{Proof}:}%
 {\unskip\nobreak\hskip 1em plus 1fil\nobreak$\Box$
 \parfillskip=0pt%
 \endtrivlist}
{\unskip\nobreak\hskip 1em plus 1fil\nobreak$\Box$
\parfillskip=0pt%
\endtrivlist}
{\unskip\nobreak\hskip 2em plus 1fil\nobreak$\Box$
\parfillskip=0pt%
\endtrivlist}
{\unskip\nobreak\hskip 2em plus 1fil\nobreak$\Box$
\parfillskip=0pt%
\endtrivlist}

\newcommand{\iparagraph}[1]{\vspace*{0ex}\paragraph{\normalfont\textit{#1}}}
\newcommand{\sparagraph}[1]{\vspace*{-2.5ex}\paragraph*{#1}}

\newcommand{\RAP}{{\normalfont\textsc{rap}}}

\newcommand{\suf}[1]{T_{#1}}

\newcommand{\junk}[1]{}

\newcommand{\M}{\ensuremath{\mbox{\tt\char`\$}}}

\newcommand{\collis}[1]{C}

\newcommand{\card}[1]{\!\left\vert#1\right\vert}
\newcommand{\acard}[1]{\!\left\|#1\right\|}

\newcommand{\set}[1]{\!\left\{#1\right\}}
\newcommand{\seq}[1]{\!\left\langle#1\right\rangle}

\newcommand{\lcp}[1]{\mathit{lcp}\!\left({#1}\right)}

\newcommand{\Oh}[1]{O \!\left ( #1 \right )}

\newcommand{\Thetah}[1]{\Theta \!\left ( #1 \right )}

\newcommand{\ceil}[1]{\!\left\lceil #1 \right\rceil}

\newcommand{\subpro}[1]{\mathcal{P}_{#1}}

\newcommand{\slab}[1]{{\ell}_{#1}}

\newcommand{\subprorank}[1]{\mathcal{R}\!\left(\subpro{#1}\right)}

\newcommand{\qlexord}{\triangleleft}
\newcommand{\qequiv}{\simeq}

\newcommand{\subdept}{\sqsubseteq^+}
\newcommand{\arrsub}{\mbox{\texttt{Suff}}}

\newcommand{\sublist}{\mbox{\texttt{SubList}}}

\newcommand{\rankstrut}{\mbox{\texttt{Ranks}}}
\newcommand{\mselmset}{\textsc{MselMset}}

\newcommand{\gsolved}{\texttt{sol\_g}}
\newcommand{\gunsolved}{\texttt{unsolved\_g}}
\newcommand{\gexhausted}{\texttt{exhausted\_g}}
\newcommand{\gundecided}{\texttt{undecided\_g}}
\newcommand{\gjoinable}{\texttt{joinable\_g}}
\newcommand{\gsljoinable}{\texttt{slicejoinable\_g}}

\newcommand{\gpairs}{\texttt{RkSuff}}
\newcommand{\visitlist}{\texttt{Con}}
\newcommand{\visitsum}{\texttt{SCon}}
\newcommand{\visitflag}{\texttt{F}}
\newcommand{\visitlead}{\texttt{Lead}}
\newcommand{\grouping}{\textsc{Group}}
\newcommand{\skipvisit}{\textsc{SkipVisit}}
\newcommand{\guidevisit}{\textsc{GuideVisit}}
\newcommand{\leadvisit}{\textsc{LeadVisit}}
\newcommand{\suffixvisit}{\textsc{SuffixVisit}}
\newcommand{\pruned}{\texttt{Pruned}}
\newcommand{\slicerec}{\textsc{SliceRec}}
\newcommand{\sliceop}{\textsc{Slice}}
\newcommand{\joinop}{\textsc{Join}}
\newcommand{\slicejoinop}{\textsc{SliceJoin}}
\newcommand{\rapbpart}{\mbox{\it{}rap}}
\newcommand{\tend}{\ensuremath{\mbox{\tt\char`\$}}}
\newcommand{\lcpsublist}{\texttt{LcpList}}
\newcommand{\neighsublist}{\texttt{NeighList}}

\newcommand{\termone}{\Gamma}
\newcommand{\termtwo}{\Lambda}
\newcommand{\termthree}{\Phi}

\newcommand{\ssubpro}[1]{\mathcal{P}_{#1}}
\newcommand{\ssubprorank}[1]{\mathcal{R}\!\left(\subpro{#1}\right)}
\newcommand{\soffset}[1]{\mbox{\it{}off}_{#1}}
\newcommand{\sless}[1]{\mbox{\it{}less}_{#1}}

\newcommand{\ppair}[1]{\left\langle{}#1\right\rangle}

\newcommand{\neigh}[1]{\mathcal{N}_{#1}}

\newcommand{\figtextfont}{\scriptsize}

\newenvironment{enumroman}
{

\begin{enumerate}}
{\end{enumerate}}

\newenvironment{enumalpha}
{

\begin{enumerate}}
{\end{enumerate}}

\title{Partial Data Compression and Text Indexing\\via Optimal Suffix Multi-Selection}
\author{
G.\ Franceschini\thanks{Dipartimento di Informatica,
Universit\`a ``La Sapienza'' di Roma, {\tt francesc@di.uniroma1.it}}
\and R. Grossi\thanks{Dipartimento di Informatica,
Universit\`a di Pisa, {\tt grossi@di.unipi.it}}
\and S. Muthukrishnan\thanks{Rutgers University, {\tt muthu@cs.rutgers.edu}}
}

\begin{document}
\date{\today}
 \maketitle
\thispagestyle{empty}

\begin{abstract}
  Consider an input text string $T \equiv T[1,N]$ drawn from an
  unbounded alphabet, so text positions can be accessed using comparisons. 
  We study
  \emph{partial} computation in suffix-based problems for Data Compression and Text
  Indexing such as 
  \begin{itemize}
  \item retrieve any segment of $K \leq N$ consecutive symbols from
    the Burrows-Wheeler transform of $T$, which is at the heart of the
    \texttt{bzip2} family of text compressors, and
  \item retrieve any chunk of $K\leq N$ consecutive entries of the
    Suffix Array or the Suffix Tree, two popular Text Indexing data
    structures for~$T$.
  \end{itemize}

Prior literature would take $O(N\log N)$ comparisons (and time) to solve these problems
by solving the \emph{total} problem of building the entire 
Burrows-Wheeler transform or Text Index for~$T$, and
performing a post-processing to single out the wanted portion. 
The technical challenge is that the suffixes of interest are
potentially of size $O(KN)$ and overlap in intricate ways: we have to
use structural properties of these overlaps to avoid rescanning them
repeatedly.

We introduce a novel adaptive approach to partial computational
problems above, and solve both the partial problems in 
\[
 O(K \log K + N)
\]
comparisons and time, improving the best known running times of $O(N
  \log N)$ for $K=o(N)$. 

These partial-computation problems are intimately
  related since they share a common bottleneck: the
  \emph{suffix multi-selection} problem, which is to output the
  suffixes of rank $r_1, r_2, \ldots, r_K$ under the lexicographic
  order, where $r_1 < r_2 < \cdots < r_K$, $r_i \in [1,N]$.  Special
  cases of this problem are well known: $K=N$ is the suffix sorting
  problem that is the workhorse in Stringology with hundreds of
  applications, and $K=1$ is the recently studied suffix selection.

  We show that suffix multi-selection can be solved in 
  \[
  \Thetah{N\log N - \sum_{j=0}^K \Delta_j \log \Delta_j+N}
  \]
  time and comparisons, where $r_0 = 0$, $r_{K+1} = N+1$, and
  $\Delta_j = r_{j+1} - r_j$ for $0 \leq j \leq K$.  This is
  asymptotically optimal, and also matches the
   bound in \cite{dobkin:munro} for multi-selection on
  atomic elements (not suffixes). Matching the bound known for atomic elements
  for strings is a long running theme and challenge from $70$'s, which 
  we achieve for the suffix multi-selection problem. 
 The partial suffix problems as well as the suffix multi-selection problem have many 
 applications.  
\end{abstract}

\newpage
\setcounter{page}{1}
\pagestyle{plain}

\section{Introduction}
\label{sec:introduction}

Consider an input text string $T \equiv T[1,N]$ and the set $S$ of its
suffixes $\suf{i} \equiv T[i,N]$ ($1 \leq i \leq N$) under the
lexicographic order, where $T[N]$ is an endmarker symbol $\tend$
smaller than any other symbol in $T$.  The alphabet $\Sigma$ from
which the symbols in $T$ are drawn is unbounded: as is standard in 
Stringology,  we assume that any two symbols in $\Sigma$ can only be compared and
this takes $O(1)$ time. Hence, comparing symbolwise any two suffixes
in~$S$ may require $O(N)$ time in the worst case.\footnote{Number of comparisons is asymptotically the same as the time for all the algorithms discussed throughout this paper, and hence we will use them interchangeably.}

We study \emph{partial computation} problems in Data Compression and
Text Indexing for $T$ where we want to quickly get a sense of the
lexicographic distribution of the text suffixes.

\sparagraph{Partial Data Compression. }
The Burrows-Wheeler transform $L$ (a.k.a.~\textsc{bwt}) \cite{Burrows:1994:BSL} of
text string~$T$ is at the heart of the \texttt{bzip2} family of text
compressors, and has many applications. The $r$th symbol in $L$ is
$T[j-1]$ if and only if $\suf{j}$ is the $r$th suffix in the sorting
(except the borderline case $j=1$, for which we take $T[N]$).
There are now efficient methods that convert $T$ to $L$ and 
vice versa, taking $O(N \log N)$ time for unbounded alphabets 
in the worst case.

A partial compression problem is to consider a range $L' \equiv
L[i..i+K-1]$ of $K$ consecutive symbols in~$L$.  Can we compute $L'$
efficiently? More precisely, can we compute $L'$ without computing 
the entire $L$? This is an interesting building block for partial estimation of data
compression ratio.

There is prior work that studies partial compression problems 
where a range $T[i,i+K-1]$ needs to be compressed (by Lempel-Ziv or Burrows-Wheeler or 
one of the other compression methods).  This can be accomplished in 
$O(K\log K + N)$ time using off-the-shelf tools. Instead, what is 
interesting in our question above is that we seek a range in
the \emph{compressed}
string $L$, and computing $L[i..i+K-1]$ amounts to sorting a
set of irregular, arbitrarily scattered suffixes $T_{j_1}, T_{j_2}, \ldots,
T_{j_K}$ for which we do not know  the positions $j_1, j_2,
\ldots, j_K$ {\em a priori}!

\sparagraph{Partial Text Indexing.}
Several text indexes, such as suffix arrays~\cite{Manber93} and suffix
trees~\cite{McC,Wei}, are based on the lexicographic order of the
suffixes in the text $T$ and the longest common prefix (\emph{lcp})
information among them. 
These can be computed in $O(N\log N)$ time using well-known 
 algorithms that exploit properties of suffixes,\footnote{ It
has an $O(N)$ time solution since 70's for constant-size
alphabet~$\Sigma$~\cite{McC,Wei} and, more recently, for (bounded
universe) integer alphabet~$\Sigma$~\cite{F}; otherwise, it uses
$O(N\log |\Sigma|)= O(N\log N)$ comparisons in unbounded
alphabets~$\Sigma$.}
while the rest of the indexes can be easily built in $O(N)$ time.

We define a
\emph{$K$-partial text index} as a range $[i,i+K-1]$ of the index (consecutive
entries of the suffix array or leaves from the suffix tree):
this corresponds to  a sorted set of irregularly scattered
suffixes $T_{j_1}, T_{j_2}, \ldots, T_{j_K}$ for which we do not know
their positions $j_1, j_2, \ldots, j_K$ {\em a priori}, together with the
length of their longest common prefix (\emph{lcp}).  
The technical challenge here is similar to partial compression above, but 
additionally, we need to compute (\emph{lcp}) information. 

We refer the reader to
Sections~\ref{sub:partial-data-compression:full}
and~\ref{sub:partial-text-indexing:full} for a more detailed
discussion.

\sparagraph{Basic questions. }
Can \emph{partial suffix-based} computations like compression and indexing above,
be solved more efficiently than solving the \emph{total} problems? Besides the inherent 
interest in such problems and their structure that will let us parameterize their complexity
in terms of $K \in [1,N]$, the main applied interest is that
these partial problems give us a way to look at spots of a long string and get a sense for
the complexity of data, be it for compressibility or performance of a full-text index.

The central technical challenge is the following.  Given $K$ 
symbols, sorting them is trivial. However, if we have to sort 
$K$ arbitrary suffixes, they  are of size $O(KN)$
in the worst case, and we can not afford to compare them symbolwise. 
In the worst case, it is better to perform a  total  sorting in $O(N \log N)$ time when $K = \Omega(\log N)$, as the
suffixes overlap in 
arbitrary ways and we have to avoid rescanning the symbols repeatedly. 
Can we better this $O(N \log N)$ bound?

\paragraph*{Our results. }
We introduce a new adaptive approach to suffix sorting and order statistics. 
%
\begin{theorem}
  \label{the:bwt2}
  \label{the:bwt2:full}
  \label{the:index2:full}
  Given a text $T$ of length $N$, partial compression and partial text indexing problems can be 
  solved in 
   $O(K \log K + N)$ time.
\end{theorem}

Hence for $K=o(N)$, partial compressing and text indexing problems can 
be solved asymptotically faster than their total counterparts. In particular,
for $K=O(N/\log N)$, these partial problems can be solved in 
$O(N)$ time, which is quite rare in the comparison model. 

We can also provide a bound for an arbitrary choice of $K$ ranks $r_1, r_2, \ldots, r_K$ in the suffix order.
%
\begin{theorem}
  \label{the:bwt1}
  \label{the:bwt1:full} 
  \label{the:index1:full}
 Given a text $T$ of length $N$, partial compression and indexing can be solved using
  \begin{equation}
    \label{eq:multi-bound}
    \Theta\bigg(N\log N - \sum_{j=0}^K \Delta_j \log \Delta_j + N\bigg)
  \end{equation}
  time and comparisons, where $r_0 = 0$, $r_{K+1} = N+1$, and
  $\Delta_j = r_{j+1} - r_j$ for $0 \leq j \leq K$; here, $1 \leq
  r_1 < r_2 < \cdots < r_K \leq N$ are the ranks of the suffixes
  involved in the output.
\end{theorem}

The algorithms behind Theorem~\ref{the:bwt2} use an 
intermediate stage before applying the algorithms behind the more
general Theorem~\ref{the:bwt1}, as otherwise the cost would be $O(K
\log N + N)$ by choosing consecutive ranks $r_1, r_2, \ldots, r_K$
from the given range of values (i.e.\mbox{} $\Delta_j = 1$ for $1 \leq
j < K$).

The above partial-computation problems share a common bottleneck: the
\emph{suffix multi-selection} problem, which is to output the suffixes
of rank $r_1, r_2, \ldots, r_K$ under the lexicographic order, where
$r_1 < r_2 < \cdots < r_K$, $r_i \in [1,N]$.  Special cases of this
problem are well known: $K=N$ is the standard suffix sorting problem,
and $K=1$ is the recently studied suffix selection for which
$O(N)$-time comparison-based solutions are now
known~\cite{Franceschini_et_al09,FranceschiniMuthu07}.
%
We refer the reader to Section~\ref{sec:our-problem:full} for a
more detailed discussion.

\begin{theorem}
  \label{theo:main}
  \label{theo:main:full}
  Given a text $T$ of length $N$, the $K$ text suffixes with ranks
  $r_1<r_2<\cdots<r_K$ (and the \textit{lcp}'s between consecutive
  suffixes) can be found within the bound stated in
  equation~\eqref{eq:multi-bound}. This is optimal.
\end{theorem}


\sparagraph{Related work.}
A long running theme in string matching has been 
matching for suffixes of a string, what one can do for atomic elements.
The earliest suffix tree algorithms of 70's~\cite{McC,Wei} were interesting because they 
sort suffixes of a string over constant-sized alphabet in $O(N)$ time, matching the 
bucket sorting bound for $N$ elements. However, it took lot longer to
match the $O(N)$ time of radix sorting for strings over an integer alphabet in 90's~\cite{F}, and in other computing models~\cite{FarachFerraginaMuthu,KarkkainenSB06}. 
For selection, the classic $O(N)$ time bound for atomic elements from 70's was 
matched only recently for string suffix selection~\cite{FranceschiniMuthu07,Franceschini_et_al09}. 

Similarly, it has been a technical challenge as we show here to match the 
multi-selection bound of atomic elements for string suffixes.
Multi-selection for suffixes is not only interesting for reasons multi-selection 
problem in general is interesting, i.e., 
for statistical analysis of string suffixes, but also because it
emerges naturally as the
computational bottleneck of several problems like the partial compression and 
indexing problems described above, which have no natural counterpart 
in study of atomic elements.

\junk{
Recall that the length of the
\emph{longest common prefix} of any two suffixes $\suf{i}$ and
$\suf{j}$ is defined as the smallest $\ell \geq 0$ such that
$T[i+\ell] \neq T[j+\ell]$. We also can obtain the following result,
which is useful for partial text indexing. 

\begin{corollary}
    \label{cor:main}
    As a byproduct, the length of the longest common prefix
    between the suffix of rank $r_j$ and that of rank $r_{j+1}$ is
    obtained, for $1 \leq j < K$.
\end{corollary}
}

For the sake of completeness, we recall that multi-selection of atomic
elements includes basic problems such as sorting ($K=N$) and selection
($K=1$) as special cases \cite{knuth3}.  Selection algorithms that
work in expected linear time~\cite{Floyd:1975:ETB,Hoare:1961:AF} or
worst-case linear time~\cite{JCSS::BlumFPRT1973} are now in textbooks.
Multi-selection can also model intermediate problems between sorting
and selection: for example, setting equally spaced $r_i$'s, it
corresponds to the quantile problem in Statistics.  
The asymptotically optimal number of comparison (and running time)
is that in equation~\eqref{eq:multi-bound} as proved  in~\cite{dobkin:munro}.\footnote{The bound in~(\ref{eq:multi-bound}) can be
refined by studying the actual constant factors hidden in the
$\Theta$~notation \cite{Kanela}. Several papers have studied other
variations
\cite{HwangT02,Kuba06,MahmoudMS95,MartinezPV04,Panholzer:2003:AMQ,Prodinger:1995:MQH}.}
Our suffix multi-selection algorithm is optimal, matching the lower bound
even for atomic elements!

\paragraph*{Paper organization.}
We first give more details on partial data compression, partial text
indexing and suffix multiselection in
Section~\ref{sec:our-applications:full}, so as to relate the former
two problems to the latter one.  Then, the rest of the paper is
devoted to suffix multi-selection (Theorem~\ref{theo:main:full}) with the following
organization. We give the main ideas in and introduce the main
concepts of subproblems and agglomerates, their data structures and
algorithms, in Section~\ref{sec:def}.  The top-level description of
our multi-selection of suffixes is given in
Section~\ref{sec:algorithm}, and then all the implementation details
are given in Section~\ref{sec:details-optimal-algorithm}. Finally, we
focus on the correctness and the analysis of the costs in
Sections~\ref{sec:correctness-analysis} and~\ref{sec:proofs}

\section{Partial Data Compression,  Partial Text Indexing, and Suffix Multi-Selection}
\label{sec:our-applications:full}

We discuss here some of the new features that make the multi-selection
problem on suffixes interesting and challenging, and focus on its
applications to Data Compression and Text Indexing. To evaluate the
benefits of our findings, we present some examples illustrating which
tasks can be done optimally with known techniques and which cannot. At
the same time we show that using our novel techniques, we can now
perform optimally the latter tasks, which only had suboptimal
algorithms so far.

\begin{table}[t]
\begin{center}
\small
\begin{tabular}{c|cc|cl}
\multicolumn{1}{c}{\sl original}&\multicolumn{2}{c}{\sl
  sorted}&\multicolumn{2}{c}{\sl suffixes}\\
\hline
&&$L$&$i$&$\suf{i}$\\
\hline
\I\SSS\SSS\I\SSS\SSS\I\PP\PP\I\M\MM&\M\MM\I\SSS\SSS\I\SSS\SSS\I\PP\PP&\I&\textbf{12}&\M\\
\PP\I\M\MM\I\SSS\SSS\I\SSS\SSS\I\PP&\I\M\MM\I\SSS\SSS\I\SSS\SSS\I\PP&\PP&11&\I\M\\
\MM\I\SSS\SSS\I\SSS\SSS\I\PP\PP\I\M&\I\PP\PP\I\M\MM\I\SSS\SSS\I\SSS&\SSS&8&\I\PP\PP\I\M\\
\M\MM\I\SSS\SSS\I\SSS\SSS\I\PP\PP\I&\I\SSS\SSS\I\PP\PP\I\M\MM\I\SSS&\SSS&\textbf{5}&\I\SSS\SSS\I\PP\PP\I\M\\
\I\M\MM\I\SSS\SSS\I\SSS\SSS\I\PP\PP&\I\SSS\SSS\I\SSS\SSS\I\PP\PP\I\M&\MM&2&\I\SSS\SSS\I\SSS\SSS\I\PP\PP\I\M\\
\PP\PP\I\M\MM\I\SSS\SSS\I\SSS\SSS\I&\MM\I\SSS\SSS\I\SSS\SSS\I\PP\PP\I&\M&1&\MM\I\SSS\SSS\I\SSS\SSS\I\PP\PP\I\M\\
\I\PP\PP\I\M\MM\I\SSS\SSS\I\SSS\SSS&\PP\I\M\MM\I\SSS\SSS\I\SSS\SSS\I&\PP&\textbf{10}&\PP\I\M\\
\SSS\I\PP\PP\I\M\MM\I\SSS\SSS\I\SSS&\PP\PP\I\M\MM\I\SSS\SSS\I\SSS\SSS&\I&9&\PP\PP\I\M\\
\SSS\SSS\I\PP\PP\I\M\MM\I\SSS\SSS\I&\SSS\I\PP\PP\I\M\MM\I\SSS\SSS\I&\SSS&7&\SSS\I\PP\PP\I\M\\
\I\SSS\SSS\I\PP\PP\I\M\MM\I\SSS\SSS&\SSS\I\SSS\SSS\I\PP\PP\I\M\MM\I&\SSS&\textbf{4}&\SSS\I\SSS\SSS\I\PP\PP\I\M\\
\SSS\I\SSS\SSS\I\PP\PP\I\M\MM\I\SSS&\SSS\SSS\I\PP\PP\I\M\MM\I\SSS\SSS&\I&6&\SSS\SSS\I\PP\PP\I\M\\
\SSS\SSS\I\SSS\SSS\I\PP\PP\I\M\MM\I&\SSS\SSS\I\SSS\SSS\I\PP\PP\I\M\MM&\I&3&\SSS\SSS\I\SSS\SSS\I\PP\PP\I\M\\
\hline
\end{tabular}
\end{center}
\caption{\textsc{bwt} $L$ for the text
  $T = \MM\I\SSS\SSS\I\SSS\SSS\I\PP\PP\I\M$ and its relation with the sorted suffixes.}
\label{tab:bwt:full}
\end{table}

\subsection{Partial Data Compression}
\label{sub:partial-data-compression:full}
The Burrows-Wheeler transform (\textsc{bwt}) \cite{Burrows:1994:BSL}
is at the heart of the \texttt{bzip2} family of text
compressors. Consider all the $N$ circular shifts of the text $T =
\MM\I\SSS\SSS\I\SSS\SSS\I\PP\PP\I\M$ as shown in the first column
({\sl original}) of Table~\ref{tab:bwt:full}.  Perform a lexicographic
sorting of these shifts, as shown in the second column ({\sl sorted}):
if we single out the last symbol from each of the circular shifts in
this order, we obtain a sequence $L$ of $N$ symbols that is called the
\textsc{bwt} of $T$. Interestingly, not only we can recover $T$ from
$L$ alone, but typically $L$ is more compressible than $T$ itself
using 0th-order compressors (e.g.~\cite{Manzini01}). 
Its relation with suffix sorting is well known: the $r$th symbol in $L$ is
$T[j-1]$ if and only if $\suf{j}$ is the $r$th suffix in the sorting
(except the borderline case $j=1$, for which we take $T[N]$), as shown
in the third column (\textsl{suffixes}).

Data compression ratio can be partially estimated by choosing a
suitable sample $L'$ of $L$ for statistical purposes. There are
several ways to make this choice, some are easy and some others are
not, as we show next. It is easy to build $L'$ if we take every other
$q$th suffix in $T$. For example, $q=3$ gives $L' = \M\PP\SSS\SSS$
since we pick $\suf{1}, \suf{4}, \suf{7}, \suf{10}$ and then perform
their lexicographic sorting, namely, $\suf{1} < \suf{10} < \suf{7} <
\suf{4}$. The latter is a simple variant of the standard suffix
sorting and takes $O(N \log (N/q))$ time: the text $T$
is conceptually partitioned as a sequence of $N/q$ macro-symbols,
where each macro-symbol is a segment of $q$ actual symbols in $T$
(could be less in the last one). Suffix sorting requires $O(N/q \log
(N/q))$ macro-comparison, each involving $q$ symbolwise comparisons,
thus giving the above bound.

What if $L'$ is chosen by taking every other $q$th symbol
\emph{directly} in $L$?  Contrarily to the previous situation, here we
guarantee a uniform sampling from $L$. For example, $q=3$ gives $L' =
\I\SSS\PP\SSS$, which corresponds to selecting the suffixes $\suf{12}
< \suf{5} < \suf{10} < \suf{4}$ as shown in boldface in
Table~\ref{tab:bwt:full}. Here comes a crucial observation: even though we
sample from $L$ with regularity, the starting positions of the chosen
suffixes from the input text $T$ form an \emph{irregular} pattern and
are difficult to predict without suffix sorting. We are not aware of
any better approach other than performing a full execution of suffix
sorting and, then, making a post-processing to single out every other
$q$th sorted item.  This yields a suboptimal cost of $O(N \log N)$
which should be compared to the $O(N \log (N/q)) = O(N \log K)$ cost
in equation~\eqref{eq:multi-bound} by setting $\Delta_j = q$ for $j<K$
and $\Delta_K \leq q+1$.
In general, specifying ranks $r_1, r_2, \ldots, r_K$ gives the sample
made up of the $r_1$th, $r_2$th,~\dots, $r_K$th symbols of $L$: using
the algorithm giving the cost in~\eqref{eq:multi-bound}, we can look at
specific
parts of the compressed string, without paying the full suffix sorting
cost. We refer the reader to Theorem~\ref{the:bwt1:full}.


An intriguing situation arises when we consider just a segment
of $K$ consecutive symbols. Let us first consider a text segment
$T[a,b]$, where $K = b-a+1$ and $1 \leq a \leq b \leq N$. It takes
$O(K \log K)$ time to perform a suffix sorting and compute the
\textsc{bwt} of that segment alone; or $O(K \log K + N)$ time to sort
the consecutive suffixes $\suf{a}, \suf{a+1}, \ldots, \suf{b}$ whose
starting positions lie in that segment, and then find their induced
symbols inside $L$. Once again, these are simple variations of suffix
sorting.

What if we want to compute only a segment $L[a,b]$ of $K = b-a+1$
consecutive symbols instead of the whole $L$?  For example, $L[3,5] =
\SSS\SSS\MM$ corresponds to suffixes $\suf{8} < \suf{5} < \suf{2}$ in
Table~\ref{tab:bwt:full}.  In general, the starting positions of these
suffixes form an irregular pattern and, as far as we know, the best
that we can do is performing a \emph{full} suffix sorting in $O(N \log
N)$ time. Instead, setting $r_1=a, r_2=a+1, \ldots, r_K=b$, we obtain
that $\Delta_0=a$, $\Delta_K=N+1-b$, and $\Delta_j = 1$ for $1 < j <
K$. Hence, the cost implied by equation~\eqref{eq:multi-bound} is $O(K \log
N
+ N)$, and the computed longest common prefixes in
Theorem~\ref{theo:main:full} will also provide the contexts for estimating
the empirical entropy of~$L$ restricted to the segment $L[a,b]$. Even
better, we can obtain $O(K \log K + N)$ using the same algorithm
behind equation~\eqref{eq:multi-bound} and an analysis focussed on this
special case (i.e.\mbox{} the wanted ranks form an interval of
consecutive values, see Lemma~\ref{lem:multisel:contiguous:full}).  When $K
= o(N)$, this compares favorably with the suboptimal $O(N \log N)$
cost of building~$L$ explicitly. We refer the reader to Theorem~\ref{the:bwt2:full}.
%

\begin{figure}[t]
  \begin{center}
    \begin{psfrags}
      \psfrag{\$}{$\tend$}
      \psfrag{I}{\I}
      \psfrag{M}{\MM}
      \psfrag{P}{\PP}
      \psfrag{S}{\SSS}
      \includegraphics*[scale=1]{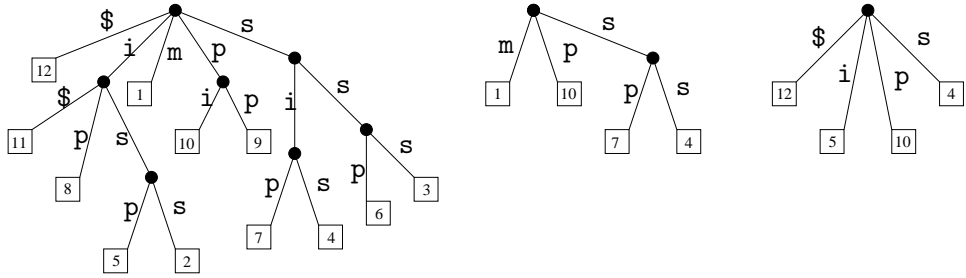}
    \end{psfrags}
  \end{center}
  \caption{A suffix tree and two ways to sample it. Only the first
    symbol is shown on each labelled edge.}
  \label{fig:sampled:full}
\end{figure}

\subsection{Partial Text Indexing}
\label{sub:partial-text-indexing:full}
Several text indexes, such as suffix arrays~\cite{Manber93} and
suffix trees~\cite{McC,Wei}, are based on the lexicographic order of
the suffixes in the text $T$ and the longest common prefix (\emph{lcp})
information among them. Ultimately, the suffix sorting and the \emph{lcp}
information constitute the kernel upon which the rest of the index can
be easily built in $O(N)$ time. An instance of suffix tree for the text $T =
\MM\I\SSS\SSS\I\SSS\SSS\I\PP\PP\I\M$ is shown in
Figure~\ref{fig:sampled:full}(left), and consists of a compacted trie
storing all the suffixes of $T$. 

Based on the above observation, we define a \emph{$K$-partial text
  index} as a subset of $K$ suffixes plus their \emph{lcp}
information. Having this, we can build the suffix array or the suffix
tree restricted to these $K$ suffixes in $O(K)$ time. Hence, the
problem of building a $K$-partial text index tantamounts to performing
a sorting of these $K$ suffixes and finding their \emph{lcp} information. We
discuss two ways of choosing these $K$ suffixes.

One possibility is sampling every other $q$th \emph{suffix} in $T$, as
shown in Figure~\ref{fig:sampled:full}(center) with $q=3$. This is the {\em
  sampled} suffix tree introduced in \cite{Esko}, and its construction
is a simple variant of the standard one and takes $O(N \log (N/q)) =
O(N \log K)$ time: as previously mentioned, the text $T$ is
conceptually partitioned as a sequence of $K=N/q$ macro-symbols.

Another possibility is sampling every other $q$th \emph{leaf} directly
from the suffix tree, as shown in Figure~\ref{fig:sampled:full}(right) with
$q=3$. However, with the current techniques we have a suboptimal
solution: build first the whole suffix tree in $O(N \log N)$ time and
then perform a post-processing to select the wanted leaves and their
ancestors (removing the possible unary nodes thus created). Using the
algorithm behind equation~\eqref{eq:multi-bound} and the related longest
common prefix information as a byproduct, we are able to build this
$K$-partial index in $O(N \log (N/q)) = O(N \log K)$ time by fixing
$\Delta_j = q$ for $j<K$ and $\Delta_K \leq q+1$. Contrarily to the
sampled suffix tree above, here the starting positions of the chosen
suffixes form an \emph{irregular} pattern even though we sample from
the suffix tree with regularity.

In general, given a text $T$ and its ranks $r_1, r_2, \ldots, r_K$, we
want to build the $K$-partial text index for the suffixes having those
ranks. For example, when employing a suffix tree, the $K$-partial text
index gives the subtrie made up of the $r_1$th, $r_2$th,~\dots,
$r_K$th leaves of the suffix tree for $T$. But we do not want to build
that full suffix tree explictly in $O(N \log N)$ time. Using the
algorithm behind equation~\eqref{eq:multi-bound}, we can attain this goal.
We refer the reader to Theorem~\ref{the:index1:full}.


A somewhat surprising situation arises when considering just a
segment of $K$ consecutive symbols. Consider first the $K$-partial
index build on the text segment $T[a,b]$, where $K = b-a+1$ and $1
\leq a \leq b \leq N$. It takes $O(K \log K)$ time to build the suffix
array or the suffix tree for that segment alone, or $O(K \log K + N)$
time to build them for the consecutive suffixes $\suf{a}, \suf{a+1},
\ldots, \suf{b}$ whose starting positions lie in that segment. Once
again, these are simple variations of known algorithms.

This is not the case when computing just a segment of $K$ consecutive
entries in the suffix array, or just a chunk of $K$ consecutive leaves
from the suffix tree. Clearly, we do not want the suboptimal solution
that builds the entire suffix array or suffix tree in $O(N \log N)$
time. Ours is a special case of $K$-partial text index, since the
wanted suffixes have \emph{consecutive ranks} $r_1=a, r_2=a+1, \ldots,
r_K=b$.  The $O(K \log N + N)$ cost implied by
equation~\eqref{eq:multi-bound}
can be refined (using Lemma~\ref{lem:multisel:contiguous:full}) so that we
obtain $O(K \log K + N)$, comparing favorably with the suboptimal cost
of building the suffix array/tree entirely. In general, the task of
computing only a chunk of $K$ consective entries from the suffix array
or the suffix tree falls within the following result.
We refer the reader to Theorem~\ref{the:index2:full}.


\subsection{Suffix Multi-Selection Problem}
\label{sec:our-problem:full}

Given a set $S$ of $N$ elements from
a total order $<$, the rank of $x \in S$ is $r \in [1,N]$ if $x$ is
the $r$th smallest in~$S$, namely, $r = 1+\card{\set{y \in S \vert y <
    x}}.$ For integers $r_1 < r_2 < \cdots < r_K$, where each $r_i \in
[1,N]$, the \emph{multi-selection} problem is to select the elements
of rank $r_1, r_2, \ldots, r_K$ from $S$.

Multi-selection includes basic problems such as sorting and selection as special
cases. When $K=N$, it finds the ranks for all the
elements, making it straightforward to arrange them in sorted order
\cite{knuth3}.  When $K=1$, it corresponds to the standard selection
problem: given an integer $r \in [1,N]$, the goal is to return
the element of rank $r$ in~$S$.
Selection algorithms that work in expected linear
time~\cite{Floyd:1975:ETB,Hoare:1961:AF} or worst-case linear
time~\cite{JCSS::BlumFPRT1973} are now in textbooks.  It can also
model intermediate problems bewteen sorting and selection: for
example, setting equally spaced $r_i$'s, it corresponds to the
quantile problem in Statistics.  Multi-selection can thus arise in
applications for partitioning the input, say for a recursive approach.

%
%
%
%
%

To the best of our knowledge the first algorithm for multi-selection was 
given in \cite{dobkin:munro}, thus establishing that the \emph{asymptotically}
optimal number of comparison (and running time) is that in
equation~\eqref{eq:multi-bound}, 
where $r_0 = 0$ and $r_{K+1} = N+1$ and $\Delta_j \equiv r_{j+1} - r_j$ for $0 \leq j \leq K$. 
The formula in~(\ref{eq:multi-bound}) can be intuitively read as
follows. Find the $\Delta_0$ smallest elements, then the $\Delta_1$ next
smallest elements, and so on, up to the last $\Delta_K$ ones. The
resulting arrangement is almost sorted, and can be fully sorted by
ordering each individual group of $\Delta_j$ elements independently in
$\Theta(\Delta_j \log \Delta_j)$ time.  Hence, take the total sorting
cost of $\Theta(N \log N)$, subtract the cost of sorting each group,
i.e.\mbox{} $\Theta(\sum_{j=0}^K \Delta_j \log \Delta_j)$, and add
$\Theta(N)$ to read all the elements as a baseline.
Note that rewriting~(\ref{eq:multi-bound}) as $\Theta(N \sum_{j=0}^K
(\Delta_j/N) \log (N/\Delta_j) + N)$, where $\sum_{j=0}^K \Delta_j =
N+1$, we can reformulate the bound in~(\ref{eq:multi-bound}) as
$\Theta(N (H_0 + 1))$ where $H_0 = - \sum_{j=0}^K p_j \log_2 p_j$ is
the empirical 0th-order entropy where $p_j = \Delta_j/N$ is the
empirical probability of having the $j$th group of size $\Delta_j$.

The asymptotical optimality of~(\ref{eq:multi-bound}) can be further
refined by studying the actual constant factors hidden in the $\Theta$
notation. Any comparison-based algorithm must perform at
least $B = N\log N - \sum_{j=0}^K \Delta_j \log \Delta_j - O(N)$
comparisons to solve multi-selection,
and the algorithm in~\cite{Kanela} is nearly optimal in this sense. It
attains $B + O(N)$ expected comparisons and $B+ o(B) +O(N)$
comparisons in the worst case, taking $O(B+N)$ running time. For the
interested reader, other papers have studied further features 
in \cite{HwangT02,Kuba06,MahmoudMS95,MartinezPV04,Panholzer:2003:AMQ,Prodinger:1995:MQH}.

We focus on the asymptotically optimality of
the bound in~(\ref{eq:multi-bound}) and call \emph{optimal} an algorithm
that
asymptotically meets that bound, namely, $\Theta(B+N)$ in the worst
case. Unless specified, the running time is always proportional to the
number of pairwise comparisons between elements.

\junk{ Recently, 
  e.g.\mbox{} the exact formula for the expected number of
  comparisons.  A new algorithm has been presented in \cite{Kanela}.
  Multi-selection is a partial version of the standard problem of
  selecting quantiles in statistical analysis. Also, it is employed in
  parallel sorting \cite{Blelloch:1996:CSA}.  }

\junk{
Interestingly, one of the peculiar characteristics of the above
solutions for selection and order statistics is that they are all
based on a purely \emph{distributive approach}. That is, they single
out some of the elements and use these to partition and distribute all
the elements into subsets. These subsets are recursively
processed in the same way, until subsets with zero or one element are
obtained.  Classic examples of a distributive approach in sorting are
quicksort (and its optimal variant based on linear time selection),
top-down radix sorting, string sorting algorithms based on multiset
sorting \cite{MS}.}

In this paper, we study the analog of the multi-selection problem in
Stringology.  Consider an input text string $T \equiv T[1,N]$ and the
set $S$ of its suffixes $\suf{i} \equiv T[i,N]$ ($1 \leq i \leq N$)
under the lexicographic order.  Let $T[N]$ be an endmarker symbol,
denoted by $\tend$, which is smaller than any other symbol in $T$.  The
alphabet $\Sigma$ from which the symbols in $T$ are drawn is unbounded
and the comparison model is adopted: any two symbols in $\Sigma$ can
only be compared and this takes constant time. Hence, comparing
symbolwise any two suffixes in~$S$ may requires $O(N)$ time in the
worst case.  Given ranks $r_1 < r_2 < \cdots < r_K$, the \emph{suffix
  multi-selection} problem is to output the suffixes of rank $r_1,
r_2, \ldots, r_K$ in $S$.
Our main contribution is that we extend the
asymptotically optimal bound in~(\ref{eq:multi-bound}) to the suffixes
in $S$. We refer the reader to Theorem~\ref{theo:main:full}.


Recall that the length of the \emph{longest common prefix} of any two
suffixes $\suf{i}$ and $\suf{j}$ is defined as the smallest $\ell \geq
0$ such that $T[i+\ell] \neq T[j+\ell]$. We also can obtain this info
as a byproduct, which is useful for the problems mentioned in
Sections~\ref{sub:partial-data-compression:full}
and~\ref{sub:partial-text-indexing:full}.


\section{Concepts, Definitions, and Main Ideas}
\label{sec:def}

For the given string $T \equiv T[1,N]$, let $T[N]$ be an endmarker
symbol, denoted by $\tend$, that is smaller than any symbol of the
alphabet and does not appear elsewhere in $T$.  Let $\suf{i} \equiv
T[i,N]$ be the $i$-th suffix of $T$, and $S =\set{\suf{i} \mid 1 \leq
  i \leq N}$ be the set of all these suffixes. Given the set
$\mathcal{R}$ of ranks $r_1 <r_2< \cdots <r_K$, we want to select the
suffixes $\suf{i_1} < \suf{i_2} < \cdots < \suf{i_K}$ such that
$\suf{i_j}$ has rank $r_j$ in~$S$, for $1 \leq j \leq K$.  Since we
already know that $\suf{N}$ is the smallest suffix of $T$, we can
assume wlog that $r_1>1$. Also, we use $<$ and $\le$ to denote string
comparison according to the lexicographical order.  For any two sets
$X, Y$, notation
$X < Y$ indicates that $x<y$ for
any pair $x \in X, y \in Y$.

Our main goal is to prove Theorem~\ref{theo:main}, as this is the
major obstacle when proving Theorems~\ref{the:bwt2}--~\ref{the:bwt1}. The tricks of the
trade all rely on the following
``golden rule'' on the \emph{lcp} information for the suffixes:
For any integer $d>0$, if $\lcp{\suf{i}, \suf{j}} > \ell$ for
  some integer $\ell \geq d$, then $\lcp{\suf{i+d}, \suf{j+d}} \geq
  \ell-d$. Thus, if $\suf{i}$ and $\suf{j}$ have been compared, the
direct comparison of the first $\ell-d$ symbols of $\suf{i+d}$ and
$\suf{j+d}$ can be avoided. Conversely, if $\suf{i+d}$ and $\suf{j+d}$
have been compared, the comparison of all but the first $d$ symbols in
$\suf{i}$ and $\suf{j}$ can be avoided.

Unfortunately, we cannot always rely on the golden rule here.
When comparing $\suf{i}$ and~$\suf{j}$, we do not yet know whether or
not both $\suf{i+d}$ and $\suf{j+d}$ will have to be compared
(directly or indirectly by transitivity) for a choice of $d > 0$, or
vice versa: simply put, we cannot predict at each stage of the
computation whether the comparison between $\suf{i+d}$ and $\suf{j+d}$
will occur or not in the future.  

\subsection{Subproblems}
\label{subsec:subproblems}

At the beginning, all the suffixes form a single problem $S \equiv
\set{\suf{1}, \suf{2}, \ldots, \suf{n}}$. At the end, we want to
obtain a partition of $S$ into \emph{subproblems}, namely, $S = S_1
\cup \set{\suf{i_1}} \cup S_2 \cup \set{\suf{i_2}} \cup \cdots \cup
S_K \cup \set{\suf{i_K}} \cup S_{K+1}$, such that $S_j$ contains all
the suffixes in $S$ of rank $r$ with $r_{j-1} < r < r_j$ for $1 \leq j
\leq K+1$. Although each $S_j$ is not internally sorted, still $S_1
< \set{\suf{i_1}} < S_2 < \set{\suf{i_2}} < \cdots < S_K <
\set{\suf{i_K}} < S_{K+1}$.

At an arbitrary stage of the computation, the suffixes in $S$ and the
ranks in $\mathcal{R}$ are partitioned amongst the
\emph{subproblems}. Let us call $\subpro{1} < \subpro{2} < \cdots <
\subpro{z}$ the subproblems in the current stage, where $S =
\subpro{1} \cup \subpro{2} \cup \cdots \cup \subpro{z}$, and let
$\subprorank{i}$ be the set of ranks associated with subproblem
$\subpro{i}$, for $1 \leq i \leq z$. Namely, $r\in \subprorank{i}$ iff
$\slab{i} < r \leq \slab{i}+\card{\subpro{i}}$, where $\slab{i} =
\sum_{j < i} \,\card{\subpro{j}}$ is the number of
lexicographically smaller suffixes.

\sparagraph{Status.}
A subproblem $\subpro{i}$ is $\emph{solved}$ if
$\card{\subpro{i}}=\card{\subprorank{i}}=1$, and $\emph{unsolved}$ if
$\card{\subpro{i}}>1$ and $\card{\subprorank{i}}\ge 1$.  A subproblem
$\subpro{i}$ and its suffixes are \emph{exhausted} if
$\card{\subprorank{i}}=0$.  A subproblem $\subpro{i}$ and its suffixes
are \emph{degenerate} if each of these suffixes share an \emph{empty}
prefix with any of the wanted suffixes $\suf{i_1} < \suf{i_2} < \cdots <
\suf{i_K}$. (Note that $(i)$ a degenerate subproblem is also exhausted
and $(ii)$ $\suf{N}$ is always degenerate since we assume $r_{1}>1$.)

A subproblem $\subpro{i}$ is never merged with others and can
  only be refined into smaller subproblems. A partition of
$\subpro{i}$ into $\subpro{i_{1}},\ldots,\subpro{i_{p}}$ is called
a \emph{refinement} if for any two $\subpro{i_{j}},\subpro{i_{j'}}$ either
$\subpro{i_{j}}<\subpro{i_{j'}}$ or $\subpro{i_{j'}}<\subpro{i_{j}}$:
note that the refinement is a stronger notion than the partition, since it
also takes into account of the lexicographic order among the suffixes.

During the computation a subproblem can be either \emph{active} or
\emph{inactive}. If $\card{\subprorank{i}}\ge 1$ then $\subpro{i}$ is
active. An \emph{active} subproblem is subjected to the \emph{refinement},
whereas this does not hold anymore once it becomes \emph{inactive}. If
$\subpro{i}$ is inactive then $\card{\subprorank{i}}=0$ and thus it is
exhausted. (The latter is a necessary condition: not all the exhausted
subproblems are inactive.)  Once a subproblem becomes inactive it
will stay inactive until the end. Degenerate subproblems are
inactive since the start. We give a complete characterization in
Section~\ref{subsec:agglomerates}.

\sparagraph{Integer labels and neighborhood.}
Ideally, we would like to maintain the integer label $\slab{i} =
\sum_{j < i} \,\card{\subpro{j}}$ for each subproblem $\subpro{i}$
during the refining process. Using the integer labels, we can define a
total order relation $\qlexord$ and an equivalence relation $\qequiv$
on the suffixes of $T$.  Consider any two suffixes $\suf{i'},\suf{j'}$
and their subproblems' labels $\slab{i},\slab{j}$.  We have that
$\suf{i'}\qlexord{}\suf{j'}$ iff \emph{(1)} $T[i']\le T[j']$ when both
$\suf{i'}$ and $\suf{j'}$ are degenerate or \emph{(2)} $\slab{i}\le
\slab{j}$ otherwise.  Similarly, $\suf{i'}\qequiv\suf{j'}$ iff \emph{(1)}
$T[i']=T[j']$ when both $\suf{i}$ and $\suf{j}$ are degenerate or \emph{(2)}
$\slab{i}=\slab{j}$ otherwise.  The total order $\qlexord$ is
consistent with the lexicographical order: if $\suf{i'}\qlexord{}\suf{j'}$
and $\suf{i'}\not\qequiv\suf{j'}$ then $\suf{i'}<\suf{j'}$.  If the labels
for two suffixes are known then comparing them according to $\qlexord$
and $\qequiv$ takes $\Oh{1}$ time.

After an active subproblem $\subpro{i}$ becomes inactive, it
is possible that some of its suffixes are moved to form other
inactive subproblems $\subpro{i_{j}}$'s but we still need the value of
$\slab{i}$. For this reason, we need to introduce a more general
notion, that of \emph{neighborhood}
$\neigh{i}=\subpro{i}\cup\bigcup_{i=1}^{{l}}\subpro{i_{j}}$, to
preserve what was once an individual active subproblem $\subpro{i}$.
Summing up: if $\subpro{i}$ is \emph{active} then
$\neigh{i}=\subpro{i}$; else, $\neigh{i} \supseteq \subpro{i}$.  In
any case, the reference label is
$\slab{i}=\sum_{\neigh{j}<\neigh{i}}\card{\neigh{j}}
=\card{\set{\suf{j}\vert \suf{j}<\suf{i}
\mbox{ for any }\suf{i}\in\neigh{i}}}$.

For the sake of description, each subproblem $\subpro{i}$ or
neighborhood $\neigh{i}$, and by extension each suffix in them, is
conceptually associated with an \emph{$\alpha$-string} $\alpha_i$.  If
$\subpro{i}$ is \emph{non-degenerate} then $\suf{j}$ has $\alpha_i$ as
prefix iff $\suf{j}\in \neigh{i}$.  If $\subpro{i}$ is
\emph{degenerate}, a weaker property holds since
$\card{\alpha_{i}}=0$: for any $\subpro{j}$ not in $\neigh{i}$, either
$\subpro{j}<\neigh{i}$ or $\neigh{i}<\subpro{j}$.  Observing that only the
integer labels are used to compare suffixes according to $\qlexord$
and $\qequiv$, the $\alpha$-strings will not be maintained
during the computation. For the presentation in the paper, we will
focus on subproblems rather than neighbors, keeping in mind that the
label of an inactive subproblem is that defined for its neighborhood.

\sparagraph{Rationale. } We refine the subproblems as follows. We pick an unsolved subproblem $\subpro{i}$
and refine it into smaller ones: We find the closest ranks
$r_j,r_{j+1}$ for $\subpro{i}$ partitioning it ``evenly,'' namely,
$r_j\leq \slab{i}+\card{\subpro{i}}/2\leq r_{j+1}$. Then, we select
the suffixes $\suf{i_j},\suf{i_{j+1}}$ with ranks $r_j,r_{j+1}$ and
partition $\subpro{i}$ into three new subproblems according to
$\suf{i_j}$ and $\suf{i_{j+1}}$. The new subproblem with the middle
suffixes is exhausted since it has no ranks associated (and will form
$S_{j+1}$), while the other two subproblems are still unsolved.  The
goal is to reach a situation in which each subproblem~$\subpro{i}$ is
either exhausted ($\subpro{i} \equiv S_j$ for some $j$) or
solved ($\subpro{i} \equiv \{\suf{i_j}\}$ for rank $r_j$
and some $j$): namely, $S = \subpro{1} \cup \subpro{2} \cup
\cdots \cup \subpro{z} \equiv S_1 \cup \set{\suf{i_1}} \cup S_2 \cup
\set{\suf{i_2}} \cup \cdots \cup S_K \cup \set{\suf{i_K}} \cup
S_{K+1}$ is the resulting refinement of the initial problem~$S$. This
scheme works if we suppose that $S$ contains independent
strings.  Unfortunately, $S$ contains the suffixes of $T$ and so the
 rescanning cost is the main obstacle. 

\subsection{Agglomerates of subproblems}
\label{subsec:agglomerates}

We group subproblems into agglomerates to model the interplay among
suffixes that share the \emph{same} $\alpha$-string.  We represent
each agglomerate as a threaded dynamic tree where each node represents
a subproblem (a subset of suffixes), as illustrated in the example of
Fig.~\ref{fig:big:full}. 

\sparagraph{Dependency.}
For any two subproblems $\subpro{i},\subpro{j}$, we say that 
$\subpro{i}$ depends directly on $\subpro{j}$
if the 
following hold: $(i)$ $\subpro{i}=\subpro{j}$ or 
$(ii)$ $\subpro{i}\not=\subpro{j}$ and for each 
$\suf{x}\in\subpro{i}$ we have that $\suf{x+1}\in\subpro{j}$.
We extend this relation by transitivity, denoted  $\subdept$, which is
a partial order.
When  $\subpro{i}\subdept\subpro{j}$ we say that 
$\subpro{i}$ \emph{depends on} $\subpro{j}$.
We extend this terminology to single suffixes: a suffix 
$\suf{i}\in\subpro{i'}$ \emph{depends on} $\suf{j}\in\subpro{j'}$ (denoted by 
$\suf{i}\subdept\suf{j}$)  if $i\le j$, $\subpro{i'}\subdept\subpro{j'}$ 
and $\suf{x}\not\in\subpro{i'}\cup\subpro{j'}$, for $i<x<j$. 

\sparagraph{Partial order and tree representation.}
A set of subproblems is an \emph{agglomerate} $A$ if there exists
$\subpro{i}\in A$ such that $\subpro{j}\subdept\subpro{i}$, for each
$\subpro{j}\in A$ (i.e. $A$ has a maximum according to the partial
order $\subdept$).  The Hasse diagram according to $\subdept$ for an
agglomerate $A$ is a tree whose root is the \emph{maximum} subproblem
and the leaves are the \emph{minimal} subproblems of $A$.  Moreover,
the children of an internal node are subproblems directly
  depending on it (see Fig.~\ref{fig:big:full}).  We denote by $\card{A}$
the number of subproblems in $A$, and apply tree terminology to
agglomerates.  A suffix $\suf{x}\in\subpro{i}\in A$ ($x>1$) is a
\emph{contact suffix} if $\suf{x-1}\in\subpro{j}$ and
$\subpro{j}\not\in A$.  A subproblem $\subpro{i}\in A$ is a
\emph{contact subproblem} if it contains at least one contact suffix.
Each leaf of $A$ is a contact node. Also, some internal
  nodes may be contact ones.  When we consider the contact nodes in
preorder (see grey nodes in Fig.~\ref{fig:big:full}), we call this the
\emph{contact visiting order}. Contacts nodes are useful for the refinement.

\sparagraph{Status.}
An agglomerate $A$ is 
\emph{exhausted} if all its subproblems are either exhausted or solved.
An unsolved subproblem  of an agglomerate $A$ 
is a \emph{leading subproblem} if none of its ancestors in $A$ are unsolved.
An agglomerate $A$ is \emph{unsolved} 
if  the following holds:
\begin{property}[Leading Subproblem Property]
  \label{prop:leading-subproblem}
There exists a leading subproblem $\subpro{w}$ in~$A$ s.t.\\
$(i)$~Each ancestor of $\subpro{w}$ has only one child, and
$(ii)$~none of the ancestors of $\subpro{w}$ is a contact node.
\end{property}
Note that Property~\ref{prop:leading-subproblem}.$(i)$ implies that
an unsolved agglomerate has only one leading subproblem.  At
any time, a subproblem $\subpro{i}\in A$ is \emph{active} iff $A$ is
unsolved and either $(i)$ $\subpro{i}$ is unsolved or $(ii)$
$\subpro{i}$ is not solved and is a contact node of $A$ (in this case
$\subpro{i}$ can be exhausted but still active).  For any agglomerate
$A$ and two active subproblems $\subpro{i},\subpro{j}\in A$,
the following properties hold: $(a)$~$\alpha_{i}$ and~$\alpha_{j}$
have a non-void common suffix; $(b)$~if $\subpro{i}$ is a descendant
of $\subpro{j}$ in $A$ then $\alpha_{j}$ is a proper suffix of
$\alpha_{i}$.

\sparagraph{Columns.}
The agglomerates induce columns: a \emph{column} $C$ of $A$ is the
pair $\ppair{c_{i},r_{i}}$ such that $\suf{c_{i}}\subdept\suf{r_{i}}$
where $\suf{c_{i}}$ is a contact suffix (i.e.\mbox{} it belongs to one
of $A$'s contact nodes) and $\suf{r_{i}}$ is a root suffix of $A$
(i.e.\mbox{} it belongs to $A$'s root).  Each column $T[c_{i},
  r_{i}]$ is a contiguous substring of $T$ and, at any time, all the
  columns of all the agglomerates form a non-overlapping partition
  of~$T$. A column $\ppair{c_{i},r_{i}}$ is associated with the
contact subproblem $\subpro{j}$ such that
$\suf{c_{i}}\in\subpro{j}$. We will denote by $\acard{\subpro{j}}$
the number of columns associated with a contact subproblem
$\subpro{j}$.  In Fig.~\ref{fig:big:full}, each agglomerate has its columns
depicted beside its tree. The columns are shown as contiguous
substrings of $T$ (and $T$ as a non-overlapping partitioning of the
columns). For any agglomerate $A$, the number of its root suffixes,
the number of its contact suffixes, the number of its columns and, if
$A$ is unsolved, the number of suffixes in its leading subproblem are
all the same: we denote this quantity by~$\acard{A}$.

\sparagraph{Data structures.}
The algorithmic challenge behind agglomerates is that the total cost
of refining each individual subproblem could be prohibitive: there is
a hidden rescanning cost that we must avoid. 
As we will describe late, the refining algorithm will pick an unsolved agglomerate $A$, and
process it by refining \emph{simultaneously} all of its subproblems
according to the ranks of its leading subproblem $\subpro{w}$. This is
a non-trivial task.
For example, to achieve optimality
(Section~\ref{sec:correctness-analysis}) the refinement of $A$
must be executed in time proportional to the number
$\card{\subpro{w}}$ (i.e\mbox{} $\acard{A}$) of suffixes in
  $\subpro{w}$ plus the number of newly created subproblems,
whereas $A$ can contain many more subproblems and suffixes.

The main structure for an agglomerate $A$ is its Hasse diagram
according to $\subdept$, which is a tree with $\card{A}$ nodes, each
one representing a subproblem of $A$.  Double links between children
and parent are maintained.  If $A$ is unsolved, we also maintain
$(a)$~a pointer to its leading subproblem $\subpro{w}$ and
$(b)$~the number of nodes in the path between $A$'s root and
$\subpro{w}$. 
Only $2\acard{A}$ suffixes are stored for $A$, i.e. its 
$\acard{A}$ columns. They are divided into lists, 
one for each contact subproblem of $A$: the one for  
$\subpro{j}$ contains all the $\acard{\subpro{j}}$
columns $\ppair{c_{i},r_{i}}$ such that $\suf{c_{i}}\in\subpro{j}$ (see 
Fig.~\ref{fig:big:full}). 

We call \emph{skip node} one that is a  
branching node (i.e. one with two or more children) or a 
contact node or the root of $A$ (the conditions are not 
mutually exclusive).
For each skip node $\subpro{j}$ we maintain a \emph{skip link} which is a double 
link between $\subpro{j}$ and its lowest ancestor that is also a skip 
node. 
Clearly, the graph induced by the skip links is also 
a tree. We refer to it as $A$'s \emph{skip tree}.
For each  internal skip node $\subpro{j}$ and each child 
$\subpro{j_{i}}$ of $\subpro{j}$ in the skip tree, we maintain a 
\emph{guide link} that goes from $\subpro{j}$
to the highest node of $A$ that is $(i)$ \emph{unsolved} and 
$(ii)$ both a descendant of $\subpro{j}$ and an ancestor of 
$\subpro{j_{i}}$ (if any such node exists). The skip tree of $A$ with 
its guide links requires $\Oh{\acard{A}}$ space.
\notes{1}{%
In Fig.~\ref{fig:big:full}, skip and guide links are 
depicted as dashed arcs and dotted arrows.
}

Besides the pointers needed by the linked structures containing it, 
each subproblem $\subpro{i}$ carries $\Oh{1}$ words of information: 
$(a)$ its integer label $\slab{i}$, $(b)$ $\card{\subpro{i}}$ 
and, when $\subpro{i}$ is a contact node, the following additional information: $(c)$ $\acard{\subpro{i}}$
$(d)$ a pointer to the list of its 
columns and $(e)$ a pointer to the root of $A$.
For a generic subproblem $\subpro{i}$,  we do 
not explicitly maintain (e.g. in a list associated with 
$\subpro{i}$) either its suffixes or its ranks in $\subprorank{i}$.  
The space required to store an agglomerate $A$ is just
$\Oh{\card{A}+\acard{A}}$ memory words.

We also maintain some \emph{global} bookkeeping structures that are
shared among the agglomerates. $(a)$~For each contact or root
suffix $\suf{i}$, a pointer $\arrsub[i]$ to the subproblem to which it
belongs. $(b)$~A sorted linked list $\sublist$ of all the
subproblems. $(c)$~An array $\rankstrut$ to find the nearest rank in
$\mathcal{R}$ in constant time, so that the set of ranks
$\subprorank{i}$ of any unsolved subproblem $\subpro{i}$ can be
retrieved in $\Oh{\card{\subprorank{i}}}$ time.

\subsection{Basic refining operations: slicing and joining agglomerates}
\label{subsec:agg:oper}

Before describing the general refining scheme in
Section~\ref{sec:algorithm}, we discuss some useful operations that
operate on agglomerates. The \sliceop{} operation takes
in input an agglomerate $A$ and assumes that each column
$\ppair{c_{j},r_{j}}$ of $A$ has been tagged with an integer in
$\set{1,\ldots,d}$, for some $d\le\acard{A}$.  For any suffix
$\suf{i}$ of a subproblem in $A$, let us denote by $\tau(\suf{i})$
the tag of the column $\ppair{c_{j},r_{j}}$ to which it belongs
($c_{j}\le i\le r_{j}$). \sliceop{}$(A)$ obtains the following \emph{without}
changing the columns of $A$:
\begin{enumalpha}
\item Each subproblem $\subpro{x}\in A$ is partitioned (not
  necessarily this gives a refinement) into new subproblems
  $\subpro{x_{1}},\ldots,\subpro{x_{d'}}$, where $d'\le d$ and all the
  suffixes $\suf{i}\in\subpro{x_{j}}$ have the same tag
  $\tau(\suf{i})$, for each $1\le j\le d'$ (after the partitioning
  $\subpro{x}$ ceases to exist).  If all the suffixes of some
  $\subpro{x}\in A$ already have the same tag, $\subpro{x}$ remains
  unmodified.
 
\item All the subproblems in $A$, both the new ones and the unmodified ones, are
  distributed into new agglomerates $A_{1},\ldots,A_{d}$ such that,
  for each $1\le t\le d$ and for each suffix
  $\suf{i}\in\subpro{x_{y}}\in A_{t}$, we have that
  $\tau(\suf{i})=t$. Hence $\sum_{i=1}^{d}\acard{A_{i}}=\acard{A}$, with
  $\card{\bigcup_{i=1}^{d}A_{i}-A}$ newly created subproblems.
\end{enumalpha}

\begin{lemma}\label{lem:slice:comp}
Slicing $A$ into $A_{1},\ldots,A_{d}$ takes 
$\Oh{\acard{A}+\card{\bigcup_{i=1}^{d}A_{i}-A}}$ time.
\end{lemma}

The \emph{join operation} \joinop{}$(A',A)$ is much simpler. An
agglomerate $A'$ is \emph{joinable} with $A$ if there exists a contact
subproblem $\subpro{x}\in A$ such that, for each suffix $\suf{i}$ of
the root subproblem $\subpro{r'}$ of $A'$, it is $\suf{i+1}\in
\subpro{x}$.  Hence $A'\cup A$ is an agglomerate and, by joining with
$A$, we have that $A'$ disappears and only $A$ remains (their trees
are fused with $\subpro{r'}$ child of $\subpro{x}$), thus some columns are \emph{modified}.

\begin{lemma}\label{lem:join:comp} 
If $A'$ is joinable with $A$ then the join operation
requires $\Oh{\acard{A'}}$ time. 
\end{lemma}
 
During the computation we may need to combine the above two
operations. Let agglomerates $A$ and $A_{*}$ be unsolved and
exhausted, respectively.  Also let us assume that $A$ is joinable with
$A_{*}$ and let $\subpro{r}$ be $A$'s root.  The
\slicejoinop{}$(A,A_{*})$ operation does the following: $(i)$ it
slices $A_{*}$ into $A_{*1}$ and $A_{*2}$, where $\ppair{c_{j},r_{j}}$
is a column of $A_{*1}$ iff $\suf{c_{j}-1}\in\subpro{r}\in A$; $(ii)$
it joins $A$ with $A_{*1}$.

\begin{lemma}\label{lem:slicejoin:comp}
\slicejoinop{}$(A,A_{*})$ requires $\Oh{\acard{A}+\card{A_{*1}}}$ time.
\end{lemma}

\section{Optimal Algorithm for Multi-Selection of Suffixes}
\label{sec:algorithm}

\medskip
\sparagraph{Initialization stage.} 
We begin by initializing the bookkeeping data structures described in
Section~\ref{subsec:agglomerates}. Then, let \mselmset{} be an optimal
multi-selection algorithm for a \emph{multiset} of items, thus
generalizing the result in \cite{dobkin:munro} to multisets (see
Section~\ref{subsec:mselmset:algo:full}).  We call
\mselmset{}$\left(\set{T[1],\ldots,T[N]},\mathcal{R}\right)$ using the
alphabet order: this refines the suffixes in $S$ according to their
first symbols into the pivotal multisets $\mathcal{M}_0$,
$\mathcal{F}_1$,$\mathcal{M}_1$,~\dots,$\mathcal{F}_t$,
$\mathcal{M}_{t}$.  Each $\mathcal{M}_{i}\not=\emptyset$ forms a
degenerate subproblem $\subpro{i}$ since no ranks fall within it,
whereas each $\mathcal{F}_{j}$ forms either a solved or unsolved
subproblem $\subpro{j}$.  The agglomerates are initialized as
singletons: each $\subpro{x}$ forms an agglomerate $A_{x}$ that is
either exhausted (i.e.\mbox{} $\subpro{x}$ is either solved or
exhausted) or unsolved (i.e.\mbox{} it satisfies
Property~\ref{prop:leading-subproblem}). Each $A_{x}$ is moved into
one of the groups $\gunsolved$ and $\gexhausted$, which will be used
in the rest of the computation.

\sparagraph{Refinement stage.}
We proceed with the \emph{refine and aggregate stage}.  We execute a
loop until the group $\gunsolved$ is empty. In each iteration of the
loop we pick an agglomerate from $\gunsolved$ and we apply to it the
\emph{Refine and Aggregate Process}, shortly \RAP{}, described in
Section~\ref{subsec:rap}.  During each \RAP{}, some \emph{temporary
  agglomerates} will be created and they may not fulfill the
characterization given in Section~\ref{subsec:agglomerates}. Thus,
any agglomerate created during a \RAP{} is to be considered as
\emph{temporary} until it disappears (since it is joined with another)
or it is moved into either one of the global groups $\gunsolved$ and
$\gexhausted$, at the end of \RAP.

\sparagraph{Finalization stage.}
Once $\gunsolved$ is empty, we run a finalization stage that
returns the $K$ wanted suffixes $\suf{i_1} < \suf{i_2} < \cdots <
\suf{i_K}$ by suitably scanning $\sublist$.

\subsection{The Refine and Aggregate Process (RAP)}
\label{subsec:rap}

\medskip

\sparagraph{First step: establishing which kind of agglomerate.}  
The first kind is when $A$ is \emph{core cyclic}: there is one contact
subproblem $\subpro{x_{*}} \in A$, called the \emph{core} of $A$, such
that $(a)$~there exists at least one column $\ppair{c_{j},r_{j}}$ of
$A$ with $\suf{r_{j}+1}\in\subpro{x_{*}}$, and $(b)$~for each column
$\ppair{c_{i},r_{i}}$ of $A$, either $\suf{r_{i}+1}\in\subpro{x_{*}}$
or $\suf{r_{i}+1}\in\subpro{y_{i}}\not\in A$.  For example, consider
Fig.~\ref{fig:big:full}, where $A \equiv A_{2}$ is core cyclic with core
$\subpro{x_{*}} \equiv \subpro{19}$: the only two columns
$\ppair{c_{i},r_{i}}$ of $A_{2}$ with $\suf{r_{i}+1}\in\subpro{19}$
are $C_{21}=\ppair{132,133}$ and $C_{22}=\ppair{134,135}$, while the
others have $\suf{r_{i}+1}\in\subpro{y_{i}}\not\in A_{2}$.  The
columns in the core $\subpro{x_{*}}$ are all \emph{equal} and
\emph{consecutive} inside~$T$.
The second kind is when $A$ is \emph{generic} (i.e.\mbox{} not core
cyclic) as illustrated by the example in Fig.~\ref{fig:big:full}:
agglomerate $A_{3}$ is generic and ``acyclic'', as each column
$\ppair{c_{i},r_{i}}$ of $A_{3}$ has
$\suf{r_{i}+1}\in\subpro{y_{i}}\not\in A_{3}$; instead, agglomerate
$A_{1}$ is generic and ``cyclic'', as $C_{1}=\ppair{53,60}$ and
$C_{16}=\ppair{118,120}$ have $\suf{61}\in\subpro{17}\in A_{1}$ and
$\suf{121}\in\subpro{16}\in A_{1}$.  More formally, a generic and
``acyclic'' agglomerate~$A$ has each of its columns
$\ppair{c_{i},r_{i}}$ with $\suf{r_{i}+1}\in\subpro{y_{i}} \not \in
A$; a generic and ``cyclic'' agglomerate~$A$ is one where there exist
columns $\ppair{c_{j},r_{j}}$ and $\ppair{c_{j'},r_{j'}}$ such that
$\suf{r_{j}+1}\in\subpro{x}\in A$ and $\suf{r_{j'}+1}\in\subpro{x'}\in
A$ but $\subpro{x}\not=\subpro{x'}$. Note that cyclic
agglomerates of the second kind are treated differently from those of
the first kind.

\sparagraph{Second step: building the keys for the agglomerate.}  
We assign a \emph{key} to each column of agglomerate $A$.  If $A$ is
core cyclic (i.e.\mbox{} first kind), the key for each column
$\ppair{c,r}$ is as follows. Let
$\ppair{c_{1},r_{1}},\ppair{c_{2},r_{2}},\ldots,\ppair{c_{f},r_{f}}$
be the maximal sequence of (equal) columns of $A$ such that
$r+1=c_{1},r_{1}+1=c_{2},\ldots,r_{f-1}+1=c_{f}$.  The key for
$\ppair{c,r}$ is the sequence
$\suf{c_{1}}\suf{c_{2}}\ldots\suf{c_{f}}\suf{r_{f}+1}$.  If $A$ is
generic (i.e.\mbox{} second kind, acyclic or not), each column
$\ppair{c,r}$ simply gets $\suf{r+1}$ as key. What about individual
suffixes?  The key of each $\suf{i}$ is \emph{implicitly} the same as
that assigned to \emph{its} column $\ppair{c_{j},r_{j}}$, where
$c_{j}\le i\le r_{j}$.  However, unlike the columns of $A$, we
  do not access each suffix of $A$ to actually assign it its key,
which is retrieved on the fly only when needed.  Recall that at any
time the total order $\qlexord$ and the equivalence relation $\qequiv$
are defined on the suffixes (Section~\ref{subsec:subproblems}).
Hence, any two keys can be compared in $\Oh{1}$, since these suffixes
are either contact or root, so we can use \arrsub{} to retrieve
their labels: if $A$ is core cyclic, its key becomes $\slab{i}
\slab{i} \cdots \slab{i} \slab{j}$ for labels $\slab{i}$ (repeated $f$
times for $\suf{c_{1}}$, $\suf{c_{2}}$,~\dots, $\suf{c_{f}}$) and
$\slab{j}$ (for $\suf{r_{f}+1}$); if $A$ is generic, its key becomes a
single label (for $\suf{r+1}$).

\sparagraph{Third step: multi-selection on the leading subproblem seen
  as a multiset.}  
We retrieve the ranks $r^{w}_{1},\ldots,r^{w}_{\card{\subprorank{w}}}$
of the leading subproblem $\subpro{w}$ of $A$, using \rankstrut\
(Section~\ref{subsec:agglomerates}). We also retrieve all the suffixes
of $\subpro{w}$ (not explicitly stored for $\subpro{w}$) and their
keys computed in the second step, in
$\Oh{\card{\subpro{w}}}=\Oh{\acard{A}}$ time. Then, we call
\mselmset{}$\left(\subpro{w},\set{\rho_{1},\ldots,\rho_{\card{\subprorank{w}}}}\right)$,
where the ranks are $\rho_{j}=r^{w}_{j}-\slab{w}$ for the label
$\slab{w}$ of $\subpro{w}$, and the order among keys is the one given
by $\qlexord$ and $\qequiv$. Let the resulting pivotal multisets be
$\mathcal{M}_0,\mathcal{F}_1,\mathcal{M}_1,\ldots,\mathcal{F}_t,\mathcal{M}_{t}$,
where $\mathcal{F}_j$ contains the suffixes that are candidates for
some ranks and $t \leq \card{\subprorank{w}}$.  We assign \emph{tags}
to the columns in~$A$ that are consistent with the order among these
pivotal multisets, so that suffixes in the same multiset receive the
same tag. (This is useful for the refinement in the next step.) For
each $0\le j\le t$, let $h_{j}$ be the number of
$\mathcal{M}_{x}\not=\emptyset$ with $0\le x\le j$ (some of the
$\mathcal{M}_{j}$'s may be empty).  We use them to tag each column $C$
of $A$: the tag is $h_{0}$ if the suffix of $\subpro{w}$ that
corresponds to $C$ belongs to $\mathcal{M}_{0}$; the tag is
$j+h_{j-1}$ if that suffix belongs to $\mathcal{F}_{j}$, or $j+h_{j}$
if that suffix belongs to $\mathcal{M}_{j}$, for some $1\le j\le t$.

\sparagraph{Fourth step: slicing the agglomerate.}  
We perform the slicing of $A$ (Section~\ref{subsec:agg:oper}) using
the $d$ tags computed in the third step. We call \sliceop{}$(A)$ and
obtain agglomerates $A_{1},\ldots,A_{d}$, where the columns in each
$A_{i}$ have the same tag and, for $i \neq j$, the tag of $A_{i}$ is
different from that of $A_{j}$. Then $A_{1},\ldots,A_{d}$ are moved
into four groups: the two global ones previously mentioned,
$\gexhausted$ and $\gunsolved$, and two local ones for temporary
agglomerates, called $\gundecided$ and $\gjoinable$, according to the
following rule for the given $A_{i}$, where $1 \leq i \leq d$:

\vspace*{-1ex}

\begin{quote}
$\bullet$~\texttt{if} $A_{i}$'s tag corresponds to a multiset $\mathcal{M}_{j}$
  (see the third step), move it to $\gundecided$;\\
$\bullet$~\texttt{else} $A_{i}$'s tag corresponds to a multiset $\mathcal{F}_{j}$:\\
\hspace*{3em} $-$~\texttt{if} $\acard{A_{i}} = 1$, move $A_{i}$ to $\gexhausted$;\\
\hspace*{3em} $-$~\texttt{else if} $A$ was generic and there is a column
    $\ppair{c_{i},r_{i}}$ of $A_{i}$ with $\suf{r_{i}+1}\in\subpro{x}\in A$,\\
\hspace*{6em}move $A_{i}$ to $\gunsolved$ (since $A_{i}$ satisfies Property~\ref{prop:leading-subproblem});\\
\hspace*{3em} $-$~\texttt{else} move $A_{i}$ to $\gjoinable$.
\end{quote}

\vspace*{-1ex}

This step has a subtle point, since it implicitly induces the
  refinement of many subproblems simultaneously without paying the
  rescanning cost.  Not only the leading subproblem $\subpro{w}$ is
refined into its pivotal multisets using its ranks in
$\subprorank{w}$, but each active subproblem $\subpro{i}$ of $A$ is
refined as well.  However, the refinement of $\subpro{i} \neq
\subpro{w}$ could be coarser than what obtained by refining $\subpro{i}$ 
directly using its ranks in $\subprorank{i}$: indeed, as pointed out
in Section~\ref{subsec:agglomerates}, the $\alpha$-string $\alpha_{w}$
is a proper suffix of $\alpha_{i}$, the
$\alpha$-string of $\subpro{i}$ (descendant of $\subpro{w}$). But $\subpro{i}$'s induced
refinement is for free since the cost is charged to $\subpro{w}$,
and any subsequent refinement for $\subpro{i}$ will surely create new
subproblems, for which we can pay (see what claimed for data structures
in Section~\ref{subsec:agglomerates} and Lemma~\ref{lem:slice:comp}).

\sparagraph{Fifth step: processing \rm$\gundecided$.}
The group $\gundecided$ collects the temporary agglomerates $A_{i}$'s
for which there are no ranks from $\subprorank{w}$ falling within
$A_{i}$. However, there could be other ranks in $\mathcal{R} -
\subprorank{w}$ that could involve $A_{i}$. Hence, $(i)$~$A_{i}$ may
or may not contain unsolved subproblems, and $(ii)$~even if $A_{i}$
contains unsolved subproblems, it may not satisfy
Property~\ref{prop:leading-subproblem}. (Here we may create a
neighborhood from a subproblem of $A_{i}$ as discussed in
Section~\ref{subsec:subproblems}.) For each $A_{i}$ in $\gundecided$,
we retrieve the topmost unsolved subproblems $\subpro{i_1} <
\subpro{i_2} < \cdots < \subpro{i_{d_i-1}}$ from $A_{i}$ in preorder,
in $O(\acard{A_{i}})$ time using its skip tree
(Section~\ref{subsec:agglomerates}).  Since no ancestor of
$\subpro{i_t}$ is unsolved, we assign tag $t$ to its columns, $1 \leq
t \leq d_i-1$. We assign tag $d_i$ to the remaining columns of
$A_{i}$, which are not associated with any $\subpro{i_t}$. We call
\sliceop{}$(A_{i})$ with these $d_i$ tags: for $1 \leq t \leq d_i-1$,
we create a new agglomerate $A_{i_t}$ with leading subproblem
$P_{i_t}$ and put it into $\gunsolved$. We create a new agglomerate
$A_{i_{d_i}}$, exhausted by construction, and put it into
$\gexhausted$.

\sparagraph{Sixth step: processing \rm$\gjoinable$.}  
This step represents the ``aggregate'' part of \RAP.  For each
agglomerate $A_{i}$ in $\gjoinable$, let $A_{i*}$ be the agglomerate
with which $A_{i}$ is joinable (Section~\ref{subsec:agg:oper}). Note
that $A_{i*}$ is either in $\gunsolved$ or $\gexhausted$ (but not in
$\gjoinable$, see the fourth step).

If $A_{i*}$ is in $\gunsolved$, or if $\acard{A_{i*}}=\acard{A_{i}}$,
we call \joinop{}$(A_{i},A_{i*})$ and move the resulting agglomerate
in $\gunsolved$.  In this way, we maintain
Property~\ref{prop:leading-subproblem}: if $A_{i*}$ is unsolved, then
it satisfies the property; if $A_{i*}$ is exhausted, the leading
subproblem of $A_{i}$ is also viable for the agglomerate obtained from
\joinop{}$(A_{i},A_{i*})$ (since $\acard{A_{i*}}=\acard{A_{i}}$).

If $A_{i*}$ is in $\gexhausted$ and $\acard{A_{i*}}>\acard{A_{i}}$:
we call \slicejoinop{}$(A_{i},A_{i*})$ producing $A_{i*1}$ and
$A_{i*2}$.  We move $A_{i*1}$ to $\gunsolved$ as it satisfies
Property~\ref{prop:leading-subproblem}. We move $A_{i*2}$ to
$\gexhausted$: since $A_{i*}$ was exhausted, it did not have a leading
subproblem and $A_{i}$'s leading subproblem is not viable for the
resulting $A_{i*2}$ (since $\acard{A_{i*}}>\acard{A_{i}}$, at least
one of the conditions of Property~\ref{prop:leading-subproblem} is
violated).

\section{Details of the Optimal Algorithm}
\label{sec:details-optimal-algorithm}

\subsection{Multi-selection on multisets}
\label{subsec:mselmset:algo:full}

Consider the multi-selection problem on \emph{multisets} of elements
that are comparable in $\Oh{1}$ time.  The multi-selection algorithm
in \cite{dobkin:munro} does not exploit the presence of equal
elements. Let us describe our variant.  Given a multiset $\mathcal{M}$
and $K$ ranks $r_1<\cdots<r_K$, we want to partition $\mathcal{M}$
into its \emph{pivotal multisets}
$\mathcal{M}_0,\mathcal{F}_1,\mathcal{M}_1,\ldots,\mathcal{F}_t,\mathcal{M}_{t}$
such that the following holds:
\begin{enumroman}
\item For any $0 <  i<j\le t$, for any 
$p_i\in\mathcal{F}_i,e_i\in\mathcal{M}_i,p_j\in\mathcal{F}_j,e_j\in\mathcal{M}_j$, 
we have that $p_i<e_i<p_j<e_j$; moreover, for any 
$e_0\in\mathcal{M}_0,p_1\in\mathcal{F}_1$, we have that $e_0<p_1$.
\item The elements in $\mathcal{F}_i$ are equal and $\card{\mathcal{F}_i}$ is the 
multiplicity of $p_i\in\mathcal{F}_{i}$.
\item For each rank $r_i$ there exists a $\mathcal{F}_j$ such 
that each $p_j$ in $\mathcal{F}_j$ has rank $r_i$.
\item For each $\mathcal{F}_j$ there exists a rank $r_i$ such 
that each $p_j$ in $\mathcal{F}_j$ has rank $r_i$.
\end{enumroman}
Notice that $t\le K$ and $\card{\mathcal{F}_i}\ge 1$ whereas there may 
be some $\mathcal{M}_{i}=\emptyset$.
Let us now describe our algorithm.
\begin{enumerate}
\item [\ ] \hspace*{-3ex}\mselmset{}$\left(\mathcal{M},r_{1},\ldots,r_{K}\right)$
\item If $K=0$ exit. If $K=1$ select and output the element of
rank $r_1$.
\item Find the largest $r_l\le\ceil{\frac{N}{2}}$. Then select the element $p'$ ($p''$) with rank 
$r_l$ ($r_{l+1}$) and all the elements of $\mathcal{M}$ that are equal 
to $p'$ ($p''$).
\item\label{alg:mm:step:part:full} Partition $\mathcal{M}$ into 
$\mathcal{M}',\mathcal{F}',\mathcal{M}'',\mathcal{F}'',\mathcal{M}'''$ 
such that $(a)$ $\mathcal{F}'$ ($\mathcal{F}''$) contains all the elements 
in $\mathcal{M}$ equal to $p'$ ($p''$) and $(b)$ for any 
$e'\in\mathcal{M}',e''\in\mathcal{M}'', e'''\in\mathcal{M}'''$ we have 
that $e'<p'<e''<p''<e'''$.
\item Find the largest $r_a\le\card{\mathcal{M}'}$. Call 
\mselmset{}$\left(\mathcal{M}',r_1,\ldots,r_a\right)$.

\item Let 
$q=\card{\mathcal{M}'}+\card{\mathcal{F}'}+\card{\mathcal{M}''}$.
Find the smallest $r_{b}$ s.t. 
$r_b>q$.\\
Call \mselmset{}$\left(\mathcal{F}''\cup\mathcal{M}''',r_b-q,\ldots,r_K-q\right)$.
\end{enumerate}

The complexity of \mselmset{} can be expressed in terms of the sizes of 
the pivotal multisets of $\mathcal{M}$ w.r.t. $r_{1},\ldots,r_{K}$.
For the sake of description, we will assume that $\log 0=0$.

\begin{lemma}\label{lem:multisel:multiset:full}
The running time of the algorithm \mselmset{} on a multiset $\mathcal{M}$ is upper bounded 
by $c\card{\mathcal{M}}\log\card{\mathcal{M}} 
-c\card{\mathcal{M}_{0}}\log \card{\mathcal{M}_{0}}
-c\sum_{i=1}^{t} \left(\card{\mathcal{F}_i}+\card{\mathcal{M}_i}\right)\log 
\left(\card{\mathcal{F}_i}+\card{\mathcal{M}_i}\right)
+c\card{\mathcal{M}}$
for a suitable integer constant $c$.
\end{lemma}

\begin{proof}
  From step~\ref{alg:mm:step:part:full} and because of the choice of $p'$ and 
$p''$, it is clear that $\mathcal{F}'$,$\mathcal{M}''$ and 
$\mathcal{F}''$ are three of the pivotal multisets. Let $\mathcal{F}'$ 
and $\mathcal{M}''$  
be $\mathcal{F}_s$ and $\mathcal{M}_s$, respectively (thus 
$\mathcal{F}''$ is $\mathcal{F}_{s+1}$). Also, the algorithm never breaks 
any pivotal multiset and the partitioning of the input elements for 
the recursive calls is also a partitioning of the (yet unknown) 
pivotal multisets. 
All the non-recursive steps of the algorithm require $\le c\card{\mathcal{M}}$ time, 
for a suitable integer constant $c$. 
Therefore we have that the running time is upper bounded by the function 
$g\left(\card{\mathcal{M}}\right)=c\card{\mathcal{M}} + 
g\left(\card{\mathcal{M}'}\right)+ 
g\left(\card{\mathcal{F}''\cup\mathcal{M}'''}\right)$.
To prove the upper bound for $g$ we use the function  
$f\left(\card{\mathcal{M}}\right)=
c\card{\mathcal{M}} +c\left(\card{\mathcal{F}_s}+\card{\mathcal{M}_s}\right)
\log\left(\card{\mathcal{F}_s}+
\card{\mathcal{M}_s}\right) +f\left(\card{\mathcal{M}'}\right)+ 
f\left(\card{\mathcal{F}''\cup\mathcal{M}'''}\right)$.
Let us show that
$f\left(\card{\mathcal{M}}\right)=
g\left(\card{\mathcal{M}}\right)+c\sum_{i=1}^{t} 
\left(\card{\mathcal{F}_i}+\card{\mathcal{M}_i}\right)\log 
\left(\card{\mathcal{F}_i}+\card{\mathcal{M}_i}\right)
+c\card{\mathcal{M}_{0}}\log \card{\mathcal{M}_{0}}
$. Since the algorithm does not break $\mathcal{M}$'s pivotal 
multisets, we know that the pivotal multisets of $\mathcal{M}'$ and 
$\mathcal{F}''\cup\mathcal{M}'''$ are 
$\mathcal{M}_0,\mathcal{F}_1,\mathcal{M}_1,\ldots,\mathcal{F}_{s-1},\mathcal{M}_{s-1}$
and 
$\mathcal{F}_{s+1},\mathcal{M}_{s+1},\ldots,\mathcal{F}_{t},\mathcal{M}_{t}$, 
respectively. Hence, by induction, we have that 
$$f\left(\card{\mathcal{M}}\right)=
c\card{\mathcal{M}} +c\left(\card{\mathcal{F}_s}+\card{\mathcal{M}_s}\right)
\log\left(\card{\mathcal{F}_s}+
\card{\mathcal{M}_s}\right)+$$ 
$$+g\left(\card{\mathcal{M}'}\right)+c\sum_{i=1}^{s-1} 
\left(\card{\mathcal{F}_i}+\card{\mathcal{M}_i}\right)\log 
\left(\card{\mathcal{F}_i}+\card{\mathcal{M}_i}\right)
+c\card{\mathcal{M}_{0}}\log \card{\mathcal{M}_{0}}+$$
$$+g\left(\card{\mathcal{F}''\cup\mathcal{M}'''}\right)+
c\sum_{i=s+1}^{t} 
\left(\card{\mathcal{F}_i}+\card{\mathcal{M}_i}\right)\log 
\left(\card{\mathcal{F}_i}+\card{\mathcal{M}_i}\right).
$$
Thus, by the definition of $g$, the relation between $f$ and $g$ 
is proven.
All we have to do now is to prove that 
$f\left(\card{\mathcal{M}}\right)\le c\card{\mathcal{M}}\log \card{\mathcal{M}} +c\card{\mathcal{M}}$, 
and the wanted upper 
bound for $g$ will follow by subtraction.
By the definition of $\mathcal{F}'$ we know that both
$\card{\mathcal{M}'}$ and $\card{\mathcal{F}''\cup\mathcal{M}'''}$ 
are less than $\frac{\card{\mathcal{M}}}{2}$. Thus, by induction we have the following:
$f\left(\card{\mathcal{M}}\right)\le c\card{\mathcal{M}}+ c\left(\card{\mathcal{F}_s}+\card{\mathcal{M}_s}\right)
\log\left(\card{\mathcal{F}_s}+\card{\mathcal{M}_s}\right)+
c\card{\mathcal{M}'}+c\card{\mathcal{M}'}\log\card{\mathcal{M}'}+
c\card{\mathcal{F}''\cup\mathcal{M}'''}+c\card{\mathcal{F}''\cup\mathcal{M}'''}\log\card{\mathcal{F}''\cup\mathcal{M}'''}
\le c\card{\mathcal{M}}+ c\left(\card{\mathcal{F}_s}+\card{\mathcal{M}_s}\right)
\log\card{\mathcal{M}}+
c\card{\mathcal{M}'}+c\card{\mathcal{M}'}\log\frac{\card{\mathcal{M}}}{2}+
c\card{\mathcal{F}''\cup\mathcal{M}'''}+c\card{\mathcal{F}''\cup\mathcal{M}'''}\log\frac{\card{\mathcal{M}}}{2}=
c\card{\mathcal{M}}\log \card{\mathcal{M}} +c\card{\mathcal{M}}$.
\end{proof}

\subsection{Dealing with the agglomerates in \gundecided}
\label{sub:undecided-details:full}

\iparagraph{First:} For each agglomerate $A_{i}$ in $\gundecided$, we
find its highest (i.e.\mbox{} closest to the tree root of $A_{i}$)
unsolved subproblems (if any) and we collect them in
$\visitlead_{i}$. Since visiting the whole tree of $A_{i}$ would cost
too much ($\Oh{\card{A_{i}}}$ time), we use $A_{i}$'s skip tree and
its guide links. In this way the visit takes $\Oh{\acard{A_{i}}}$
time.  Specifically, we use two other lists $\visitlist_{i}$ and
$\visitflag_{i}$ (initially empty) besides $\visitlead_{i}$, starting
from the root with procedure \leadvisit{}$(\subpro{r})$ defined as
follows for a generic subproblem $\subpro{r}$:
\begin{enumerate}
\item If $\subpro{r}$ is unsolved, we append $\subpro{r}$, $1$ and
  $\card{\subpro{r}}$ at the end of $\visitlead_{i}$, $\visitflag_{i}$
  and $\visitlist_{i}$, respectively, and return.
\item Otherwise, if
$\subpro{r}$ is a contact node, we append $0$ and $\acard{\subpro{r}}$
to $\visitflag_{i}$ and $\visitlist_{i}$, respectively (but we do not
return yet). 
\item For each child $\subpro{r_{j}}$ of
$\subpro{r}$ (in $A_{i}$'s skip tree), from the leftmost to the
rightmost one, we do the following. If $\subpro{r}$ has a guide link
to an ancestor $\subpro{r_{x}}$ of $\subpro{r_{j}}$, we append
$\subpro{r_{x}}$, $1$ and $\card{\subpro{r_{x}}}$ at the end of
$\visitlead_{i}$, $\visitflag_{i}$ and $\visitlist_{i}$, respectively,
and return. Otherwise, if no such guide link exists, we call
\leadvisit{}$(\subpro{r_{j}}$).  
\end{enumerate}

\iparagraph{Second:} For each agglomerate $A_{i}$ in $\gundecided$, we
tag its columns so that we are able to either $(a)$~classify $A_{i}$
as unsolved or exhausted or $(b)$~partition $A_{i}$ into some smaller
unsolved or exhausted agglomerates.  Basically, a column $C$ of
$A_{i}$ has the tag $l$, $1\le l<\card{\visitlead_{i}}+1$, if the
subtree rooted at the node $\subpro{j}$ in position $l$ in
$\visitlead_{i}$ contains the contact node with which $C$ is
associated. Otherwise, if no such node exists in $\visitlead_{i}$, $C$
has the tag $\card{\visitlead_{i}}+1$.  Specifically, we compute the
(inclusive) prefix sum $\visitsum_{i}$ of $\visitlist_{i}$ and set
$\visitsum_{i}[0]=0$.  We scan the columns of $A_{i}$ in contact
visiting order (i.e.\mbox{} by using the preorder of the contact nodes
in $A_{i}$). Let us consider the $j$-th one of them and let $l$, $1\le
l\le\card{\visitsum_{i}}$, be the index such that $\visitsum_{i}[l-1]<
j\le\visitsum_{i}[l]$. The $j$-th column is tagged with $l$
($\card{\visitlead_{i}}+1$) if $\visitflag_{i}[l]=1$ ($=0$).

\iparagraph{Third:} For each agglomerate $A_{i}$ in $\gundecided$, we
perform the slicing with the above tags. Let
$t(i)=\card{\visitlead_{i}}$.  If the tags of $A_{i}$ are all equal to
some $l$ s.t. $1\le l\le t(i)$ (to $t(i)+1$), we move $A_{i}$ to
$\gunsolved$ (to $\gexhausted$) and the step ends.  Otherwise, we call
\sliceop{}$(A_{i})$ obtaining $A_{i_{1}},\ldots,A_{i_{t(i)}}$ and,
possibly, $A_{i_{t(i)+1}}$, where $A_{i_{l}}$ is the agglomerate whose
columns have $l$ as tag, for each $1\le l\le t(i)+1$.  The node in
position $l$ of $\visitlead_{i}$ is the leading subproblem of
$A_{i_{l}}$, for each $1\le l\le t(i)$.  Finally, we move
$A_{i_{t(i)+1}}$ to $\gexhausted$, and all the other $A_{i_{j}}$'s to
$\gunsolved$.

To understand why the above computation is correct, let us describe
some properties of the tagging done. For any suffix $\suf{j}$, let us
denote with $\tau(\suf{j})$ the tag of the column $\ppair{c,r}$ such
that $c\le j\le r$.

If a subproblem $\subpro{r}$ of $A_{i}$ is \emph{unsolved} then the 
following holds: 
$(a)$ $\tau(\suf{j'})=\tau(\suf{j''})$, for any two 
$\suf{j'},\suf{j''}\in\subpro{r}$ and $(b)$ 
$1\le\tau(\suf{j'})\le t(i)$.

On the other hand, if for a subproblem $\subpro{r}$ of $A_{i}$ we have that 
conditions $(a)$ and $(b)$ hold, then on the path from 
$\subpro{r}$ to the root there has to be at least one \emph{unsolved subproblem} 
whose suffixes have the same tag $\tau(*)$ as $\subpro{r}$'s (maybe only 
$\subpro{r}$ itself, if it is unsolved). Also, the highest unsolved 
subproblem on said path must be the one in position $\tau(\suf{j'})$ in 
$\visitlead_{i}$.

Given the definition of \sliceop{} in Section~\ref{subsec:agg:oper},
the correctness follows.

\subsection{Slicing agglomerates}
\label{subsec:slice:alg:full}

As we have seen in Section~\ref{subsec:agg:oper}, the slice operation 
receives in input an agglomerate $A$ whose columns have been tagged 
with integers in $\set{1,\ldots,d}$, where $d\le\acard{A}$. Note that
there should be at least two columns with different tags, since
otherwise we do not need to run the slicing.

During the slice operation we deal with instances of the following
\emph{grouping problem}: We are given a list $L$ of objects, each with
an integer tag in $\set{1,\ldots,d,d+1}$ (where the tags of the
columns are $d$ but during parts of the slicing we will need an extra
tag for special purposes).  We want to partition $L$ into $d'\le d+1$
lists $L_{i_{1}},\ldots,L_{i_{d'}}$ such that an object $o\in L_{j}$
iff $o$'s tag is $j$, for each $j\in\set{i_{1},\ldots,i_{d'}}
\subseteq \set{1,\ldots,d,d+1}$. Note that the order in which the
lists $L_{i_{1}},\ldots,L_{i_{d'}}$ are produced does not
necessarily have to follow the order of the tags. Also, the problem is
easy if $|L| = \Omega(d)$ since it falls within the radix sort
scheme. In our case, we can spend $O(d)$ preprocessing time and space
beforehand: after that, for each instance $L$ of the grouping problem,
we assume that $|L| = o(d)$, and we cannot pay $O(d)$ time but just
$O(|L|)$ time.

The procedure \grouping{} solves the above problem as follows. Between
one call and the other, it reuses the same array $J$ of $d+1$ slots
that is allocated at the beginning of the slice operation and is never
reset from one call to another.  Each slot $J[i]$ has two fields:
$J[i].p$, a list pointer, and $J[i].t$, an integer. Let $\eta$ be an
integer timestamp unique for each call (since we can just maintain an
increasing integer throughout all the calls to \grouping{}).
While we scan $L$, we build a list $L'$ of lists and then
return it.  Let $c$ be the tag of the current object $o$. During the
scan we have two cases: $(i)$ if $J[c].t\not=\eta$, we set
$J[c].t=\eta$ and start a new list $L_c$ with $o$ as first element;
also, we append $L_c$ to $L'$ and set $J[c].p$ to point to $L_{c}$;
$(ii)$ if $J[c].t=\eta$, we append $o$ to the list pointed by
$J[c].p$.

\begin{lemma}
\label{lem:grouping:full}
After $O(d)$ preprocessing time and space, each call to the procedure
\grouping{} requires $\Oh{\card{L}}$ time.
\end{lemma}

At this point, we can describe \sliceop{}$(A)$, which has three main
phases: pruning, slicing, and finishing.

\subsubsection{Pruning phase}

The goal is to identify some relevant nodes in the tree reprensenting
the agglomerate~$A$. Speficically, a node $\subpro{i}\in A$ is
\emph{homogeneous} if all the columns associated with all the contact
nodes in the subtree rooted at $\subpro{i}$ have the same tag. We want
to find the set $\pruned$ of the \emph{highest} (i.e. closest to the
root) \emph{homogeneous nodes} of $A$ and tag each one of them
\emph{with the common tag of their columns}. Let us recall
(Lemma~\ref{lem:slice:comp}) that we need an implementation of
\sliceop{} with a time complexity that is of the order of the
number of columns of $A$ plus the total number of the new
  subproblems created. Hence, we cannot touch all the homogeneous
nodes of $A$ since they will not be refined into new subproblems at
this stage. If we are able to find the set $\pruned$ efficiently then,
later on, each subtree of $A$ rooted at a node in $\pruned$ will be just
linked to the corresponding new agglomerate without accessing any of
its internal nodes.  The pruning phase proceeds with the following
steps.

\iparagraph{First:} We traverse \emph{the skip tree of $A$} level by
level from the bottom one.
\begin{itemize}
\item For each $\subpro{l}$ at the lowest level we do the
  following. $(i)$ Since $\subpro{l}$ is a contact node (and a leaf),
  we call \grouping{} on its (tagged) column list. $(ii)$ Then we tag
  $\subpro{l}$ with $c$ (resp., with $d+1$) if from \grouping{} we get
  only one list $L_{c}$ (resp., we get more than one list).
\item For each $\subpro{i}$ at a generic level (different from
  the lowest) we do the following. $(i)$ If $\subpro{i}$ is a leaf, we
  perform the same steps as we did at the lowest level.  $(ii)$
  Otherwise, if $\subpro{i}$ is an internal node, we call \grouping{}
  on the list of objects, where each object is one of $\subpro{i}$'s
  children (and their tags); if $\subpro{i}$ is also a contact node,
  we add more objects to the list, where each object is one of
  $\subpro{i}$'s columns (which already have tags). $(iii)$ Then we
  tag $\subpro{i}$ as we did at the lowest level.
\end{itemize}

\iparagraph{Second:} We traverse $A$'s \emph{skip tree} with the
following recursive \skipvisit{}$(\subpro{r})$, starting from the
root.  \skipvisit{}$(\subpro{i})$ is defined for a node $\subpro{i}$
in the skip tree as follows: $(i)$ if $\subpro{i}$'s tag is $<d+1$, we
output a pointer to $\subpro{i}$ and return; $(ii)$ otherwise, if the
tag is $d+1$ we call \skipvisit{}$(\subpro{i_{j}})$, for each child
(in the skip tree) $\subpro{i_{j}}$ of $\subpro{i}$.

\iparagraph{Third:} We are ready to retrieve and mark the nodes to be
added to $\pruned$. (Note that the root $\subpro{r}$ cannot belong to
$\pruned$.) For each node $\subpro{i}$ outputted in the previous step,
we proceed as follows:
\begin{enumroman}
\item We retrieve $\subpro{i}$'s parent in the skip 
tree $\subpro{j}$. 
\item If $\subpro{i}$ is the $x$-th child of $\subpro{j}$ in the
    skip tree of $A$, we find $\subpro{j}$'s $x$-th child in
    the tree of $A$, say $\subpro{j_{x}}$: note that $\subpro{i}$ is
  a descendant of $\subpro{j_{x}}$ in the tree of $A$ (they could be
  even the same node in some cases).  Then we mark $\subpro{j_{x}}$
  (since we add it to $\pruned{}$) and tag it with $\subpro{i}$'s tag.
\item If $\subpro{i}$ and $\subpro{j_{x}}$ are not the same node, we
  do the following. $(a)$ We create a \emph{temporary skip link}
  between them (for the next phase---the slicing phase). $(b)$ Let $v$
  be the ancestor node of $\subpro{i}$ pointed by the corresponding
  guide link of $\subpro{j}$ (if any). For the sake of clarity,
  observe that traversing the tree of $A$ from its root to
  $\subpro{i}$, we meet $\subpro{j}$, $\subpro{j_{x}}$, $v$, and
  $\subpro{i}$. We create a \emph{temporary guide link} between 
  $\subpro{j_{x}}$ and $v$ (they may be the same node in some cases).
\end{enumroman}

\iparagraph{Fourth:} Let $\mathcal{Y}$ be the set of nodes of $A$ that
are not in a subtree rooted at a node belonging to $\pruned$. We assign
to each $\subpro{i}\in\mathcal{Y}$ a \emph{fingerprint} that is a
unique integer from $\set{1,\ldots,\card{\mathcal{Y}}}$ (any choice
would do, for example the DFS numbering).

\begin{lemma}
\label{lem:pruning-phase:full}
The pruning phase requires $\Oh{\acard{A}+\card{\bigcup_{i=1}^{d}A_{i}-A}}$ time.
\end{lemma}

\subsubsection{Slicing phase} 

The goal is to actually slice the agglomerate $A$ into the
agglomerates defined by the columns' tags, as claimed in
Lemma~\ref{lem:slice:comp}, with some provisions. Indeed, the only
missing things to complete this task are the following: $(a)$~the link
between each contact node and the root; $(b)$~the contact node list;
$(c)$~the correct labels and the correct ordering in $\sublist$ for
the new subproblems created; $(d)$~the guide links of skip nodes that
are not descendant of any node in $\pruned$.  We will deal with these
things in the next phase---the finishing phase.

We invoke the recursive procedure \slicerec{}$()$ on the root of $A$.
A generic call \slicerec{}$(\subpro{r})$, where $\subpro{r}$ indcates
now a generic node in~$A$, has three cases.

\paragraph{\boldmath ${\subpro{r}}$ is a leaf.}
First, we call \grouping{} on $\subpro{r}$'s column list and obtain 
$d'\le d+1$ lists $L_{r_1},\ldots,L_{r_d'}$. By scanning each $L_{r_i}$ 
we obtain $(a)$ all the distinct tags $t_1,\ldots,t_{d'}$ 
and $(b)$ $n_1,\ldots,n_{d'}$ where $n_{i}$ is the number 
of columns in list $L_{r_i}$. 

Second, we create $d'$ new subproblems
$\subpro{r_1},\ldots\subpro{r_{d'}}$ and insert them into \sublist{}
in place of $\subpro{r}$. If $\subpro{r}$ was $A$'s leading
subproblem, we set $\subpro{r_1},\ldots\subpro{r_{d'}}$ to be leading
subproblems of their respective agglomerates (they might not actually
be, see Section~\ref{sub:undecided-details:full}).  For each $1\le i\le d'$:
$(a)$ we
set $\subpro{r_{i}}$'s tag and fingerprint to $t_{i}$ and
$\subpro{r}$'s fingerprint, respectively; $(b)$ we set
$\card{\subpro{r_{i}}}$, $\subpro{r_{i}}$'s column list pointer and
$\slab{r_{i}}$ to $n_{i}$, $L_{r_i}$ and $\slab{r}$, respectively.
Finally, we eliminate $\subpro{r}$ and return
$\subpro{r_1},\subpro{r_2}\ldots\subpro{r_{d'}}$.

\paragraph{\boldmath ${\subpro{r}}$ is a contact node but not a leaf.}
If $\subpro{r}\in\pruned$, we return it immediately (it has been
tagged in the previous phase--the pruning phase). Otherwise we proceed
as follows.

First, we call 
\slicerec{}$(\subpro{r_{i}})$, for each child $\subpro{r_{i}}$ of 
$\subpro{r}$. From each call \slicerec{}$(\subpro{r_{i}})$ we receive 
a set $\mathcal{Q}_{i}$ of root nodes. We call \grouping{} on the list 
with the objects in $\mathcal{C}\cup 
\mathcal{Q}_{1}\cup\cdots\cup\mathcal{Q}_{x}$, where $\mathcal{C}$ 
contains  $\subpro{r}$'s columns. This produces $d'\le d+1$
lists $L_{r_1},\ldots,L_{r_d'}$. By scanning the lists, we obtain 
$(a)$ all the distinct tags $t_1,\ldots,t_{d'}$  
and $(b)$ pairs $\seq{col_1,nod_1},\ldots,\seq{col_{d'},nod_{d'}}$, where 
lists $col_i$ and $nod_i$ contain all the columns and all the nodes in $L_{r_i}$, 
respectively (some of them may be empty). 
Then by scanning each $col_{i}$ and $nod_{i}$ we obtain $(c)$ 
$n_1,\ldots,n_{d'}$, where $n_{i}$ is the number of columns in 
$col_{i}$, and $(d)$ $p_1,\ldots,p_{d'}$, where $p_{i}$ is the total 
number of suffixes in each subproblem in list $nod_i$.

Second, we
create $d'$ new subproblems 
$\subpro{r_1},\ldots\subpro{r_{d'}}$ and insert them in 
\sublist{} in place of $\subpro{r}$. If $\subpro{r}$ was $A$'s leading 
subproblem, we set  $\subpro{r_1},\ldots\subpro{r_{d'}}$ to be leading 
subproblems of their respective agglomerates (same as the above case
when $\subpro{r}$ is a leaf).
For each $1\le i\le d'$: $(a)$ we set $\subpro{r_{i}}$'s tag and 
fingerprint to $t_{i}$ and $\subpro{r}$'s fingerprint, respectively;
$(b)$ we set $\card{\subpro{r_{i}}}$, $\subpro{r_{i}}$ column list 
pointer (in case $\subpro{r_{i}}$ is a new contact node) and $\slab{r_{i}}$ 
to $n_{i}+p_{i}$, $col_{i}$ and $\slab{r}$, respectively; $(c)$ we make 
the nodes in $nod_{i}$ be $\subpro{r_{i}}$'s children.

Third, for each $\subpro{r_{i}}$ and each child
$\subpro{r_{ij}}$ of $\subpro{r_{i}}$ we do as follows. 
\begin{enumroman}
\item If $\subpro{r_{i}}$ is not a skip node but its only child
$\subpro{r_{ij}}$ \emph{is}, we 
create a temporary link between them. 
\item If neither  $\subpro{r_{i}}$
nor $\subpro{r_{ij}}$ is a skip node, we redirect to
$\subpro{r_{i}}$ the (only) temporary link 
that goes into $\subpro{r_{ij}}$.
\item If $\subpro{r_{i}}$ is a skip node and 
$\subpro{r_{ij}}$ \emph{is not}, we  
redirect  to $\subpro{r_{i}}$ the temporary link 
that goes into $\subpro{r_{ij}}$, and we 
change it into a skip link. 
\item If \emph{both} $\subpro{r_{i}}$ and $\subpro{r_{ij}}$ 
are skip nodes, we create a skip link between them.
\end{enumroman}
After that, we eliminate $\subpro{r}$ and return  
$\subpro{r_1},\subpro{r_2}\ldots\subpro{r_{d'}}$.

\paragraph{\boldmath${\subpro{r}}$ is not a contact node.} 
This case is analogous to the previous one, only simpler, because 
$\subpro{r}$ is not a contact node and no new contact nodes can be 
created from it.

\begin{lemma}
\label{lem:slicing-phase:full}
The slicing phase requires $\Oh{\acard{A}+\card{\bigcup_{i=1}^{d}A_{i}-A}}$ time.
\end{lemma}

\subsubsection{Finishing phase}
\label{subsubsec:slice:finishing:full}

The finishing phase adds the missing information from the previous
phase---the slicing phase. It proceeds with the following steps.

\iparagraph{First:} We sort \emph{all} the new subproblems
$\subpro{j}$'s according the keys $\seq{m_j,t_j}$, where $m_j$ and
$t_j$ are $\subpro{j}$'s fingerprint and tag, respectively. Since each
$\seq{m_j,t_j}$ has $\Oh{\log \card{A}}$ bits, we can use radix sort.

\iparagraph{Second:} After the sorting, for each old \emph{active}
subproblem $\subpro{i}$, we have that new subproblems
$\subpro{i_1},\ldots,\subpro{i_{x}}$ are grouped together \emph{and}
in their correct relative order. Hence, we can reattach them in
$\sublist$ in the correct order and also set the correct label
$\slab{i_j}$ for each one of them. If $\subpro{i}$ was inactive, we
leave $\subpro{i_1},\ldots,\subpro{i_{x}}$ and their labels as they
are.

\iparagraph{Third:} For each new agglomerate $A_{i}$ we build its
contact node list by visiting its skip tree. Then we scan its contact
node list and for each contact node we set the link to $A_{i}$'s root.

\iparagraph{Fourth:} For each new agglomerate $A_{i}$ we need to
create the guide links for the skip nodes of $A_{i}$ that are not
descendants of the nodes in $\pruned$.{} To that end, we call
\guidevisit{}$(\subpro{i})$, defined as follows, on $A_{i}$'s root.
\begin{enumroman}
\item For each 
child $\subpro{i_{j}}$ of $\subpro{i}$ in $A$'s \emph{skip tree}, we scan the 
nodes $\subpro{i_{x}}$ of  $A$'s \emph{tree} that are both \emph{descendants of 
$\subpro{i}$} and \emph{ancestors of  $\subpro{i_{j}}$} starting from the 
highest one. We keep scanning them until we find a $\subpro{i_{x}}$ 
that falls into one of three cases: $(a)$ it is unsolved, 
$(b)$ it is in $\pruned$ or $(c)$ it is $\subpro{i_{j}}$. 
In case $(a)$ we create a guide link between $\subpro{i}$ and 
$\subpro{i_{x}}$. In case $(b)$ if $\subpro{i_{x}}$ has a temporary 
guide link (possibly created in the pruning phase) to a node 
$v$, we create a guide link between $\subpro{i}$ and 
$v$. Otherwise no guide link is created.
\item  We call \guidevisit{}$(\subpro{i_{j}})$, for each 
child $\subpro{i_{j}}$ of $\subpro{i}$ in $A$'s \emph{skip tree}.
\end{enumroman}

\begin{lemma}
\label{lem:finishing-phase:full}
The finishing phase requires $\Oh{\acard{A}+\card{\bigcup_{i=1}^{d}A_{i}-A}}$ time.
\end{lemma}

\subsection{Joining agglomerates}
\label{subsec:join:alg:full}

Given an agglomerate $A'$ that is joinable to $A$, we need to attach
the root of $A'$ to a suitable place in $A$. Let $\subpro{x}\in
A$ be the contact subproblem such that $\suf{i+1}\in \subpro{x}$, for
each suffix $\suf{i}$ of the root subproblem $\subpro{r'}$ of $A'$.
The operation \joinop{}$(A',A)$ (see Section~\ref{subsec:agg:oper})
proceeds with the following steps.

\iparagraph{First:} We fuse the trees of $A$ and $A'$  by 
making $\subpro{r'}$ be the new leftmost children of $\subpro{x}$.

\iparagraph{Second:} Let $L$ and $L'$ be the contact node lists of $A$
and $A'$, respectively. Let $p$ and $s$ be the predecessor and
successor of $\subpro{x}$ in $L$, respectively. Let $b$ and $e$ be the
leftmost and the rightmost node in $L'$, respectively. If after the
fusion $\subpro{x}$ is not (is still) a contact node, we link $p$
($\subpro{x}$) to $b$ and $e$ to $s$.

\iparagraph{Third:} 
Let us define $v$ as follows: $(i)$ if $\subpro{x}$ is still a 
skip node after the fusion (it may not be a contact node anymore),
then $v$ is $\subpro{x}$; $(ii)$ otherwise $v$ is the ancestor of 
$\subpro{x}$ pointed by its skip link.
If $\subpro{r'}$ is 
a contact or branching node, we create a skip link between $v$ and $\subpro{r'}$.
Otherwise $\subpro{r'}$ is not a skip node of the final
agglomerate. Hence, the only skip link from a node in $A'$ to $\subpro{r'}$ 
is redirected to $v$. We do the same with the only guide link of 
$\subpro{r'}$.

\iparagraph{Fourth:} 
Each column $\ppair{c'_{i},r'_{i}}$ of $A'$ is changed 
into $\ppair{c'_{i},r_{j}}$ where 
$\ppair{c_{j},r_{j}}$ is the column of $A$ associated with $\subpro{x}$ 
such that $r'_{i}+1=c_{j}$. The column $\ppair{c_{j},r_{j}}$ is deleted 
since $\suf{c_{j}}$ is not a contact suffix anymore.

\iparagraph{Fifth:} 
We set the root pointer of each contact node of 
$A'$ to $A$'s root.

Since the total number of skip links, columns and contact 
nodes of $A'$ and deleted pairs 
of $A$ is $\Oh{\acard{A'}}$, Lemma~\ref{lem:join:comp} is proven.

\subsection{Slicing and joining with an exhausted agglomerate} 
\label{subsec:slicejoin:alg:full}

Let us now describe the \slicejoinop{}$(A,A_{*})$ operation used in
the last step of \RAP{} and whose effect has been described in
Section~\ref{subsec:agg:oper}. Let $A$ and $A_{*}$ be unsolved and
exhausted, respectively, and let us assume that $A$ is joinable with
$A_{*}$. Let $\subpro{x}$ be the contact node of $A_{*}$ such that
$\suf{i+1}\in \subpro{x}$, for each suffix $\suf{i}$ of the root
subproblem $\subpro{r}$ of $A$.

For the slicing part of \slicejoinop{}, we \emph{do not} explicitly
tag $A_{*}$'s columns but, conceptually, we would have the following:
$(i)$ only two tags ($1$ and $2$) for the columns; $(ii)$ $\subpro{x}$
is the only contact subproblem of $A_{*}$ with some columns
with tag $1$, and all the other columns with tag $2$. Because of that, we
know that the only subproblems of $A_{*}$ that are partitioned during
the slicing part of \slicejoinop{} are $\subpro{x}$ and its
ancestors. Thus, we do not need the pruning phase of
\sliceop{} because $\subpro{x}$ is the only non-homogeneous node.

The slicing phase is the same as in \sliceop{} except for 
two things. First, the 
recursion does not touch any node that is not on the path from the root 
of $A_{*}$ to 
$\subpro{x}$ (since they are implicitly in $\pruned$).{}
Second, when we treat $\subpro{x}$ (which is a contact node) we do not touch 
its column list at all, we just create the two subproblems 
$\subpro{x_{1}}$ and 
$\subpro{x_{2}}$, and we link \emph{the whole column list} of 
$\subpro{x}$ to  $\subpro{x_{2}}$  (which will be part of $A_{*2}$ and 
is still a contact node). We can leave in the column list of 
$\subpro{x_{2}}$ all those columns that, after a normal \sliceop{} of 
$A_{*}$, would end up being 
associated with $\subpro{x_{1}}$  without incurring in any trouble for 
two reasons: $(i)$ after 
the join of $A$ and $A_{*1}$ they would disappear anyway, and $(ii)$ $A_{*2}$ 
is exhausted and its contact nodes are not active anymore.  

We also do not need the finishing 
phase of \sliceop{}, 
since all the subproblems in $A_{*}$ are exhausted and $A$ will join with 
$A_{*1}$ after the slicing part of \slicejoinop{}.

Because of the characteristics of $A_{*1}$, to join $A$ with it we just 
$(i)$ link $A$'s root to the only contact node of $A_{*1}$ 
and $(ii)$ change each column $\ppair{c_{j},r_{j}}$ of $A$ to  
$\ppair{c_{j},r_{j}+l}$, where $l$ is the number of nodes  $A_{*1}$ 
(whose tree is just a path).

Thus, Lemma~\ref{lem:slicejoin:comp} follows as a corollary of 
Lemma~\ref{lem:slice:comp}.

\subsection{Finalization stage}
\label{sub:finalization-stage:full}

We finally have to store into $\sublist$ all the $K$ wanted
suffixes. Note that we need to retrieve them from the columns that
contains them. To this end, we output the set of $K$ pairs
$\gpairs=\set{\seq{r_{i},j}\vert r_{i}\in\mathcal{R} \mbox{ and
    $\suf{j}$ has rank $r_{i}$}}$ in the following way. We scan
$\sublist$ and for each subproblem $\subpro{x}$ that is both
\emph{solved} and a \emph{leaf} we do as follows. First we add
$\seq{r',j_{x}}$ to $\gpairs$, where $r'$ and $\suf{j_{x}}$ are
$\subpro{x}$'s only rank and only suffix. Then we retrieve all the
ancestors of $\subpro{x}$ (in the tree of its agglomerate) that
  are also solved. They are all the nodes closest to $\subpro{x}$ in
its leaf-to-root path. Let $\subpro{x_{y}}$ be the $y$-th closest one
of them, we add $\seq{r^{y},j_{x}+y}$ to $\gpairs$, where $r^{y}$ is
the only rank of $\subpro{x_{y}}$ (and $\suf{j_{x}+y}$ is clearly
$\subpro{x_{y}}$'s only suffix).

\begin{lemma}
  \label{lem:finalization-stage:full}
  The finalization stage requires $O(N)$ time.
\end{lemma}

\section{Correctness and Analysis}
\label{sec:correctness-analysis}
\label{sec:analysis:full}

Correctness and complexity are strictly related, so we discuss them
together. here we give the lemmas needed to prove the theorems stated
in the Introduction, and few proofs. We devote
Section~\ref{sec:proofs} to the remaining proofs.

Consider first a simplified scenario where we have to
perform multi-selection on a prefix-free set $\mathcal{Z}$ of $N$
\emph{independent} strings of total length~$L$, using our set
${\mathcal R}$ of $K$ ranks. We adopt the same notation as
in formula~\eqref{eq:multi-bound} and Section~\ref{sec:def}.

We run \mselmset{}$({\mathcal Z}, {\mathcal R})$ on the first symbol
of all the strings in $\mathcal{Z}$. This partitions the strings into
unsolved, solved and exhausted subproblems: $\subpro{i}$ is unsolved
when it contains all the strings with the same first symbols and
$\card{\subpro{i}}>\card{\subprorank{i}}\ge 1$; or, $\subpro{i}$ is
solved when $\card{\subpro{i}}=\card{\subprorank{i}}=1$; finally, if
$\subpro{i}$ is exhausted then $\card{\subprorank{i}}=0$ and the first
symbols of its strings may not be the same. We repeat the refining
steps until there are no more unsolved subproblems.  We pick any
unsolved subproblem $\subpro{i}$. Let $a_i$ be the length of the
common prefix (of its strings) examined so far: we invoke
\mselmset{}$(\subpro{i}, \subprorank{i})$ using the alphabetic order
on the symbols $y[a_i+1]$ for $y \in \subpro{i}$.  Thus we refine
$\subpro{i}$ into smaller subproblems, classify them as described
above, and repeat the steps.
\begin{lemma}
  \label{lem:msel:strings:comp}
  The running time of the multi-selection algorithm with $K$ ranks for
  a prefix-free set of $N$ independent strings of total length $L$ is
  upper bounded by $O\left(N\log N - \sum_{j=0}^K \Delta_j \log
    \Delta_j+N+L\right)$.
\end{lemma}

Here we focus on the multi-selection for our set~$S$ of $N$
suffixes, and show how to consider them as a set of $N$ independent
\emph{virtual} strings.

\begin{lemma}
\label{lem:msel:suffixes:comp}
The running time of the suffix multi-selection algorithm for a text of
length $N$ is upper bounded by $O\left(N\log N - \sum_{j=0}^K \Delta_j
  \log \Delta_j+N+\rapbpart\right)$, where $\rapbpart$ is the total 
time required by all the \RAP{}s minus the time for the \mselmset{} 
calls.
\end{lemma}
\begin{proof}
  Consider the computation described in Sections~\ref{sec:def}
  and~\ref{sec:algorithm}, and the \emph{virtual symbol cost} of the
  calls to \mselmset{}: when applied to a subproblem $\subpro{w}$, it
  performs comparisons using a certain order on $\subpro{w}$'s keys,
  which can be seen as the \emph{virtual} symbols of independent
  strings.  Indeed, these virtual symbols are exclusively created and
  ``used'' for $\subpro{w}$. Unlike $T$'s symbols, virtual
    symbols are not shared by subproblems.  Each $\subpro{w}$ has
  associated $\card{\subpro{w}}$ virtual strings that are made
  up of all the virtual symbols \emph{created} to refine $\subpro{w}$
  during several {\RAP}s, every time $\subpro{w}$ is the leading
    subproblem of its current agglomerate. And, unlike the suffixes
  of $T$, all the virtual strings are independent.

  Let $L$ be the total number of virtual symbols thus created by all
  the \RAP{}s. Lemma~\ref{lem:msel:strings:comp} reports the virtual
  symbol cost for all the \mselmset{} calls, that is, their total
  contribution to the final cost of our multi-selection algorithm. We
  have to add the cost of the rest of the computation, which is
  $O(\rapbpart)$ by definition. It remains to show that $L = O(\rapbpart)$,
  thus proving the claimed bound.

  When \mselmset{} is applied to $\subpro{w}$, let $A$ be the
  agglomerate for which $\subpro{w}$ is its leading subproblem.  The
  number $\card{\subpro{w}}$ of virtual symbols created in this call
  satisfies $\card{\subpro{w}} = \acard{A}$. Since the \RAP\
  computation time for this step (minus the call to \mselmset) is
  $\Omega(\acard{A})$ (e.g.\mbox{}
  Lemmas~\ref{lem:slice:comp}--~\ref{lem:slicejoin:comp}), this
  computation time is an upper bound for $\card{\subpro{w}}$.  Summing
  up over all the {\RAP}s, we obtain that the total number $L$ of
  virtual symbols thus created is upper bounded by the computation
  time of all the {\RAP}s minus the \mselmset{} calls, namely,
  $L = O(\rapbpart)$.
\end{proof}

\begin{lemma}
\label{lem:msel:Ksuffixes:comp}
The running time of the suffix multi-selection algorithm for a text of
length $N$ is upper bounded by  $O(K\log K+N+\rapbpart)$
when ${\mathcal R}$ is an interval of $K$ consecutive ranks.
\end{lemma}

We need an intermediate stage to find the suffixes of ranks $r_1$ and
$r_K$, and create a subproblem with the remaining $K-2$ ones. After
that, we run our multi-selection. 

\begin{lemma}
  \label{lem:rapN}
  Independently of the choice of the ranks in ${\mathcal R}$, 
  $\rapbpart = O(N)$. 
\end{lemma}
\begin{proof}
  We give an analysis based on counting the following \emph{types of
    events}:
  \textsf{(a)~\emph{Subproblem creation}}: when some $\subpro{j}$ is
  partitioned into $\subpro{i_{1}},\ldots,\subpro{i_{p}}$ (during a
  \sliceop{} or at the beginning of a \slicejoinop{}).
  \textsf{(b)~\emph{Suffix discovery}}: when some
  $\suf{w}\in\subpro{w}$ is recognized as one of the wanted suffixes
  with rank in ${\mathcal R}$ (third step of \RAP{}).
  \textsf{(c)~\emph{Suffix exhaustion}}: when some suffix
  $\suf{e}\in\subpro{w}\in A$ becomes exhausted, where $\subpro{w}$ is
  the leading subproblem of $A$ (third step of \RAP{}).
  \textsf{(d)~\emph{Column fusion}}: when a column of $A'$ is fused
  with one of $A$ (during \joinop{}$(A',A)$ or at the end of a
  \slicejoinop{}).
  \textsf{(e)~\emph{Inner collision}}: when processing a column
  $\ppair{c_{i},r_{i}}$ of $A$ such that $\ppair{c_{j},r_{j}}$ is also
  in $A$ and satisfies $c_{j}=r_{i}+1$, while $\suf{c_{i}}$ and
  $\suf{c_{j}}$ do not belong to the same subproblem (second
  step of \RAP{}).
  
  \textit{Claim:} There are overall $\Oh{N}$ events occurring in any
  execution of our algorithms. Indeed, the same event cannot
    repeat.  A column is never divided and a subproblem is never
  merged with others.  If a suffix is exhausted or a wanted suffix is
  found, they can never be in an unsolved subproblem again. An inner
  collision is unique, since the two colliding columns will not be
  part of the same agglomerate.  Since we start with $N$ columns and one
  subproblem, the total number of events is $\Oh{N}$.
  
  For a generic \RAP{} on an agglomerate $A$, let $N_A$ be the number
  of events thus occurring.  Since the overall number of events is
  $O(N)$, this implies that $\sum_A N_A = O(N)$.  We show that $N_A
  \geq \acard{A} + \Pi_A$, where $\Pi_A$ is the number of created
  subproblems.  Consider events~\textsf{(a)}, and observe that their
  contribution is $\Pi_A$, which is the sum of three quantities:
  $|\cup_{i=1}^{d}A_{i}-A|$, where $A_{1},\ldots,A_{d}$ are the
  agglomerates into which $A$ is sliced in the fourth step;
  $\sum_{\{A_{i}\in\gundecided\}}|\cup_{j=1}^{d_i}A_{i_{j}}-A_{i}|$,
  where $A_{i_{1}},\ldots,A_{i_{d_i}}$ are the agglomerates into which
  each $A_{i}$ is sliced in the fifth step;
  $\sum_{\{A_{i}\in\gsljoinable\}} \card{A_{i*1}}$, where
  $\gsljoinable$ denotes all the $A_{i}$'s in $\gjoinable$ that need a
  \slicejoinop{} in the sixth step.
  
  As for events~\textsf{(b)}--\textsf{(e)}, they totalize at least
  $\acard{A}$ in number and occur in the fourth step. Namely, the
  number of \textsf{(b)}s and \textsf{(c)}s is at least the total
  number of columns of $A_i$'s that go to $\gexhausted$ or to
  $\gundecided$: each column of each $A_i$ moved to $\gundecided$
  corresponds to~\textsf{(c)}, and each column of each $A_i$ moved to
  $\gexhausted$ corresponds to either~\textsf{(b)}
  or~\textsf{(c)}. The number of \textsf{(d)}s and \textsf{(e)}s is at
  least the total number of columns of the $A_{i}$'s that are
  moved to $\gjoinable$ and $\gunsolved$, respectively.  Since the
  total number of involved columns in the fourth step is
  $\sum_{i=1}^{d}\acard{A_{i}}=\acard{A}$, we obtain the claimed
  number.
  
  At this point, to prove $\rapbpart = O(N)$, it remains to see that
  the cost of a generic \RAP{} (\mselmset{} excluded) on an
  agglomerate $A$ is $O(\acard{A} + \Pi_A)$ time.  The first three
  steps of the \RAP{} take $\Oh{\acard{A}}$ time.  The costs of the
  fourth and fifth steps are given by Lemma~\ref{lem:slice:comp}:
  precisely, $\Oh{\acard{A}+|\cup_{i=1}^{d}A_{i}-A|}$ and
  $O(\sum_{\{A_{i}\in\gundecided\}}(\acard{A_{i}}+
  |\cup_{j=1}^{d_i}A_{i_{j}}-A_{i}|))$.  By Lemmas~\ref{lem:join:comp}
  and~\ref{lem:slicejoin:comp}, the total cost of the sixth step is
  $O(\sum_{\{A_{i}\in\gjoinable\}} \acard{A_{i}}+
  \sum_{\{A_{i}\in\gsljoinable\}} \card{A_{i*1}})$.  Since
  $\sum_{\{A_{i}\in\gundecided\}}\acard{A_{i}}+\sum_{\{A_{i}\in\gjoinable\}}
  \acard{A_{i}}$ $\leq\acard{A}$, the total cost is $O(\acard{A}+\Pi_A)$
  and $\rapbpart = O(\sum_A (\acard{A}+\Pi_A)) = O(\sum_A N_A) =
  O(N)$.
\end{proof}


For any two subproblems $\subpro{i}$ and $\subpro{i'}$ such that
$\neigh{i}\not=\neigh{i'}$, let $\lcp{\subpro{i},\subpro{i'}}$ be the
length of the longest common prefix of any two suffixes
$\suf{j_i}\in\subpro{i}$ and $\suf{j_{i'}}\in\subpro{i'}$.  
We have the following:
\begin{lemma}\label{lem:lcp}
  The suffix multi-selection algorithm can support
  $\lcp{\subpro{i},\subpro{i'}}$ queries in $\Oh{1}$ time, for any two
  $\subpro{i},\subpro{i'}$ such that $\neigh{i}\not=\neigh{i'}$,
  without changing its asymptotic time complexity.
\end{lemma}
\begin{proof}
  Consider a snapshot of the computation, recalling that we maintain
  the sorted linked list $\sublist$ of all the subproblems (where some
  of them belong to the same neighborhood). Suppose by induction that
  we have the \textit{lcp}'s between consecutive subproblems. We use a
  Cartesian Tree (CT) plus lowest common ancestor (LCA) dynamic
  queries to compute the \textit{lcp} for any two subproblems. We
  maintain the induction at the end of the \sliceop\ operation: when
  $A_i$ is sliced into $A_{i_1}, \ldots, A_{i_d}$, the newly created
  subproblems $\subpro{x_1}, \ldots, \subpro{x_{d'}}$ for each
  $A_{i_j}$ are stored contiguously in $\sublist$. We compute
  $\mathit{lcp}(\subpro{x_1}, \subpro{x_2})$,~\dots,
  $\mathit{lcp}(\subpro{x_{d'-1}}, \subpro{x_{d'}})$ using using their
  keys, the CT, and the LCA queries, in $O(d')$ time, which we can pay
  for.  Since these \textit{lcp}'s are longer, we just need to add
  $d'-1$ new leaves to CT, without increasing the asymptotic complexity.
\end{proof}

\section{Proofs}
\label{sec:proofs}

\subsection{Multi-selection on a set of strings}
\label{subsec:msel:strings:full}

In order to prove the upper bound for the time complexity of our 
suffix multi-selection algorithm, we first need to prove the complexity 
of the following algorithm for multi-selection of \emph{independent 
strings} (i.e.\mbox{} not sharing symbols). 

We have a set $\mathcal{Z}$ of $N$ \emph{independent strings}
of total length $L$ such that none of them is the 
prefix of another. We want to select the $K$ strings with ranks in 
$\mathcal{R}=\set{r_{1},\ldots,r_{K}}$, where 
$r_1<\ldots< r_K$. Let $r_0$ and $r_{K+1}$ be $0$ and $N+1$ 
respectively. 

In our multi-selection algorithm for independent strings a subproblem 
$\ssubpro{i}$ is a subset of $\mathcal{Z}$ associated with 
$(a)$ $\ssubprorank{i}\subseteq\mathcal{R}$, $(b)$ $\sless{i}$, the number 
of strings in $\mathcal{Z}$ that are less than each one in 
$\ssubpro{i}$ and $(c)$ an offset $\soffset{i}\in\set{1,\ldots,l_{i}}$, 
where $l_{i}$ is the length of the shortest string in $\ssubpro{i}$.

\begin{enumerate}
\item\label{multistring:stepone:full} We put in $\gunsolved$ the first
subproblem 
$\ssubpro{0}=\mathcal{Z}$ where $\ssubprorank{0}=\mathcal{R}$, 
$\sless{0}=0$ and $\soffset{0}=1$.

\item\label{multistring:steptwo:full} We repeat the following steps until
$\gunsolved$ is empty. 
\begin{enumerate}
\item\label{multistring:steptwo:a:full} We pick a 
subproblem $\subpro{i}$ in $\gunsolved$ and we call 
\mselmset{}$(\mathcal{M},\mathcal{R}_{i}^{-})$ where 
$\mathcal{M}=\set{s[\soffset{i}]\;\vert\; s\in\ssubpro{i}}$
and $\mathcal{R}_{i}^{-}=\set{r_{j}-\sless{i}\;\vert\; r_{j}\in\ssubprorank{i}}$.
This partitions $\mathcal{M}$ into its pivotal multisets 
$\mathcal{M}_0,\mathcal{F}_1,\mathcal{M}_1,\ldots,\mathcal{F}_t,\mathcal{M}_{t}$. 

\item\label{multistring:steptwo:b:full} For each $\mathcal{F}_{j}$, we collect 
the set $\ssubpro{i_{j}}\subseteq\ssubpro{i}$ of the strings 
corresponding to the symbols in $\mathcal{F}_{j}$, and  
we set $\soffset{i_{j}}=\soffset{i}+1$, 
$\sless{i_{j}}=\card{\mathcal{M}_{0}}+\sum_{x=1}^{j-1}\card{\mathcal{F}_{x}\cup\mathcal{M}_{x}}$
and $\ssubprorank{i_{j}}=\set{r_{y}\in\ssubprorank{i}\;\vert\; 
\sless{i_{j}}<r_{y}\le\sless{i_{j}}+\card{\ssubpro{{i}_{j}}}}$.

\item If $\card{\ssubpro{i_{j}}}=1$ we move $\ssubpro{i_{j}}$ to $\gsolved$. 
Otherwise, we move $\ssubpro{i_{j}}$ to $\gunsolved$. 
\end{enumerate}

\item\label{multistring:stepthree:full} For each $\ssubpro{i}$ in $\gsolved$, we
output 
$\seq{s_{i},r_{i}}$, where $s_{i}\in\ssubpro{i}$ and 
$r_{i}\in\ssubprorank{i}$.
\end{enumerate}

\begin{lemma}\label{lem:msel:strings:comp:full}
The running time of the multi-selection algorithm for a set 
$\mathcal{Z}$ of independent  
(and prefix free) strings of total length $L$ is upper bounded by
$$ c\left(N\log N - \sum_{j=0}^K \Delta_j \log \Delta_j+N+5L'\right) 
$$
where $c$ is a suitable integer constant, 
$r_0 = 0$, $r_{K+1} = N+1$, $\Delta_j \equiv r_{j+1} - r_j$ for
$0 \leq j \leq K$, and $L'\le L$ is the smallest number of symbols we 
need to probe to find the wanted strings.
\end{lemma}
\begin{proof}
The \emph{offset size} of a subproblem $\ssubpro{i}$ is the 
total length of the strings in
$\set{s[\soffset{i}\cdots\card{s}]\; \vert\; s\in\ssubpro{i}}$.
It is easy to see that the computations on distinct subproblems proceed 
independently from one another.
Hence, we prove the thesis by induction on the offset size of subproblems.

Let us consider subproblem $\ssubpro{0}$ 
(step~\ref{multistring:stepone:full}), its corresponding multiset 
$\mathcal{M}$ and the pivotal multisets of  $\mathcal{M}$
computed in step~\ref{multistring:steptwo:full}. 
Because of Lemma~\ref{lem:multisel:multiset:full}, after the first 
execution of step~\ref{multistring:steptwo:a:full} the 
terms $cN\log N$ and $cN$ of the wanted upper bound are accounted for.

Let us first deal with the pivotal multisets $\mathcal{F}_{*}$ of $\mathcal{M}$ 
such that $\card{\mathcal{F}_{*}}=1$. They correspond to solved 
subproblems and they are not involved in 
step~\ref{multistring:steptwo:full} any longer. 
Let us consider all the maximal 
consecutive groups of them. Let one such group be  
$\mathcal{F}_{i_g},\mathcal{F}_{i_g+1},\dots,\mathcal{F}_{i_g+z}$, for 
some $\mathcal{F}_{i_g}$. We know that 
$\mathcal{F}_{i_g}$ is associated with just one rank $r_{i_{g'}}$, 
$g'\ge g$ and 
the only string in the corresponding subproblem is the one of rank $r_{i_{g'}}$ in $\mathcal{Z}$.
Analogously, $\mathcal{F}_{i_g+1}$ is associated with just one rank 
$r_{i_{g'}+1}$ and so forth up to $\mathcal{F}_{i_g+z-1}$ included. 
Hence, for any such  $\mathcal{F}_{i_g+w}$, we have that
$\card{\mathcal{F}_{i_g+w}}+\card{\mathcal{M}_{i_g+w}}=\Delta_{i_{g'}+w}$.
Therefore, because of Lemma~\ref{lem:multisel:multiset:full}, the term 
$-c\Delta_{i_{g'}+w} \log \Delta_{i_{g'}+w}$ of the wanted bound is 
accounted for. 
What we have said so far does not hold for $\mathcal{F}_{i_g+z}$ (unless it is 
the rightmost one of all the pivotal multisets $\mathcal{F}_{*}$). Let 
us deal with these cases later.

Let us now consider any pivotal multiset $\mathcal{F}_o$ created in 
the first execution of step~\ref{multistring:steptwo:full} and containing more 
that one symbol. Let $\ssubpro{z_{o}}$ the subproblem corresponding to  
$\mathcal{F}_o$. Let $r_p,r_{p+1},\ldots,r_x$ be the ranks corresponding to 
$\mathcal{F}_o$ and let $N_o$ be the 
number of symbols in $\mathcal{M}$ less than the ones in 
$\mathcal{F}_o$. Since $\ssubpro{z_{o}}$ 
has offset size less than the one of $\subpro{0}$, by inductive hypothesis 
we know that the algorithm retrieves the strings of $\mathcal{Z}$ with 
ranks  $r_p,r_{p+1},\ldots,r_x$ using time 

$$c\card{\mathcal{F}_o}\log\card{\mathcal{F}_o} 
+c\card{\mathcal{F}_o}+c5L'_o -c\sum_{j=p}^{x-1} 
\Delta_j \log \Delta_j+$$
$$-c(r_p-N_o)\log(r_p-N_o) 
-c(\card{\mathcal{F}_o}-r_x+N_o+1)\log(\card{\mathcal{F}_o}-r_x+N_o+1)$$

\noindent{}where $L'_o+\card{\ssubpro{z_{o}}}$ is the smallest number of symbols we 
need to probe to find the wanted strings amongst the ones in 
$\ssubpro{z_{o}}$ (since $\soffset{z_{o}}=2$, i.e. the first 
symbol of each string in $\ssubpro{z_{o}}$ can be skipped). 
Let us call $-c(r_p-N_o)\log(r_p-N_o)$ the \emph{left inductive 
term} and the 
two terms $c\card{\mathcal{F}_o}$, $c5L'_o$ \emph{inductive remainder 
terms}. 
We will deal with them later. 

By Lemma~\ref{lem:multisel:multiset:full}, we have term 
$-c\left(\card{\mathcal{F}_o}+\card{\mathcal{M}_o}\right)\log 
\left(\card{\mathcal{F}_o}+\card{\mathcal{M}_o}\right)$ from the first 
execution of step~\ref{multistring:steptwo:a:full}. 
We can upper bound that term by the two-term expression 
$-c\card{\mathcal{F}_o}\log \card{\mathcal{F}_o}-c\card{\mathcal{M}_o}\log \card{\mathcal{M}_o}$. 
The first term $-c\card{\mathcal{F}_o}\log \card{\mathcal{F}_o}$ can 
be cancelled with the opposite term in the expression given by the inductive 
hypothesis.

For each rank $r_j \in\set{r_p,r_{p+1},\ldots,r_{x-1}}$   
the term $-c\Delta_j \log \Delta_j$ of the wanted upper bound is 
accounted for, since it is in 
the expression given by the inductive 
hypothesis. Amongst the 
ranks corresponding to $\card{\mathcal{F}_o}$, $r_x$ is the one with 
which we 
have yet to deal.  

Let us first assume that 
$\card{\mathcal{F}_{o+1}}=1$. If that is true, we know that 
$N_*+\card{\mathcal{M}_o}=\Delta_x$, where we denoted 
with $N_*$ the number of strings in $\ssubpro{z_{o}}$ with rank (in 
$\mathcal{Z}$) greater 
than or equal to $r_x$. But $N_*=\card{\mathcal{F}_o}-r_x+N_o+1$. 
Hence, to obtain the term $-c\Delta_x\log\Delta_x$ for the wanted 
bound, we need to mix the terms 
$-c(\card{\mathcal{F}_o}-r_x+N_o+1)\log(\card{\mathcal{F}_o}-r_x+N_o+1)$
(from the inductive hypothesis on $\ssubpro{z_{o}}$) and $-c\card{\mathcal{M}_o}\log \card{\mathcal{M}_o}$
(from the first execution of step~\ref{multistring:steptwo:a:full}). 

It is easy to prove that for any $a,b\ge 1$, 
$(a+b)\log(a+b)\ge a\log a+b\log b +2(a+b)$. Thus, we have that  

$$-cN_*\log N_*-c\card{\mathcal{M}_o}\log \card{\mathcal{M}_o}\le$$ 
$$-c(N_*+\card{\mathcal{M}_o})\log(N_*+\card{\mathcal{M}_o})+
2c(N_*+\card{\mathcal{M}_o})=-c\Delta_x\log\Delta_x+2c\Delta_x.$$

\noindent{}Thus, we have now obtained the $-c\Delta_x\log\Delta_x$ term for the last 
rank $r_x$ of $\mathcal{F}_o$ for the case 
$\card{\mathcal{F}_{o+1}}=1$. However, we still have to deal with the 
extra $2c\Delta_x$, let us call it the \emph{extra remainder 
term}.

Let us now consider the case 
$\card{\mathcal{F}_{o+1}}>1$. Thus we can use the inductive hypothesis 
and in this case 
$\mathcal{F}_{o+1}$ contributes to the bound a \emph{left inductive
term} $-c(r_p'-N_o')\log(r_p'-N_o')$ (analogous to the 
$-c(r_p-N_o)\log(r_p-N_o)$ for $\mathcal{F}_{o}$),
 where $r_p'$ is the smallest rank associated with 
$\mathcal{F}_{o+1}$ and $N_o'$ is the number of symbols in 
$\mathcal{M}$ less than the ones in $\mathcal{F}_{o+1}$.
The term $-c\Delta_x\log\Delta_x$ is obtained in the same 
way we did before, only this time we have to mix three intervals 
instead of two, since in this case 
$\Delta_x=N_*+\card{\mathcal{M}_o}+(r_p'-N_o')$. Hence, 
in this case the \emph{extra remainder term} is $4c\Delta_x$. 

Before we account for all the remainder terms, we still need to deal 
with the rightmost multiset of each maximal 
consecutive group of those multisets $\mathcal{F}_*$ such that 
$\card{\mathcal{F}_*}=1$. The (single) ranks $r_u$ of each one of these rightmost 
multisets are the only ones for which we have not yet obtained the 
term $-c\Delta_u\log\Delta_u$ of the wanted upper bound. For any such multiset 
$\mathcal{F}_i$ with rank $r_i'$ ($i\le i'$),  
$-c\Delta_i'\log\Delta_i'$ can be obtained in the exact same way we 
did above and hence we have an extra remainder term for each one of 
these ranks too.

Finally, let us account for the remainder terms. For each 
$\mathcal{F}_o$, we have at most three kinds 
of remainders: $c\card{\mathcal{F}_o}$, $c5L'_o$ and a third kind that 
is at most $4c\Delta_x$, where $r_x$ is the largest rank associated with 
$\mathcal{F}_o$ and $L'_o+\card{\ssubpro{z_{o}}}$ is the smallest number of symbols we 
need to probe to find the wanted strings amongst the ones in 
$\ssubpro{z_{o}}$ (if any). Overall, the first and third kinds
are upper bounded by $c5N$. By the definition of the second kind of 
remainder and since all the $N$ symbols in $\mathcal{M}$ 
need to be probed to find the wanted strings, we have that 
$c5N+\sum_{\set{\mathcal{F}_o\left\vert\card{\mathcal{F}_o}>1}\right.} 
c5L'_o=c5L'$.

\end{proof}

\subsection{Proof of Lemma~\ref{lem:slice:comp}}
\label{subsec:slice:comp:full}

Let us consider the pruning phase. In the first step, the level-by-level visit of the 
skip tree of $A$ can be easily done in $\Oh{\acard{A}}$: the skip tree 
contains only contact and branching subproblems, hence its has 
$\Oh{\acard{A}}$ nodes; also, by Lemma~\ref{lem:grouping:full}, all the 
calls to \grouping{} have a total cost which is linear in the number 
of columns of $A$, which is $\acard{A}$. About the second step, the 
cost of \skipvisit{} on the skip tree is clearly $\Oh{\acard{A}}$.  
So is the cost of the third step  where we do an $\Oh{1}$ amount of 
work for each node in $\pruned$. 
In the fourth step we do $\Oh{1}$ amount of work for each subproblem 
that is not in a subtree rooted at some $\subpro{w_{i}}\in\pruned$. 
By the definition of $\pruned$ we already know that each one of the 
subproblems we touch in the fourth step will be later partitioned into two 
or more smaller subproblems. Thus they are certainly less than 
$\card{\bigcup_{i=1}^{d}A_{i}-A}$ (which is the total number of new 
subproblems that are created by the slicing of $A$). Thus the pruning phase takes 
$\Oh{\acard{A}+\card{\bigcup_{i=1}^{d}A_{i}-A}}$ time.

Let us consider the slicing phase. First of all, \slicerec{} performs 
a depth first visit of the tree of $A$ that whenever encounters a 
node in $\pruned$ it does not go any deeper. Thus, the total number of 
nodes visited is equal to the number of subproblems that are partitioned
into two or more subproblems by the slice operation. 

Let us consider the internal nodes of $A$. For any such $\subpro{r}$, in the 
first step, after the 
recursive calls to \slicerec{} have returned, we call \grouping{} on a 
list containing  
$\card{\mathcal{C}_{r}}+\sum_{i=1}^{f}\mbox{\it{}sub}_{r_{i}}$ objects, 
where $\mathcal{C}_{r}$ is the set of contact suffixes of $\subpro{r}$ 
(if it is a contact node) and $\mbox{\it{}sub}_{r_{i}}$ is the number 
of subproblems into which the $i$-th child of $\subpro{r}$ has been 
refined into (during the recursive call). We charge each 
$\mbox{\it{}sub}_{r_{i}}$ term to the corresponding child. The second 
step costs $\Oh{\mbox{\it{}sub}_{r}}$ time. The third step requires 
$\Oh{1}$ amount of work for each one of the children of each new 
subproblem into which $\subpro{r}$ has been partitioned: we charge any 
such $\Oh{1}$ amount of work the the corresponding child.
Thus, the total cost for an internal node $\subpro{r}$ of $A$ (including the costs 
charged to $\subpro{r}$ by its parent) is $\Oh{\mbox{\it{}sub}_{r}}$.

For each leaf $\subpro{r}$ of $A$, the total work we do is of the order of the number of columns
of $\subpro{r}$ (by Lemma~\ref{lem:grouping:full}). The amount charged to  
$\subpro{r}$ by its parent is of the same order. Thus, overall the 
cost for all the leaves of $A$, the internal nodes of $A$ and the new 
nodes of each $A_{i}$ is 
$\Oh{\acard{A}+\card{\bigcup_{i=1}^{d}A_{i}-A}}$.

Finally let us consider the finishing phase. In the first step we can 
use radix sorting and thus the first and second steps take 
$\Oh{\card{\bigcup_{i=1}^{d}A_{i}-A}}$ time. The third step accesses 
each skip node of each new agglomerate $\Oh{1}$ times, thus 
costing $\Oh{\acard{A}}$. The fourth step, accesses at most all the nodes 
that are not descendants of any node in $\pruned$. By 
using $\rankstrut$, each access to establish if a node is unsolved 
takes $\Oh{1}$ time. Thus the total cost of the fourth step is
$\Oh{\card{\bigcup_{i=1}^{d}A_{i}-A}}$.

\subsection{Proof of Theorem~\ref{theo:main:full}}

At this point, Theorem~\ref{theo:main:full} follows directly from the
following Lemmas~\ref{lem:msel:suffixes:comp:full}
and~\ref{lem:rapN:full}, whose proofs detail some of the ideas
presented in Section~\ref{sec:correctness-analysis}.

\begin{lemma}
\label{lem:msel:suffixes:comp:full}
The running time of the suffix multi-selection algorithm for a text of
length $N$ is upper bounded by $O\left(N\log N - \sum_{j=0}^K \Delta_j
  \log \Delta_j+N+\rapbpart\right)$, where $\rapbpart$ is the total 
time required by all the \RAP{}s minus the time for the \mselmset{} 
calls.
\end{lemma}
\begin{proof}
The cost of the initialization stage is dominated by the call to 
\mselmset{} (building \rankstrut{} takes $\Oh{N}$ time). 
By Lemma~\ref{lem:multisel:multiset:full}, we know that the cost 
of that call is within our target bound.
After the Refine and Aggregate stage, the cost of the finalization 
stage (Section~\ref{sub:finalization-stage:full}) is
$\Oh{N}$.

Let us now account for the contribution of the calls to 
\mselmset{} to the total cost. Instead of the normal time cost, let us consider 
the \emph{virtual symbol cost} of the calls to \mselmset{}: a call to 
\mselmset{} costs or (conceptually) \emph{creates} $x$ virtual symbols 
if the multiset that it receives in 
input has $x$ objects. The virtual symbols created by the call to 
\mselmset{} for some leading subproblem are exclusively created and 
``used'' for that subproblem. Thus, unlike $T$'s symbols, 
virtual symbols are not shared by subproblems. Naturally, 
virtual symbols form virtual strings: each subproblem $\subpro{i}$ has 
associated $\card{\subpro{i}}$ \emph{virtual strings} that are made up of all 
the virtual symbols that will be created during the computation 
to 
refine $\subpro{i}$ every time it is the leading subproblem of its 
current agglomerate. And, unlike the suffixes of $T$, all the 
virtual strings are \emph{independent}. 

Picking 
an unsolved agglomerate $A$ and refining it with a \RAP{} can be seen as 
a two-part process:
$(a)$ picking an 
unsolved subproblem, the leading subproblem of $A$, and refining it with the 
call to \mselmset{}; $(b)$ refining all the other unsolved subproblems of $A$
with the rest of the \RAP{} (mainly the call to \sliceop{}). 
Thus, if we take aside all the $(b)$ parts of the \RAP{}s and 
if we consider the subproblems to be subsets of virtual strings, 
then the suffix multi-selection 
algorithm behaves exactly like the multi-selection algorithm for 
independent strings in Section~\ref{subsec:msel:strings:full}.
Therefore, by Lemma~\ref{lem:msel:strings:comp:full}, we have that the 
cost of the algorithm is 
$ \Oh{N\log N - \sum_{j=0}^K \Delta_j \log 
\Delta_j+N+L+\rapbpart} 
$,
where $r_0 = 0$, $r_{K+1} = N+1$, $\Delta_j \equiv r_{j+1} - r_j$ for
$0 \leq j \leq K$, $L$ is the total number of virtual symbols (and 
$\rapbpart$ is the total cost of the $(b)$ parts of the \RAP{})s.

Let us evaluate $L$. As we have seen above, a \RAP{} an agglomerate $A$ creates
a number of new virtual symbols that is equal to the cardinality of the multiset
passed to \mselmset{} (equal to $\acard{A}$). Since it is not computed during
\mselmset{} call, the cardinality of such multiset must be $\Oh{\rapbpart(A)}$
(where $\rapbpart(A)$ denotes the total cost of the \RAP{} on $A$ minus the
cost of the \mselmset{} call). Therefore $L=\Oh{\rapbpart}$.
\end{proof}

\begin{lemma}
  \label{lem:rapN:full}
  $\rapbpart = O(N)$.
\end{lemma}
\begin{proof}
Let us first evaluate the total cost of the steps of a \RAP{} on an 
agglomerate $A$ minus the cost of the call to \mselmset{} in the 
third step.

About the first step, to verify which kind of an agglomerate we are dealing 
with (generic or core cyclic) $\Oh{1}$ scans of the contact node list 
of $A$ (while using the \arrsub{} structure) are enough. 
Thus, the first step takes $\Oh{\acard{A}}$ time.

About the second step, the slightly more complex case is when $A$ is 
core cyclic. In that case we first find find the columns 
$\ppair{c_{j},r_{j}}$ such that $\suf{c_{j}-1}\in\subpro{i}\not\in A$, 
then their keys and from those we produce the keys for all the other 
columns. As we noticed, we do not actually access each suffix 
of $A$ to assign it its key. We do that only with $A$'s columns. 
Thus, the second step takes $\Oh{\acard{A}}$ time.

About the third step. Thanks to \rankstrut{}, retrieving the 
$\card{\subprorank{w}}$ ranks of the 
leading subproblem $\subpro{w}$ of $A$ takes 
$\Oh{\card{\subprorank{w}}}$ time. As we said, there is a 1-to-1 
correspondence between the leading suffixes and the root suffixes. 
Hence, retrieving the suffixes of 
$\subpro{w}$'s and their keys takes $\Oh{\card{\subpro{w}}}$.
We will deal 
with the cost of the call to \mselmset{} later.
Scanning the pivotal multisets and tagging the 
the columns after the call to \mselmset{} clearly cost $\Oh{\acard{A}}$.  
If we exclude \mselmset{}, the cost of the 
third step is $\Oh{\acard{A}}$ (since $\card{\subpro{w}}=\acard{A}$).
  
The cost of the fourth step is dominated by the cost of the call 
\sliceop{}$(A)$ which, by Lemma~\ref{lem:slice:comp}, is 
$\Oh{\acard{A}+\card{\bigcup_{i=1}^{q}A_{i}-A}}$, where 
$A_{1},\ldots,A_{q}$ are the agglomerates into which $A$ is sliced.
 
Let us consider the fifth step and, in particular, its three substeps 
described in Section~\ref{sub:undecided-details:full}. 
In the first substep, we call \leadvisit{} on the root of each 
agglomerate in \gundecided{}. \leadvisit{} accesses the nodes of 
the skip tree of the agglomerate it received in input and for each 
node does $\Oh{1}$ time worth of work. 
In the second substep we compute the prefix sum array $\visitsum_{i}$ of 
$\visitlist_{i}$ for each $A_{i}\in\gundecided$ and then we  
tag each column of $A_{i}$. So both substeps require 
$\Oh{\sum_{\set{A_{i}\in\gundecided}} \acard{A_{i}}}$ time.
In the third substep we call \sliceop{}$(A_{i})$ for each 
$A_{i}\in\gundecided$. All those calls clearly dominates the total 
cost of the substep. For each $A_{i}\in\gundecided$, let 
$A_{i_{1}},\ldots,A_{i_{t(i)}},A_{i_{t(i)+1}}$ be the $t(i)+1$ agglomerates into which $A_{i}$ 
is sliced (for 
simplicity's sake let us assume that there is a $A_{i_{t(i)+1}}$, i.e.   
an exhausted agglomerate, for each $A_{i}$). 
By  Lemma~\ref{lem:slice:comp}, the total cost of the third substep is 
$\Oh{\sum_{\set{A_{i}\in\gundecided}}\left(\acard{A_{i}}+\card{\bigcup_{j=1}^{t(i)+1}A_{i_{j}}-A_{i}}\right)}$.

Finally, in the sixth step we operate on each $A_{i}\in\gjoinable$. 
We have two cases in which we either do a \joinop{} or a \slicejoinop{} 
between $A_{i}$ and the agglomerate $A_{i*}$ 
with which $A_{i}$ is joinable. In the second case, let $A_{i*1}$ be the agglomerate 
``sliceable'' from $A_{i*}$ such that $\ppair{c_{j},r_{j}}$ 
is a column of $A_{i*1}$ iff $\suf{c_{j}-1}\in\subpro{z}\in A_{i}$.
By Lemmas~\ref{lem:join:comp} 
and~\ref{lem:slicejoin:comp} the total cost of the sixth
step is $\Oh{\sum_{\set{A_{i}\in\gjoinable}} \acard{A_{i}}}$ in the 
first case and
$\Oh{\sum_{\set{A_{i}\in\gjoinable}} \acard{A_{i}}+\card{A_{i*1}}}$ in 
the second. Let us denote with $\gsljoinable$ all the $A_{i}$'s in 
$\gjoinable$ that need a \slicejoinop{} in the eighth step.

Excluding the cost of the call to \mselmset{} in the 
third step and, since 
$\sum_{\set{A_{i}\in\gundecided}}\acard{A_{i}}+\sum_{\set{A_{i}\in\gjoinable}} \acard{A_{i}}\le\acard{A}$ 
(in the fourth step some of the $A_{i}$'s sliced from $A$ may have 
been moved to \gunsolved{}), the total time for the \RAP{}  is 
\begin{gather*}
\Oh{\acard{A}+\termone+\termtwo+\termthree}\\
\\
\termone=\card{\bigcup_{i=1}^{q}A_{i}-A}\;\;\;\;
\termtwo=\sum_{\set{A_{i}\in\gundecided}}\card{\bigcup_{j=1}^{t(i)+1}A_{i_{j}}-A_{i}}\;\;\;\;
\termthree=\sum_{\set{A_{i}\in\gsljoinable}}\card{A_{i*1}}
\end{gather*}

To complete the analysis of $\rapbpart$ let us introduce the \emph{events}.
During 
the algorithm five kinds of crucial events happen. $(a)$ \emph{Column 
fusions}: when during a \joinop{}$(A',A)$ (or at the end of a 
\slicejoinop{}) a column of $\ppair{c_{i},r_{i}}$ of $A'$ is fused with 
the column $\ppair{c_{j},r_{j}}$ of $A$ such that $r_{i}+1=c_{j}$, to 
form the column $\ppair{c_{i},r_{j}}$. 
$(b)$ \emph{Subproblem creations}: 
when during a \sliceop{} (or at the beginning of a 
\slicejoinop{})
some subproblem $\subpro{j}$ is partitioned into smaller subproblems 
$\subpro{i_{1}},\ldots,\subpro{i_{p}}$, we have \emph{$p$ subproblem 
creation events}.
$(c)$ \emph{Suffix exhaustions}: when during the third step of a \RAP{} for 
$A$, after the \mselmset{} call, some suffix 
$\suf{e}\in\subpro{w}\in A$, where $\subpro{w}$ is the leading 
subproblem of $A$, becomes exhausted (i.e. 
$\suf{e}$ belongs to a pivotal multiset $\mathcal{M}_{i}$, thus we know 
for sure that it is not one of the wanted suffixes).
$(d)$ \emph{Suffix discoveries}: when during the third step of a 
\RAP{} for $A$, some suffix $\suf{w}$ in $\subpro{w}$, the leading 
subproblem of $A$, is recognized as one of the wanted suffixes (i.e. 
$\suf{w}$ belongs to a pivotal multiset $\mathcal{F}_{i}$ such that 
$\card{\mathcal{F}_{i}}=1$).
$(e)$ \emph{Inner collisions}: when during the second step of a \RAP{} 
for $A$, we encounter a column $\ppair{c_{i},r_{i}}$ of $A$ such that 
$\ppair{c_{j},r_{j}}$, where $c_{j}=r_{i}+1$, is also in $A$ 
while $\suf{c_{i}}$ and $\suf{c_{j}}$ do not belong to the same 
subproblem.

A column is never divided into 
smaller ones and a subproblem is never merged with others 
to form a larger subproblem. Also, when a suffix becomes exhausted or a wanted suffix is 
discovered, they can never be in an unsolved 
subproblem again. Finally, the same inner collision cannot be repeated 
after the \RAP{} in which it has been detected has ended, since the 
two columns colliding will not be part of the same agglomerate any 
longer. By all the above, it is easy to see that one particular 
event cannot be repeated twice. Since there are $N$ suffixes and we start with $N$ 
columns and one subproblem, we can conclude that the total number of events 
during the computation is $\Oh{N}$, or $\le 6N$ to be precise. 

Let us establish how many events take place \emph{during} a generic 
\RAP{} for an agglomerate $A$. Let us start with \emph{subproblem 
creations}. 
In the fourth step, for each agglomerate $A_{i}$ sliced from 
$A$, we have $\card{A_{i}-A}$ subproblem creations, one for 
each new (i.e. not coming from $A$) subproblem in $A_{i}$. 
Analogously, in the third substep of the fifth step
(Section~\ref{sub:undecided-details:full}), 
for each $A_{i}\in\gundecided$ and 
for each $A_{i_{j}}$ sliced from $A_{i}$, we have 
$\card{A_{i_{j}}-A_{i}}$ subproblem creations. Finally, in the eighth 
step, for each $A_{i}\in\gsljoinable$, we have $\card{A_{i*1}}$ 
subproblem creations. Summing up,
$\termone+\termtwo+\termthree$ subproblem creations take place during the \RAP{} 
for an agglomerate $A$.
The total number of \emph{suffix discoveries} and \emph{suffix 
exhaustions} occurring during the \RAP{} on $A$ is equal to the total 
number of columns of the agglomerates that in the fourth step are 
moved either to $\gexhausted$ or to $\gundecided$. Each column of each 
agglomerate moved to $\gundecided$ corresponds to a \emph{suffix 
exhaustion}. Each column of each agglomerate moved to $\gexhausted$  
corresponds to either a \emph{suffix discovery} or a \emph{suffix 
exhaustion}.
Analogously, the numbers of \emph{column fusions} and \emph{inner collisions} 
during the \RAP{} of $A$  are equal to the total numbers of 
columns of the $A_{i}$'s that in the fourth step are moved to 
$\gjoinable$ and $\gunsolved$, respectively. 
Since $\sum_{i=1}^{q}\acard{A_{i}}=\acard{A}$, we can conclude that 
the \emph{total number of events} occurring during the \RAP{} for $A$ 
is $\acard{A}+\termone+\termtwo+\termthree$. As we have seen, if we 
exclude the calls to \mselmset{}, the cost of the \RAP{} for $A$ is 
$\Oh{\acard{A}+\termone+\termtwo+\termthree}$. Therefore the total 
time required by all the \RAP{}s minus the time for the \mselmset{} 
calls is of the order of the total number of events, which is $\Oh{N}$.
\end{proof}

\subsection{Proof of Theorem~\ref{the:bwt2}}

Theorem~\ref{the:bwt2} follows directly from 
Lemmas~\ref{lem:multisel:contiguous:full} and~\ref{lem:lcp:full}.

\begin{lemma}\label{lem:multisel:contiguous:full}
Given a text $T$ of $N$ symbols drawn from an unbounded alphabet and 
$K$ consecutive ranks $r_1,\ldots,r_K$ (i.e. 
$r_2=r_1+1,r_3=r_2+1,\ldots, r_K=r_{K-1}+1$),
the $K$ text suffixes of ranks $r_1,\ldots,r_K$ can be found 
using  
$\Oh{K\log K+N}$ time and comparisons.

\end{lemma}
\begin{proof}
To retrieve  the wanted suffixes in $\Oh{K\log K +N}$ time we first 
apply the suffix multi-selection algorithm on $T$ with only $r_1$ and 
$r_K$. Then we go through an \emph{intermediate stage} that takes the 
subproblems left by the suffix multi-selection and prepares them for a 
second suffix multi-selection. After that we apply 
again the suffix multi-selection on $T$ but this time $(a)$ we use all the 
ranks $r_1,\ldots,r_K$ and $(b)$ we skip the \emph{initialization 
stage} (because of the work done in the intermediate stage). 

Let us give the details of the process.

We execute the suffix multi-selection algorithm on $T$ with 
$\mathcal{R}=\set{r_1,r_K}$. After the 
computation ends, we have 
the two wanted suffixes, let them be 
$\suf{l}$ and $\suf{r}$ ($\suf{l}<\suf{r}$), the exhausted agglomerates 
and a subproblem list 
$\sublist{}=\seq{\subpro{0},\ldots,\subpro{i_l},\ldots,\subpro{i_r},\ldots,\subpro{p}}$, ($i_l<i_r<p$), 
with the following properties. 
$(a)$ $\subpro{i_l}=\set{\suf{l}}$ and
$\subpro{i_r}=\set{\suf{r}}$ (they are the only solved subproblems). $(b)$ 
$\subpro{i}<\subpro{i_l}$, 
$\subpro{i_l}<\subpro{i'}<\subpro{i_r}$ and 
$\subpro{i_r}<\subpro{i''}$, for each $i <i_l$, $i_l<i'<i_r$ and 
$i''>i_r$, respectively. $(c)$ All together the subproblems 
$\subpro{i'}$ with $i_l<i'<i_r$ contains \emph{exactly} $K-2$ suffixes and 
they are the ones with ranks $r_2,\ldots,r_{K-1}$ (but for each 
of those suffixes we do not know which $r_2,\ldots,r_{K-1}$ is its rank). 

The \emph{intermediate stage} has the following steps.

First, for each subproblem in $\sublist{}$ we retrieve its suffixes. 
Recall that for each agglomerate $A_i$ only contact and root suffixes are explicitly 
stored during the computation. To retrieve all the suffixes of each 
$\subpro{j}\in A_i$ we simply need to visit the tree of $A_i$ from its 
root with \suffixvisit{}$(\subpro{r})$ defined as follows. $(i)$ If 
$\subpro{r}$ is a leaf then
$\subpro{r}=\set{\suf{c}\;\vert\; \ppair{c,r} \mbox{ is in $\subpro{r}$'s column 
list}}$. After we retrieved the suffixes of $\subpro{r}$ 
we return the set $\set{\suf{i+1}\;\vert\; \suf{i}\in\subpro{r}}$. 
$(ii)$ Otherwise, if $\subpro{r}$ is a contact node (but not a leaf) we 
retrieve $Con=\set{\suf{c}\;\vert\; \ppair{c,r} \mbox{ is in $\subpro{r}$'s column 
list}}$. $(iii)$ In any case, we call \suffixvisit{}$(\subpro{r_i})$ for each children 
$\subpro{r_i}$ of $\subpro{r}$, let $\mathcal{S}$ be the set of 
all the suffixes we receive from all these recursive calls. $(iv)$ 
Then, we set $\subpro{r}=Con\cup\mathcal{S}$  and we return the set
$\set{\suf{i+1}\;\vert\; \suf{i}\in\subpro{r}}$.

Second, with a scan of $\sublist{}$, we do the following. We merge 
$\subpro{0}$ with all the subproblems in its neighborhood 
$\neigh{0}$, let them be $\subpro{1},\subpro{2},\ldots,\subpro{n_0}$, 
into one
(the subproblems in the same neighborhood are adjacent in 
$\sublist{}$). We do the same for
$\subpro{n_0+1}$ and its neighborhood 
$\subpro{n_0+2},\ldots,\subpro{n_1}$, then for $\subpro{n_1+1}$ and so 
forth until all the subproblems $\subpro{i}<\subpro{i_l}$ have been 
treated.  After that we do the same for all the subproblems 
$\subpro{''}>\subpro{i_r}$. Finally we merge all the subproblems 
$\subpro{i'}$ with $i_l<i'<i_r$ into one, let it be $\subpro{\#}$.
After the first step, the new 
$\sublist{}=\seq{\subpro{0},\ldots,\subpro{i_l},\subpro{\#},\subpro{i_r},\ldots,\subpro{p'}}$
maintains the same properties of the original one. However, the 
meaning of \emph{the integer labels} of the subproblems may have 
changed. 
Now for each subproblem $\subpro{i}\in\sublist$, $\slab{i}$ is the 
number of suffixes of $T$ that are lexicographically smaller than each 
$\suf{j}\in\subpro{i}$ (because now for each $\subpro{i}$ we have that 
$\neigh{i}=\subpro{i}$, whereas originally that was guaranteed only 
for $\subpro{i_l}$ and $\subpro{i_r}$).

Third, for each $\suf{i}\in\subpro{\#}$, let $\suf{i}$'s \emph{key} be its 
first symbol $T[i]$. We call 
\mselmset{}$\left(\subpro{\#},r_2,\ldots,r_{K-1}\right)$ and we obtain 
$\subpro{\#}$'s pivotal subsets 
$\mathcal{M}_0,\mathcal{F}_1,\mathcal{M}_1,\ldots,\mathcal{F}_{t},\mathcal{M}_{t}$. 
Since we know that $\subpro{\#}$ contains all and only the $K-2$ 
suffixes with ranks $r_2,\ldots,r_{K-1}$, all the $\mathcal{M}_i$ 
\emph{are void}. Thus, from each $\mathcal{F}_i$ we create a subproblem 
$\subpro{\#i}$ with the following properties: $(a)$ 
$\card{\subpro{\#i}}=\card{\subprorank{\#i}}$ and $(b)$ $\slab{\#i}$ is the 
number of suffixes of $T$ smaller than each one in $\subpro{\#i}$.
After this step we have 
$\sublist{}=\seq{\subpro{0},\ldots,\subpro{i_l},\subpro{\#1},\ldots,\subpro{\#t},\subpro{i_r},\ldots,\subpro{p'}}$.

Fourth, from each $\subpro{i}$ with $i<i_l$ or $i>i_r$ (they are all 
exhausted) we make an 
agglomerate $A_{i}$ and we move it to $\gexhausted$. We do the same 
for $\subpro{i_l}$ and $\subpro{i_r}$ (although, as subproblems, they 
are solved ones). Finally, from each $\subpro{\#j}$ we make an agglomerate 
$A_{\#i}$ and we move it to $\gunsolved$.

After the intermediate stage, we apply again the suffix 
multi-selection algorithm on $T$ with the full rank set 
$\mathcal{R}=\set{r_1,\ldots,r_K}$ in the following way. We skip the 
initialization stage completely and we start immediately with the refine and 
aggregate stage (using $\sublist$, $\gunsolved$, $\gexhausted$, the agglomerates
and the subproblems we already have after the intermediate stage). The 
rest of the computation proceeds normally.

Let us now evaluate the cost of the whole computation.
The first call of the suffix 
multi-selection algorithm is done with just two ranks, $r_1$ and $r_K$. 
Thus, by Theorem~\ref{theo:main:full}, its cost is $\Oh{N}$.
Let us consider the intermediate stage. 
The first step requires $\Oh{N}$ time: for each agglomerate $A_i$, \suffixvisit{}
costs $\Oh{suf_i}$ where $suf_i=\sum_{\set{\subpro{j}\in A_i}} 
\card{\subpro{j}}$. 
For the second step a scan of 
$\sublist$ is enough and the cost is $\Oh{N}$. The cost of the third 
step si dominated by the cost of the call to \mselmset{}. Since it is
done on a multiset of $K-2$ elements (and with a set 
of $K-2$ ranks), its cost is clearly $\Oh{K\log K}$.
Finally, an $\Oh{N}$ time scan of $\sublist$ is enough for the fourth 
step. 

Let us now consider the cost of the second  execution of the 
suffix multi-selection. The complexity proof is the same of the one 
for Theorem~\ref{theo:main:full} except for two aspects. First, 
there is no initialization stage, thus the cost of finding 
the pivotal multisets of $\set{T[1],\ldots,T[N]}$ disappears. Second, 
the total contribution of the \mselmset{} calls made during 
the whole refine and aggregate stage changes as follows. The total number of 
suffixes of all the leading subproblems is at most $K-2$. Thus the total 
number of virtual strings is at most $K-2$. On the other hand, the events that 
take place during all the \RAP{}s are still $\Oh{N}$. Thus, the total 
number of virtual symbols in the virtual strings is $\Oh{N}$.
All the other additional costs of the \RAP{}s remain the same.  
Therefore, by Lemma~\ref{lem:msel:strings:comp:full}, we have that the 
total cost of the second call of the suffix multi-selection is is 
$\Oh{K\log K +N}$.
\end{proof}

For any two subproblems $\subpro{i}$ and $\subpro{i'}$ such that 
$\neigh{i}\not=\neigh{i'}$, let 
$\lcp{\subpro{i},\subpro{i'}}$ be the length of the longest common prefix of any two 
suffixes $\suf{j_i}\in\subpro{i}$ and $\suf{j_{i'}}\in\subpro{i'}$. 
Let us extend the notation to neighborhoods: for any two 
$\neigh{i}\not=\neigh{i'}$, $\lcp{\neigh{i},\neigh{i'}}$ is equal to
$\lcp{\subpro{i},\subpro{i'}}$.
We have the following:
\begin{lemma}\label{lem:lcp:full}
The suffix multi-selection algorithm can be modified (without changing 
its asymptotical time complexity) so that 
$\lcp{\subpro{i},\subpro{i'}}$ can be computed in $\Oh{1}$ time, 
for any two $\subpro{i},\subpro{i'}$ such that 
$\neigh{i}\not=\neigh{i'}$. 
\end{lemma}
\begin{proof}
As we have seen, for any subproblem $\subpro{i}$, the subproblems in its 
neighborhood $\neigh{i}$ are in contiguous positions in $\sublist$.
We maintain another list $\neighsublist$ whose elements represent the 
neighborhoods: the $i$-th element in $\neighsublist$ represents the $i$-th
contiguous group of subproblems in $\sublist$ forming a neighborhood.
Each $\subpro{i}$ in $\sublist$ has a link to the element in 
$\neighsublist$ corresponding to its neighborhood.
We also maintain a third list $\lcpsublist$ such that 
if $\subpro{i}$ and $\subpro{i'}$ are in two adjacent neighborhoods
$\neighsublist[j]$ and $\neighsublist[j+1]$ then 
$\lcpsublist[j]=\lcp{\subpro{i},\subpro{i'}}$. 
A cartesian tree is maintained on $\lcpsublist$ and a structure for
dynamic LCA queries  (e.g. \cite{Cole:2005}) is maintained on the 
cartesian tree.

$\sublist$, $\neighsublist$ and  
$\lcpsublist$ are updated in the finishing phase of the \sliceop{} 
operation (see Section~\ref{subsubsec:slice:finishing:full}). As we have 
seen, a subproblem $\subpro{i}$ in $\sublist$ is replaced by a 
partitioning $\subpro{i_1},\ldots,\subpro{i_r}$ of it. If $\subpro{i}$ 
was active then the partitioning is necessarily a refining one, otherwise we do 
not care. Thus, if $\subpro{i}$ was inactive, the set of the suffixes 
whose subproblem belong to $\neigh{i}$ does not change and neither 
$\neighsublist$ nor $\lcpsublist$ needs to be updated.

On the other hand, if $\subpro{i}$ was active then $\neigh{i}=\subpro{i}$ 
and  $\neigh{i_j}=\subpro{i_j}$, for each $1\le j\le r$. Thus, the 
element in $\neighsublist$ for $\neigh{i}$ is replaced by 
the ones for $\neigh{i_1},\ldots,\neigh{i_r}$. Since the partitioning 
is a refining one, we have that 
$\lcp{\neigh{p},\neigh{i_1}}=\lcp{\neigh{p},\neigh{i}}$ and 
$\lcp{\neigh{i_r},\neigh{s}}=\lcp{\neigh{i},\neigh{s}}$, where 
$\neigh{p}$ and $\neigh{s}$ are (were) the predecessor and successor 
of $\subpro{i}$ in $\neighsublist$. Thus the corresponding entries in 
$\lcpsublist$ do not need to be updated (and neither does the cartesian 
tree nor the structure for LCA queries). The values 
$\lcp{\neigh{i_1},\neigh{i_2}},\ldots,\lcp{\neigh{i_{r-1}},\neigh{i_r}}$ 
can be easily found by doing $r-1$ LCA queries. Since $\subpro{i}$ has 
been refined, we know that none of the values 
$\lcp{\neigh{i_1},\neigh{i_2}},\ldots,\lcp{\neigh{i_{r-1}},\neigh{i_r}}$ 
can be smaller than either $\lcp{\neigh{p},\neigh{i}}$ or 
$\lcp{\neigh{i},\neigh{s}}$. Thus each pair of insertions of $\neigh{i_j}$ in 
$\neighsublist$ and of $\lcp{\neigh{i_j},\neigh{i_{j+1}}}$ in 
$\lcpsublist$ corresponds to the insertion of a leaf on the 
cartesian tree built on $\lcpsublist$. This kind of updates of the 
structure for LCA queries can be done in $\Oh{1}$ time (hence adding 
an extra $\Oh{N}$ term to the complexity bound for the suffix multi-selection 
algorithm).
\end{proof}
  
\section{Conclusions}
\label{sec:conclusions:full}

We studied partial compression and text indexing problems, and as a
technical piece, the suffix multi-selection problem. The main theme is that 
when comparing an arbitrary set of suffixes which might overlap in sophisticated
ways, we need to devise methods to avoid rescanning characters to get optimal
results. We achieve this with a variety of structural observations and carefully 
arranging computations, and achieve bounds optimal with respect to those known
for atomic elements. Other partial suffix problems  will be of great interest
in Stringology and its applications. 
\junk{
the asymptotically optimal complexity of the
multi-selection problem on suffixes, describing a way to explictly
represent and handle the inherent complex patterns that arise among
the suffixes when performing the multi-selection task on them. It is
an open problem to count the exact number of comparisons needed to
solve the multi-selection problem on suffixes. It would be interesting
to see if it is possible to simplify our algorithm and get exact
bounds as in~\cite{Kanela} for our problem too.
}
  

\begin{thebibliography}{10}

\bibitem{Blelloch:1996:CSA}
G.~E. Blelloch, C.~E. Leiserson, B.~M. Maggs, C.~G. Plaxton,
  S.~J. Smith, and M. Zagha.
\newblock A comparison of sorting algorithms for the {Connection Machine CM-2}.
\newblock {\em C.ACM}, 39:273--297, 1996.

\bibitem{JCSS::BlumFPRT1973}
M. Blum, R.~W. Floyd, V. Pratt, R.~L. Rivest, and R.~E. Tarjan.
\newblock Time bounds for selection.
\newblock {\em Journal of Computer and System Sciences}, 7(4):448--461, August
  1973.

\bibitem{Burrows:1994:BSL}
M.~Burrows and D.~J. Wheeler.
\newblock A block-sorting lossless data compression algorithm.
\newblock Research Report 124, Digital SRC, Palo Alto, CA, USA, May 1994.

\bibitem{Chambers71}
J.~M. Chambers.
\newblock Partial sorting (algorithm 410).
\newblock {\em Commun. ACM}, 14(5):357--358, 1971.

\bibitem{Cole:2005}
Richard Cole and Ramesh Hariharan.
\newblock Dynamic {LCA} queries on trees.
\newblock {\em SIAM Journal on Computing}, 34(4):894--923, August 2005.

\bibitem{Cunto:1989:ACS}
W. Cunto and J.~I. Munro.
\newblock Average case selection.
\newblock {\em Journal of the ACM}, 36(2):270--279, April 1989.

\bibitem{dobkin:munro}
D.~Dobkin and I.~Munro.
\newblock Optimal time minimal space selection algorithms.
\newblock {\em Journal of the ACM}, 28(3):454--461, July 1981.

\bibitem{F}
M.~Farach.
\newblock Optimal suffix tree construction with large alphabets.
\newblock In {\em Proceedings of the 38th Annual Symposium on Foundations of
  Computer Science ({FOCS})}, pages 137--143. IEEE Computer Society Press,
  1997.

\bibitem{FarachFerraginaMuthu}
Martin Farach-Colton, Paolo Ferragina and S. Muthukrishnan.
\newblock On the sorting-complexity of suffix tree construction
\newblock {\em J. ACM}, 47(6):987--1011, 2000.

\bibitem{Floyd:1975:ETB}
R.~W. Floyd and R.~L. Rivest.
\newblock Expected time bounds for selection.
\newblock {\em C.ACM}, 18(3):165--172, 1975.

\bibitem{Franceschini_et_al09}
G.~Franceschini, R.~Grossi, and S.~Muthukrishnan.
\newblock Optimal cache-aware suffix selection.
\newblock In {\em STACS}, volume~3, pages 457--468, 2009.

\bibitem{FranceschiniMuthu07}
G.~Franceschini and S.~Muthukrishnan.
\newblock Optimal suffix selection.
\newblock In {\em STOC}, pages 328--337. ACM, 2007.

\bibitem{Hoare:1961:AF}
C.~A.~R. Hoare.
\newblock {Algorithm 65}: {Find}.
\newblock {\em Communications of the ACM}, 4(7):321--322, July 1961.

\bibitem{HwangT02}
H.-K. Hwang and T.-H. Tsai.
\newblock Quickselect and the Dickman function.
\newblock {\em Comb.,Prob.\&Comp.}, 11(4), 2002.

\bibitem{Kanela}
K. Kaligosi, K. Mehlhorn, J.~I. Munro, and P. Sanders.
\newblock Towards optimal multiple selection.
\newblock {\em ICALP} 2005, {\em LNCS} volume 3580, 103--114, 2005.


\bibitem{KarkkainenSB06}
J. K{\"a}rkk{\"a}inen, P. Sanders, and S. Burkhardt.
\newblock Linear work suffix array construction.
\newblock {\em J.ACM}, 53(6):918--936, 2006.

\bibitem{Esko}
J.~K{\"{a}}rkk{\"{a}}inen and E.~Ukkonen.
\newblock Sparse suffix trees.
\newblock {\em LNCS}, 1090:219, 1996.

\bibitem{knuth3}
D.E. Knuth.
\newblock The Art of Computer Programming, vol. 3.
\newblock {\em Addison-Wesley}, 1998.


\bibitem{Kuba06}
M. Kuba.
\newblock On quickselect, partial sorting and multiple quickselect.
\newblock {\em IPL}, 99(5):181--186, 2006.

\bibitem{MahmoudMS95}
H.~M. Mahmoud, R. Modarres, and R.~T. Smythe.
\newblock Analysis of quickselect: An algorithm for order statistics.
\newblock {\em Informatique th\'eorique et applications}, 29(4), 1995.

\bibitem{Manber93}
Udi Manber and Gene Myers.
\newblock Suffix arrays: {A} new method for on-line string searches.
\newblock {\em SIAM Journal on Computing}, 22(5):935--948, October 1993.

\bibitem{Manzini01}
Giovanni Manzini.
\newblock An analysis of the {B}urrows-{W}heeler transform.
\newblock {\em J. ACM}, 48(3):407--430, 2001.

\bibitem{MartinezPV04}
C. Martinez, D. Panario, and A. Viola.
\newblock Adaptive sampling for quickselect.
\newblock {\em SODA} 2004, 447--455, 2004.

\bibitem{McC}
E.~M. McCreight.
\newblock A space-economical suffix tree construction algorithm.
\newblock {\em J.ACM}, 23(2):262--272, 1976.

\bibitem{MR}
\newblock J. Ian Munro and Venkatesh Raman. 
\newblock Sorting multisets and vectors in-place. 
{\em WADS} 1991: 473-480.

\bibitem{Panholzer:2003:AMQ}
A. Panholzer.
\newblock Analysis of multiple quickselect variants.
\newblock {\em TCS}, 302(1--3):45--91, 2003.

\bibitem{Pohl:1972:SPC}
I. Pohl.
\newblock A sorting problem and its complexity.
\newblock {\em C.ACM}, 15(6):462--464, June 1972.

\bibitem{Prodinger:1995:MQH}
H. Prodinger.
\newblock {Multiple Quickselect}---{Hoare}'s {Find} algorithm for several
  elements.
\newblock {\em IPL}, 56:123, 1995.

\bibitem{Wei}
P.~Weiner.
\newblock Linear pattern matching algorithms.
\newblock In {\em FOCS'73}, 1--11, 1973.

\end{thebibliography}

\begin{figure}[!h]
  \begin{center}
    \begin{psfrags}
          \psfrag{\$}{\figtextfont\texttt{\tend}}
      \psfrag{a}{\figtextfont\texttt{a}}
      \psfrag{b}{\figtextfont\texttt{b}}
      \psfrag{c}{\figtextfont\texttt{c}}
      \psfrag{d}{\figtextfont\texttt{d}}
      \psfrag{e}{\figtextfont\texttt{e}}
      \psfrag{f}{\figtextfont\texttt{f}} 
      \psfrag{g}{\figtextfont\texttt{g}}      
      \psfrag{h}{\figtextfont\texttt{h}}      
      \psfrag{i}{\figtextfont\texttt{i}}      
      \psfrag{j}{\figtextfont\texttt{j}}      
      \psfrag{k}{\figtextfont\texttt{k}}      
      \psfrag{l}{\figtextfont\texttt{l}}      
      \psfrag{m}{\figtextfont\texttt{m}}      
      \psfrag{n}{\figtextfont\texttt{n}}      
      \psfrag{o}{\figtextfont\texttt{o}}      
      \psfrag{p}{\figtextfont\texttt{p}}      
      \psfrag{q}{\figtextfont\texttt{q}}  
      \psfrag{r}{\figtextfont\texttt{r}}
      \psfrag{s}{\figtextfont\texttt{s}}
      \psfrag{t}{\figtextfont\texttt{t}}
      \psfrag{u}{\figtextfont\texttt{u}}
      \psfrag{v}{\figtextfont\texttt{v}}
      \psfrag{w}{\figtextfont\texttt{w}}
      \psfrag{x}{\figtextfont\texttt{x}}
      \psfrag{y}{\figtextfont\texttt{y}}
      \psfrag{z}{\figtextfont\texttt{z}} 
      \psfrag{1}{\figtextfont{}$\suf{1}$}
      \psfrag{5}{\figtextfont{}$\suf{5}$}
      \psfrag{10}{\figtextfont{}$\suf{10}$}
      \psfrag{15}{\figtextfont{}$\suf{15}$}
      \psfrag{20}{\figtextfont{}$\suf{20}$}
      \psfrag{25}{\figtextfont{}$\suf{25}$}
      \psfrag{30}{\figtextfont{}$\suf{30}$}
      \psfrag{35}{\figtextfont{}$\suf{35}$}
      \psfrag{40}{\figtextfont{}$\suf{40}$}
      \psfrag{45}{\figtextfont{}$\suf{45}$}
      \psfrag{50}{\figtextfont{}$\suf{50}$}
      \psfrag{55}{\figtextfont{}$\suf{55}$}
      \psfrag{60}{\figtextfont{}$\suf{60}$}
      \psfrag{61}{\figtextfont{}$\suf{61}$}
      \psfrag{65}{\figtextfont{}$\suf{65}$}
      \psfrag{70}{\figtextfont{}$\suf{70}$}
      \psfrag{75}{\figtextfont{}$\suf{75}$}
      \psfrag{80}{\figtextfont{}$\suf{80}$}
      \psfrag{85}{\figtextfont{}$\suf{85}$}
      \psfrag{90}{\figtextfont{}$\suf{90}$}
      \psfrag{95}{\figtextfont{}$\suf{95}$}
      \psfrag{100}{\figtextfont{}$\suf{100}$}
      \psfrag{105}{\figtextfont{}$\suf{105}$}
      \psfrag{110}{\figtextfont{}$\suf{110}$}
      \psfrag{115}{\figtextfont{}$\suf{115}$}
      \psfrag{120}{\figtextfont{}$\suf{120}$}
      \psfrag{121}{\figtextfont{}$\suf{121}$}
      \psfrag{125}{\figtextfont{}$\suf{125}$}
      \psfrag{130}{\figtextfont{}$\suf{130}$}
      \psfrag{135}{\figtextfont{}$\suf{135}$}
      \psfrag{140}{\figtextfont{}$\suf{140}$}
      \psfrag{145}{\figtextfont{}$\suf{145}$}
      \psfrag{150}{\figtextfont{}$\suf{150}$}
      \psfrag{155}{\figtextfont{}$\suf{155}$}
      \psfrag{160}{\figtextfont{}$\suf{160}$}
      \psfrag{165}{\figtextfont{}$\suf{165}$}
      \psfrag{170}{\figtextfont{}$\suf{170}$}
      \psfrag{175}{\figtextfont{}$\suf{175}$}
      \psfrag{180}{\figtextfont{}$\suf{180}$}
      \psfrag{c0}{\figtextfont{}$C_{0}$}
      \psfrag{c1}{\figtextfont{}$C_{1}$}
      \psfrag{c2}{\figtextfont{}$C_{2}$}
      \psfrag{c3}{\figtextfont{}$C_{3}$}
      \psfrag{c4}{\figtextfont{}$C_{4}$}
      \psfrag{c5}{\figtextfont{}$C_{5}$}
      \psfrag{c6}{\figtextfont{}$C_{6}$}
      \psfrag{c7}{\figtextfont{}$C_{7}$}
      \psfrag{c8}{\figtextfont{}$C_{8}$}
      \psfrag{c9}{\figtextfont{}$C_{9}$}
      \psfrag{c10}{\figtextfont{}$C_{10}$}
      \psfrag{c11}{\figtextfont{}$C_{11}$}
      \psfrag{c12}{\figtextfont{}$C_{12}$}
      \psfrag{c13}{\figtextfont{}$C_{13}$}
      \psfrag{c14}{\figtextfont{}$C_{14}$}
      \psfrag{c15}{\figtextfont{}$C_{15}$}
      \psfrag{c16}{\figtextfont{}$C_{16}$}
      \psfrag{c17}{\figtextfont{}$C_{17}$}
      \psfrag{c18}{\figtextfont{}$C_{18}$}
      \psfrag{c19}{\figtextfont{}$C_{19}$}
      \psfrag{c20}{\figtextfont{}$C_{20}$}
      \psfrag{c21}{\figtextfont{}$C_{21}$}
      \psfrag{c22}{\figtextfont{}$C_{22}$}
      \psfrag{c23}{\figtextfont{}$C_{23}$}
      \psfrag{c24}{\figtextfont{}$C_{24}$}
      \psfrag{c25}{\figtextfont{}$C_{25}$}
      \psfrag{c26}{\figtextfont{}$C_{26}$}
      \psfrag{c27}{\figtextfont{}$C_{27}$}
      \psfrag{c28}{\figtextfont{}$C_{28}$}
      \psfrag{c29}{\figtextfont{}$C_{29}$}
      \psfrag{c30}{\figtextfont{}$C_{30}$}
      \psfrag{c31}{\figtextfont{}$C_{31}$}
      \psfrag{c32}{\figtextfont{}$C_{32}$}
      \psfrag{c33}{\figtextfont{}$C_{33}$}
      \psfrag{c34}{\figtextfont{}$C_{34}$}
      \psfrag{c35}{\figtextfont{}$C_{35}$}
      \psfrag{c36}{\figtextfont{}$C_{36}$}
      \psfrag{c37}{\figtextfont{}$C_{37}$}
      \psfrag{c38}{\figtextfont{}$C_{38}$}
      \psfrag{c39}{\figtextfont{}$C_{39}$}
      \psfrag{p0}{\figtextfont{}$\subpro{0}$}
      \psfrag{p1}{\figtextfont{}$\subpro{1}$}
      \psfrag{p2}{\figtextfont{}$\subpro{2}$}
      \psfrag{p3}{\figtextfont{}$\subpro{3}$}
      \psfrag{p4}{\figtextfont{}$\subpro{4}$}
      \psfrag{p5}{\figtextfont{}$\subpro{5}$}
      \psfrag{p6}{\figtextfont{}$\subpro{6}$}
      \psfrag{p7}{\figtextfont{}$\subpro{7}$}
      \psfrag{p8}{\figtextfont{}$\subpro{8}$}
      \psfrag{p9}{\figtextfont{}$\subpro{9}$}
      \psfrag{p10}{\figtextfont{}$\subpro{10}$}
      \psfrag{p11}{\figtextfont{}$\subpro{11}$}
      \psfrag{p12}{\figtextfont{}$\subpro{12}$}
      \psfrag{p13}{\figtextfont{}$\subpro{13}$}
      \psfrag{p14}{\figtextfont{}$\subpro{14}$}
      \psfrag{p15}{\figtextfont{}$\subpro{15}$}
      \psfrag{p16}{\figtextfont{}$\subpro{16}$}
      \psfrag{p17}{\figtextfont{}$\subpro{17}$}
      \psfrag{p18}{\figtextfont{}$\subpro{18}$}
      \psfrag{p19}{\figtextfont{}$\subpro{19}$}
      \psfrag{p20}{\figtextfont{}$\subpro{20}$}
      \psfrag{p21}{\figtextfont{}$\subpro{21}$}
      \psfrag{p22}{\figtextfont{}$\subpro{22}$}
      \psfrag{p23}{\figtextfont{}$\subpro{23}$}
      \psfrag{p24}{\figtextfont{}$\subpro{24}$}
      \psfrag{p25}{\figtextfont{}$\subpro{25}$}
      \psfrag{p26}{\figtextfont{}$\subpro{26}$}
      \psfrag{p27}{\figtextfont{}$\subpro{27}$}
      \psfrag{p28}{\figtextfont{}$\subpro{28}$}
      \psfrag{p29}{\figtextfont{}$\subpro{29}$}
      \psfrag{p30}{\figtextfont{}$\subpro{30}$}
      \psfrag{p31}{\figtextfont{}$\subpro{31}$}
      \psfrag{p32}{\figtextfont{}$\subpro{32}$}
      \psfrag{p33}{\figtextfont{}$\subpro{33}$}
      \psfrag{ex}{\figtextfont{}$e$}
      \psfrag{so}{\figtextfont{}$s$}
      \psfrag{un}{\figtextfont{}$u$}
      \psfrag{lead}{\figtextfont{}Leading sub.}
      \psfrag{cont}{\figtextfont{}Contact sub.}
      \psfrag{skip}{\figtextfont{}Skip link}
      \psfrag{guide}{\figtextfont{}Guide link}
      \psfrag{unso}{\figtextfont{}Unsolved sub.}
      \psfrag{exa}{\figtextfont{}Exhausted sub.}
      \psfrag{solv}{\figtextfont{}Solved sub.}
      \psfrag{a1}{$A_{1}$}
      \psfrag{a2}{$A_{2}$}
      \psfrag{a3}{$A_{3}$}
      \psfrag{a4}{$A_{4}$}
      \psfrag{a5}{$A_{5}$}
      \psfrag{a\$}{$A_0$}
      \psfrag{rs}{\figtextfont{}Rank sets}
      \psfrag{rs1}{\figtextfont{}$\subprorank{2}=\set{20,22,26}$} 
      \psfrag{rs2}{\figtextfont{}$\subprorank{3}=\set{42,49}$}
      \psfrag{rs3}{\figtextfont{}$\subprorank{6}=\set{68}$}
      \psfrag{rs4}{\figtextfont{}$\subprorank{7}=\set{93,99}$}
      \psfrag{rs5}{\figtextfont{}$\subprorank{9}=\set{77}$}
      \psfrag{rs6}{\figtextfont{}$\subprorank{10}=\set{103,105}$}
      \psfrag{rs7}{\figtextfont{}$\subprorank{13}=\set{125}$}
      \psfrag{rs8}{\figtextfont{}$\subprorank{14}=\set{81,83}$}
      \psfrag{rs9}{\figtextfont{}$\subprorank{15}=\set{116}$}
      \psfrag{rs10}{\figtextfont{}$\subprorank{16}=\set{128}$}
      \psfrag{rs11}{\figtextfont{}$\subprorank{17}=\set{121,122}$}
      \psfrag{rs12}{\figtextfont{}$\subprorank{18}=\set{130,137}$}
      \psfrag{rs13}{\figtextfont{}$\subprorank{19}=\set{139}$}
      \psfrag{rs14}{\figtextfont{}$\subprorank{20}=\set{157}$}
      \psfrag{rs15}{\figtextfont{}$\subprorank{23}=\set{153}$}
      \psfrag{rs16}{\figtextfont{}$\subprorank{24}=\set{65}$}
      \psfrag{rs17}{\figtextfont{}$\subprorank{25}=\set{158}$}
      \psfrag{rs18}{\figtextfont{}$\subprorank{26}=\set{163,168,173}$}
      \psfrag{rs19}{\figtextfont{}$\subprorank{27}=\set{175}$}
      \psfrag{rs20}{\figtextfont{}$\subprorank{28}=\set{179}$}
      \psfrag{rs21}{\figtextfont{}$\subprorank{29}=\set{36}$}
      \psfrag{rs22}{\figtextfont{}$\subprorank{31}=\set{162}$}
      \psfrag{rs23}{\figtextfont{}$\subprorank{32}=\set{57}$}
      \psfrag{rs24}{\figtextfont{}$\subprorank{33}=\set{150}$}
      \psfrag{rc}{\begin{minipage}{0.15\textwidth}\figtextfont{}Lists of columns of contact nodes
                  \end{minipage}}
      \psfrag{rc0}{\figtextfont{}$\subpro{0}\rightarrow\seq{C_{0}}$}            
      \psfrag{rc1}{\figtextfont{}$\subpro{3}\rightarrow\seq{C_{16},C_{17}}$}
      \psfrag{rc2}{\figtextfont{}$\subpro{6}\rightarrow\seq{C_{6}}$} 
      \psfrag{rc3}{\figtextfont{}$\subpro{9}\rightarrow\seq{C_{4},C_{5}}$}      
      \psfrag{rc4}{\figtextfont{}$\subpro{13}\rightarrow\seq{C_{13},C_{14},C_{15}}$}
      \psfrag{rc5}{\figtextfont{}$\subpro{14}\rightarrow\seq{C_{1},C_{2},C_{3}}$}
      \psfrag{rc6}{\figtextfont{}$\subpro{16}\rightarrow\seq{C_{11},C_{12}}$}
      \psfrag{rc7}{\figtextfont{}$\subpro{17}\rightarrow\seq{C_{7},C_{8},C_{9},C_{10}}$}
      \psfrag{rc8}{\figtextfont{}$\subpro{19}\rightarrow\seq{C_{21},C_{22},C_{23}}$}
      \psfrag{rc9}{\figtextfont{}$\subpro{20}\rightarrow\seq{C_{24},C_{25}}$}
      \psfrag{rc10}{\figtextfont{}$\subpro{22}\rightarrow\seq{C_{38},C_{39}}$}
      \psfrag{rc11}{\figtextfont{}$\subpro{23}\rightarrow\seq{C_{18},C_{19},C_{20}}$}
      \psfrag{rc12}{\figtextfont{}$\subpro{24}\rightarrow\seq{C_{33},C_{34}}$}
      \psfrag{rc13}{\figtextfont{}$\subpro{25}\rightarrow\seq{C_{37}}$}
      \psfrag{rc14}{\figtextfont{}$\subpro{27}\rightarrow\seq{C_{29}}$}
      \psfrag{rc15}{\figtextfont{}$\subpro{28}\rightarrow\seq{C_{32}}$}
      \psfrag{rc16}{\figtextfont{}$\subpro{29}\rightarrow\seq{C_{36}}$}
      \psfrag{rc17}{\figtextfont{}$\subpro{31}\rightarrow\seq{C_{30},C_{31}}$}
      \psfrag{rc18}{\figtextfont{}$\subpro{33}\rightarrow\seq{C_{26},C_{27},C_{28}}$}
      \psfrag{sl}{\figtextfont{}\sublist{}$=\langle\subpro{0},\subpro{1},\subpro{2},\subpro{29},
      \subpro{3},\subpro{32},\subpro{4},\subpro{24},\subpro{6},\subpro{8},\subpro{9},
      \subpro{11},\subpro{14},\subpro{5},\subpro{7},\subpro{10},\subpro{12},$}
      \psfrag{sl2}{\figtextfont{}$\subpro{15},\subpro{17},\subpro{13},
      \subpro{16},\subpro{18},\subpro{19},
      \subpro{30},\subpro{21},\subpro{33},\subpro{23},\subpro{20},\subpro{22},
      \subpro{25},\subpro{31},\subpro{26},\subpro{27},\subpro{28}\rangle$}
      \psfrag{rf}{\figtextfont{}Subproblems' fields}
      \psfrag{rf0}{\figtextfont{}$\card{\subpro{0}}=1$, $\slab{0}=0$}
      \psfrag{rf1}{\figtextfont{}$\card{\subpro{1}}=17$, $\slab{1}=1$}
      \psfrag{rf2}{\figtextfont{}$\card{\subpro{2}}=17$, $\slab{2}=18$}
      \psfrag{rf3}{\figtextfont{}$\card{\subpro{3}}=17$, $\slab{3}=39$}
      \psfrag{rf4}{\figtextfont{}$\card{\subpro{4}}=6$, $\slab{4}=59$}
      \psfrag{rf5}{\figtextfont{}$\card{\subpro{5}}=9$, $\slab{5}=84$}
      \psfrag{rf6}{\figtextfont{}$\card{\subpro{6}}=6$, $\slab{6}=67$}
      \psfrag{rf7}{\figtextfont{}$\card{\subpro{7}}=9$, $\slab{7}=93$}
      \psfrag{rf8}{\figtextfont{}$\card{\subpro{8}}=3$, $\slab{8}=73$}
      \psfrag{rf9}{\figtextfont{}$\card{\subpro{9}}=2$, $\slab{9}=76$}
      \psfrag{rf10}{\figtextfont{}$\card{\subpro{10}}=9$, $\slab{10}=102$}
      \psfrag{rf11}{\figtextfont{}$\card{\subpro{11}}=3$, $\slab{11}=78$}
      \psfrag{rf12}{\figtextfont{}$\card{\subpro{12}}=4$, $\slab{12}=111$}
      \psfrag{rf13}{\figtextfont{}$\card{\subpro{13}}=5$, $\slab{13}=123$}
      \psfrag{rf14}{\figtextfont{}$\card{\subpro{14}}=3$, $\slab{14}=81$}
      \psfrag{rf15}{\figtextfont{}$\card{\subpro{15}}=4$, $\slab{15}=115$}
      \psfrag{rf16}{\figtextfont{}$\card{\subpro{16}}=2$, $\slab{16}=128$}
      \psfrag{rf17}{\figtextfont{}$\card{\subpro{17}}=4$, $\slab{17}=119$}
      \psfrag{rf18}{\figtextfont{}$\card{\subpro{18}}=8$, $\slab{18}=130$}
      \psfrag{rf19}{\figtextfont{}$\card{\subpro{19}}=6$, $\slab{19}=138$}
      \psfrag{rf20}{\figtextfont{}$\card{\subpro{20}}=2$, $\slab{20}=156$}
      \psfrag{rf21}{\figtextfont{}$\card{\subpro{21}}=3$, $\slab{21}=147$}
      \psfrag{rf22}{\figtextfont{}$\card{\subpro{22}}=2$, $\slab{22}=158$}
      \psfrag{rf23}{\figtextfont{}$\card{\subpro{23}}=3$, $\slab{23}=153$}
      \psfrag{rf24}{\figtextfont{}$\card{\subpro{24}}=2$, $\slab{24}=65$}
      \psfrag{rf25}{\figtextfont{}$\card{\subpro{25}}=1$, $\slab{25}=160$}
      \psfrag{rf26}{\figtextfont{}$\card{\subpro{26}}=11$, $\slab{26}=163$}
      \psfrag{rf27}{\figtextfont{}$\card{\subpro{27}}=4$, $\slab{27}=174$}
      \psfrag{rf28}{\figtextfont{}$\card{\subpro{28}}=3$, $\slab{28}=178$}
      \psfrag{rf29}{\figtextfont{}$\card{\subpro{29}}=4$, $\slab{29}=35$}
      \psfrag{rf30}{\figtextfont{}$\card{\subpro{30}}=3$, $\slab{30}=144$}
      \psfrag{rf31}{\figtextfont{}$\card{\subpro{31}}=2$, $\slab{31}=161$}
      \psfrag{rf32}{\figtextfont{}$\card{\subpro{32}}=3$, $\slab{32}=56$}
      \psfrag{rf33}{\figtextfont{}$\card{\subpro{33}}=3$, $\slab{33}=150$}
      \psfrag{col1}{\figtextfont{}$\ppair{53,60}$}
      \psfrag{col2}{\figtextfont{}$\ppair{30,33}$}
      \psfrag{col3}{\figtextfont{}$\ppair{70,73}$}
      \includegraphics*[scale=0.7]{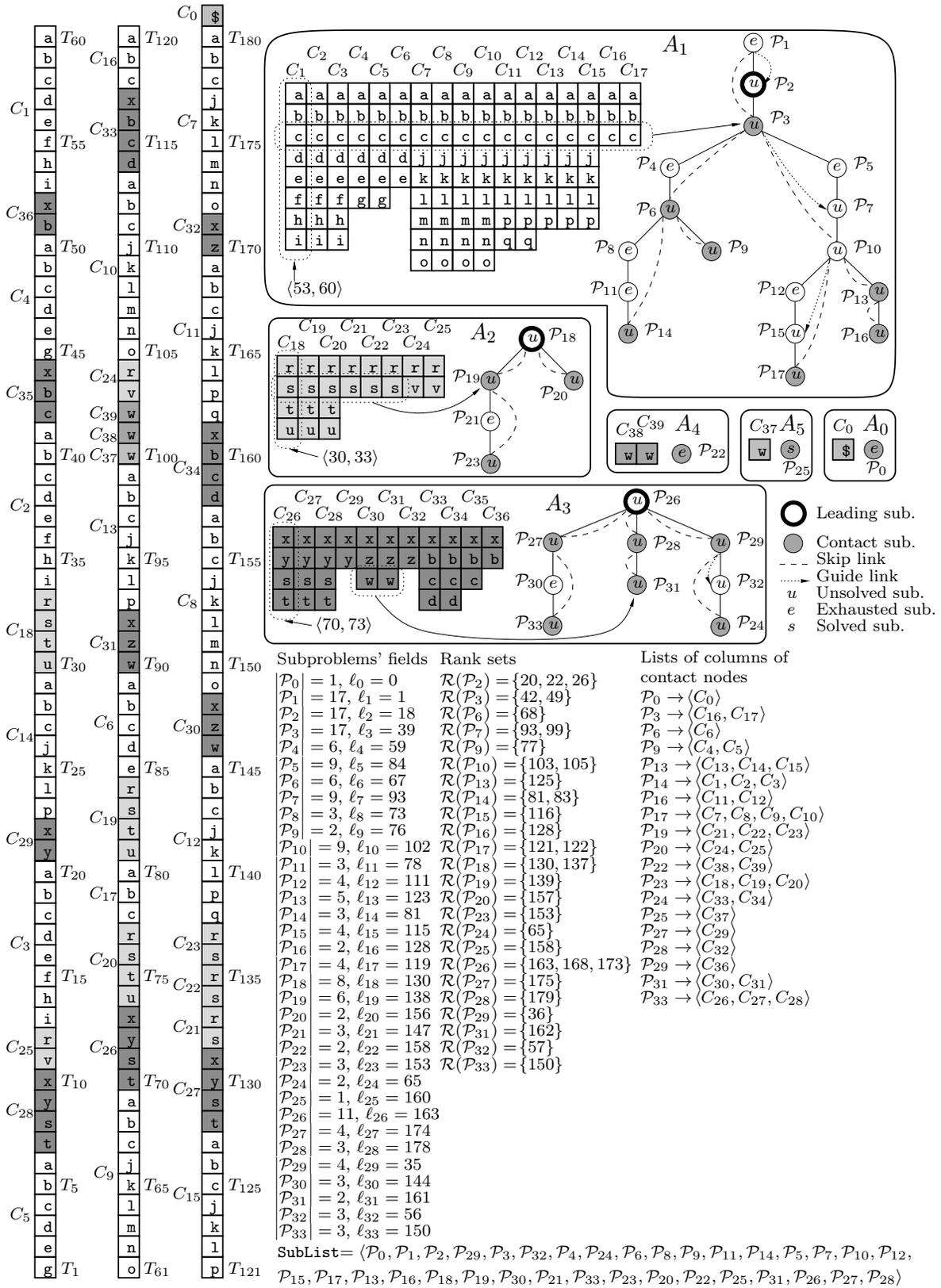}
  \end{psfrags}
  \end{center}
\caption{``Snapshot'' of the computation for a text $T$ with $N=181$ 
symbols and $K=34$ ranks in $\mathcal{R}$. Agglomerates $A_{1}$, $A_{2}$, and $A_{3}$ 
are unsolved, while $A_0$, $A_{4}$ and $A_{5}$ are exhausted. The columns of 
each agglomerate are pictured beside it (as contiguous substrings of 
$T$). $T$ is pictured as a partitioning of the columns of the 
agglomerates.}
\label{fig:big:full}
\end{figure}

\end{document}